\documentclass{article}
\IfFileExists{microtype.sty}{\usepackage{microtype}}{}
\usepackage{graphicx}
\usepackage{subcaption}
\usepackage{booktabs} %
\usepackage{etoc}
\AtBeginDocument{\etocdepthtag.toc{main}}
\usepackage{xspace}
\usepackage{hyperref}
\usepackage[preprint]{icml2026}
\makeatletter
\@starttoc{toc}
\makeatother
\usepackage{amsmath}
\usepackage{amssymb}
\usepackage{mathtools}
\usepackage{amsthm}

\usepackage[capitalize,noabbrev]{cleveref}
\graphicspath{{fig/}}
\usepackage[textsize=tiny]{todonotes}

\icmltitlerunning{Counterfactual Forecasting for Panel Data}

\newcommand{\converge}[1][]{\xrightarrow{#1}}

\newcommand{\asconverge}{\converge[\text{a.s.}]}
\newcommand{\dconverge}{\converge[d]}
\newcommand{\pconverge}{\converge[\pr]}
\newcommand{\normal}[2]{\mathcal{N}\left(#1,#2\right)}
\newcommand{\RR}{\mathbb{R}}	%
\newcommand{\NN}{\mathbb{N}}	%
\newcommand{\pr}{\mathbb{P}}	%
\newcommand{\EE}{\mathbb{E}}	%
\newcommand{\var}{\mathrm{Var}}	%
\newcommand{\cov}{\mathrm{Cov}}	%
\newcommand{\dnt}{\delta_{NT}}
\newcommand{\focus}{\textsc{Focus}\xspace}
\newcommand{\sforest}{\sigma^2_{i,T,h}}
\newcommand{\shatforest}{\hat\sigma^2_{i,T,h}}
\newcommand{\sest}{\xi^2_{i,T,h}}
\newcommand{\sfor}{\tau^2_{i,T,h}}
\newcommand{\shatest}{\hat\xi^2_{i,T,h}}
\newcommand{\shatfor}{\hat\tau^2_{i,T}}

\newcommand{\SigmaFobs}{\Sigma_F^{\mathrm{obs}}}
\newcommand{\SigmaFobshat}{\hat\Sigma_F^{\mathrm{obs}}}

\newcommand{\SigmaLobs}{\Sigma_{\Lambda,i}^{\mathrm{obs}}}
\newcommand{\SigmaLobshat}{\hat\Sigma_{\Lambda,i}^{\mathrm{obs}}}

\newcommand{\SigmaFmiss}{\Sigma_{F,T}^{\mathrm{miss}}}
\newcommand{\SigmaFmisshat}{\hat\Sigma_{F,T}^{\mathrm{miss}}}

\newcommand{\SigmaLmiss}{\Sigma_{\Lambda,i}^{\mathrm{miss}}}
\newcommand{\SigmaLmisshat}{\hat\Sigma_{\Lambda,i}^{\mathrm{miss}}}

\newcommand{\SigmaFLmisscov}{\Sigma_{F,\Lambda,T,i}^{\mathrm{miss, cov}}}

\newcommand{\SigmaFLmisscovhat}{\hat\Sigma_{F,\Lambda,T,i}^{\mathrm{miss, cov}}}

\newcommand{\Fktsq}{\overline{F}_{k:T}^2 }

\newcommand{\cQ}{\mathcal{Q}}
\newcommand{\dload}{\Delta_{\Lambda, i}}
\newcommand{\dfactor}{\Delta_{F, T}}
\newcommand{\dcoeff}{\Delta_{A,h}}
\newcommand{\vecop}[1]{\mathrm{vec}(#1)}

\newcommand{\target}[3]{\theta_{#1,#2 : #2 + #3}}
\newcommand{\forest}[3]{\hat{\theta}_{#1,#2 : #2 + #3}}

\newcommand{\thetaith}{\target{i}{T}{h}}
\newcommand{\thetahatith}{\forest{i}{T}{h}}

\newcommand{\pfend}{$\hfill\square$}

\newtheorem{theorem}{Theorem}[section]

\newcounter{assumpblock}
\newtheoremstyle{compact}
  {3pt} {3pt}
  {\itshape}
  {} {\bfseries} {.}
  {.5em} {}
\theoremstyle{compact}

\theoremstyle{compact}
\newtheorem{assumption}{Assumption}
\newcommand{\asn}{\text{Assum. }}

\newtheorem{lemma}{Lemma}[section]

\newtheorem{corollary}{Corollary}[section]

\theoremstyle{remark}
\newtheorem{remark}{Remark}[section]

\newcommand{\indicator}{\mbf 1}

\newcommand{\x}{x}

\newcommand{\axi}[1][i]{\x_{#1}}

\newcommand{\pseqxn}[1][n]{(\axi[i])_{i\geq 1}} %
\newcommand{\pseqxnn}[1][n]{(\axi[i])_{i=1}^n} %

\def\balign#1\ealign{\begin{align}#1\end{align}}
\def\baligns#1\ealigns{\begin{align*}#1\end{align*}}
\def\balignat#1\ealign{\begin{alignat}#1\end{alignat}}
\def\balignats#1\ealigns{\begin{alignat*}#1\end{alignat*}}
\def\bitemize#1\eitemize{\begin{itemize}#1\end{itemize}}
\def\benumerate#1\eenumerate{\begin{enumerate}#1\end{enumerate}}

\newenvironment{talign*}
 {\csname align*\endcsname}
 {\endalign}
\newenvironment{talign}
 {\csname align\endcsname}
 {\endalign}

\def\balignst#1\ealignst{\begin{talign*}#1\end{talign*}}
\def\balignt#1\ealignt{\begin{talign}#1\end{talign}}
\let\originalleft\left
\let\originalright\right
\renewcommand{\left}{\mathopen{}\mathclose\bgroup\originalleft}
\renewcommand{\right}{\aftergroup\egroup\originalright}

\def\tinycitep*#1{{\tiny\citep*{#1}}}
\def\tinycitealt*#1{{\tiny\citealt*{#1}}}
\def\tinycite*#1{{\tiny\cite*{#1}}}
\def\smallcitep*#1{{\scriptsize\citep*{#1}}}
\def\smallcitealt*#1{{\scriptsize\citealt*{#1}}}
\def\smallcite*#1{{\scriptsize\cite*{#1}}}

\def\mbf#1{\mathbf{#1}}
\def\mbb#1{\mathbb{#1}}

\def\<{\left\langle} %
\def\>{\right\rangle}

\def\defeq{\triangleq} %
\def\indic#1{\mbb{I}\left[{#1}\right]} %

\def\bigO#1{\mathcal{O}(#1)} %
\def\bigOP#1{\mathcal{O}_\pr\left(#1\right)} %
\def\littleOP#1{o_\pr(#1)} %

\def\indep{\perp\!\!\!\perp} %
\newcommand{\iid}{\mathrm{i.i.d.}}

\ifdefined\nonewproofenvironments\else
\ifdefined\ispres\else
\newenvironment{proof-sketch}{\noindent\textbf{Proof Sketch}
  \hspace*{1em}}{\qed\bigskip\\}
\newenvironment{proof-idea}{\noindent\textbf{Proof Idea}
  \hspace*{1em}}{\qed\bigskip\\}
\newenvironment{proof-of-lemma}[1][{}]{\noindent\textbf{Proof of Lemma {#1}}
  \hspace*{1em}}{\qed\\}
\newenvironment{proof-of-theorem}[1][{}]{\noindent\textbf{Proof of Theorem {#1}}
  \hspace*{1em}}{\qed\\}
\newenvironment{proof-attempt}{\noindent\textbf{Proof Attempt}
  \hspace*{1em}}{\qed\bigskip\\}

\begin{document}
\addtocontents{toc}{}

\twocolumn[
  \icmltitle{Counterfactual Forecasting for Panel Data}
  \icmlsetsymbol{equal}{*}

  \begin{icmlauthorlist}
    \icmlauthor{Navonil Deb}{1}
    \icmlauthor{Raaz Dwivedi}{2}
    \icmlauthor{Sumanta Basu}{1}
  \end{icmlauthorlist}

  \icmlaffiliation{1}{Department of Statistics and Data Science, Cornell University. Emails: \texttt{\{\href{mailto:nd329@cornell.edu}{nd329}, \href{mailto:sumbose@cornell.edu}{sumbose}\}@cornell.edu}}
  \icmlaffiliation{2}{Department of Operations Research and Information Engineering, Cornell Tech. Email: \href{mailto:dwivedi@cornell.edu}{\texttt{dwivedi@cornell.edu}}}

  \icmlcorrespondingauthor{Navonil Deb}{\href{mailto:nd329@cornell.edu}{\texttt{nd329@cornell.edu}}}
  \icmlkeywords{Machine Learning, ICML}

  \vskip 0.3in
]

\printAffiliationsAndNotice{}  %
\begin{abstract}
    We address the challenge of forecasting counterfactual outcomes in a panel data with missing entries and temporally dependent latent factors --- a common scenario in causal inference, where estimating unobserved potential outcomes ahead of time is essential. We propose Forecasting Counterfactuals under Stochastic Dynamics (\focus), a method that extends traditional matrix completion methods by leveraging time series dynamics of the factors, thereby enhancing the prediction accuracy of future counterfactuals. Building upon a consistent estimator of the factors, our method accommodates both stochastic and deterministic components within the factors, and provides a flexible framework for various applications. In case of stationary autoregressive factors and under standard conditions, we derive error bounds and establish asymptotic normality of our estimator. Empirical evaluations demonstrate that our method outperforms existing benchmarks when the latent factors have an autoregressive component. We illustrate \focus results on HeartSteps, a mobile health study, illustrating its effectiveness in forecasting step counts for users receiving activity prompts, thereby leveraging temporal patterns in user behavior.
\end{abstract}

\vspace{-15pt}
\section{Introduction}

Counterfactual estimation is a central challenge in panel data settings, with wide-ranging applications in personalized healthcare, economics, recommendation systems, and policy evaluation. An additional key task is forecasting counterfactual outcomes under hypothetical interventions even when the future interventions are not assigned yet -- especially relevant to disciplines that require prospective future decision-making. The difficulty particularly arises when the outcomes are noisy and temporally dependent across units. A powerful approach models outcome trajectories via low-dimensional latent factors that capture shared variation over time and units. This idea underpins prominent methods in causal inference, including synthetic controls \citep{abadie2010synthetic}, difference-in-differences \citep{arkhangelsky2021synthetic, goodman2021difference} and matrix completion methods \citep{agarwal2020synthetic, athey2021matrix, dwivedi2022doubly}, which rely on low-rank assumptions. Incorporating the temporal evolution of latent factors enables forecasting counterfactual outcomes beyond the observed panel.

To formalize the problem, we consider a potential outcome model with $N$ units and $T$ time points. For unit $i$ and time $t$, the potential outcome under treatment $w \in \{0,1\}$ is denoted by $Y_{i,t}(w)$ and modeled as
\begin{equation}\label{eq:counterfactual_model}
    Y_{i,t}(w) = \theta _{i,t}(w) + \varepsilon_{i,t},
\end{equation}
where $\theta_{i,t}(w)$ is called a \textit{common component}, and $\varepsilon_{i,t}$ is zero mean \textit{idiosyncratic noise}. We consider a linear factor model on the common components as follows:
\begin{equation}\label{eq:counterfactual_factor}
    \theta_{i,t}(w) = \Lambda_i^\top(w)F_t(w)
\end{equation}
where $F_t(w)$ are the underlying factors and $\Lambda_i(w)$ are the loadings, both having dimension $r \leq \min\{N,T\}$. Our goal is to forecast the mean potential outcome $\theta_{i,T+h}(w)$ for a given horizon $h$, under time series structure on $F_t(w)$. A special case of the factors dynamics illustrated in this work is a vector autoregressive model of order 1 (VAR(1)) \citep[Sec. 2.1.1]{lutkepohl2005new}
\begin{equation}\label{eq:counterfactual_var1}
    F_t(w) = A(w) F_{t-1}(w) + \eta_t(w),
\end{equation}
where $A_t(w)\in \RR^{r\times r}$ is the coefficient matrix of the VAR(1) process. $\eta_t(w)$ is a stationary noise process, independent across $t$, with mean $0$ and covariance matrix $\Sigma_{\eta}(w)$.

One motivation for counterfactual forecasting arises in policy-making and mobile health (mHealth) applications, where timely decision-making is critical. In the HeartSteps V1 study \citep{heartstepsv1, liao2020personalized}, users received five daily prompts encouraging walking activities. Walking behaviors of each user often exhibit a clear temporal structure across users, i.e. step counts in consecutive decision slots are negatively correlated with higher activity in one slot typically followed by lower activity in the next (see Fig. \ref{fig:jb30_time_plots}). When modeled with a linear factor model, the latent factors obtained from the steps under the intervention show strong correlations across some consecutive slots (Fig. \ref{fig:heartsteps_results}(a) and \ref{fig:hs_extra_plot}(a)-(c) in Appendix). Capturing this dependencies in the factor level enables accurate forecasting of future steps under intervention (or control), thereby facilitating prospective evaluation of the intervention’s effectiveness.

\subsection{Related works}

Recent advances address matrix completion and treatment effect estimation via low-rank factor models, however these works do not exploit temporal dynamics of the factors, and do not conduct forecasting \citep{bai2021matrix, jin2021factor, cahan2023factor, xiong2023large, agarwal2023causal, choi2024matrix}. \citet{goldin2022forecasting} propose SyNBEATS -- a neural network-based method with synthetic controls that performs post-treatment counterfactual estimation. However they do not provide theoretical guarantees, and do not leverage the latent dynamics to forecast out-of-sample. Another limitation of SyNBEATS is its reliance on a fully observed pre-treatment panel, which prevents it from addressing missing entries in control units before the intervention. \citet{pang2022bayesian} propose a Bayesian alternative to synthetic control methods and consider autoregressive factors in the model, however do not address the out-of-sample forecast performance and lacks theoretical guarantees. \citet{ben2023estimating, chen2023multi} incorporate multi-task Gaussian process to model the potential outcome dynamics, however do not conduct forecasting. \citet{agarwal2020multivariate, alomar2024samossa} address forecasting in panel data under missing observations with Multivariate singular spectrum analysis (mSSA) that addresses low rank assumption and deterministic dynamic factors with singular spectrum, but does not accommodate stochastic and serially correlated factors with non-singular spectra. Some related Neural CDE-based methods \citep{lim2018forecasting, seedat2022continuous, vanderschueren2023accounting} model counterfactual outcomes using high-capacity encoder-decoder architectures and continuous-time dynamics, which are computationally demanding, and require stronger treatment and overlap assumptions, and do not provide theoretical guarantees. Hence these methods are not well suited for small-panels and unit $\times$ horizon counterfactual forecasting settings driven by latent factors.

Several works show that smoothing the estimated factors improves forecast accuracy by exploiting their temporal structure \citep{doz2011two, poncela2021factor}. Often in context of factor-augmented VAR models \citep{bernanke2005measuring, lin2020regularized}, dynamic factor models use \textit{collapsing techniques} for capturing the latent factor dynamics \citep{jungbacker2008likelihood, brauning2014forecasting}. However, the existing forecast methods in dynamic factor models assume limited extent of missingness, i.e. the proportion of missing entries is typically finite and smaller than that in counterfactual estimation frameworks -- necessitating a counterfactual forecast method equipped to handle wide range of factor dynamics and observation patterns.

\subsection{Our contributions}
We introduce a method namely \focus (\textit{FOrecasting Counterfactuals Under Stochastic dynamics}) in Sec.\ref{sec:method}. To the best of our knowledge, this is the first work to address entry-by-entry forecasting counterfactuals out-of-sample. Equipped with a reliable factor estimation algorithm, \focus has two steps- (i) restricts the observed panel under treatment (or control) and estimates the latent factors by considering the control (or treated) observations as missing, (ii) fits an appropriate time series model on the estimated factors to construct a forecast estimator of the for unit $\times$ horizon level conditional mean counterfactual. On a general note, our method integrates time series smoothing and filtering techniques with matrix completion methods. Additionally, \focus is equipped to capture latent and stochastic temporal dependencies by accommodating serial correlation and non-singular spectra of the factors. Our algorithm is equipped to handle wide range of observation patterns, hence overcomes the limitation of restricted missingness of the forecasting methods in dynamic factor model. 

When the factors are estimated with the PCA algorithm of \citet{xiong2023large}, under a VAR(1) assumption on the true factors and standard assumptions on the loadings, the outcome noise and observation mechanism, we establish error bounds of the forecast (Thm. \ref{thm:forecast_error_bound}). Under additional regularities on the observation pattern, we establish $\min\{\sqrt{N},\sqrt{T}\}$-consistency and asymptotic normality of the forecast estimator (Thm. \ref{thm:asymptotic_normality}) and provide asymptotically valid confidence intervals (Cor. \ref{cor:ci}). Our key theoretical contribution is to extend the results of \citet{xiong2023large} from i.i.d. latent factors to temporally dependent factors with adjustments in the proof and by verifying that the dependent factor process satisfies the general CLT assumptions in Sec. \ref{sec:general_asn}. 

In several simulation settings, \focus yields more accurate forecast trajectories than our benchmarks mSSA and SyNBEATS, supporting the importance of leveraging the stochastic temporal latent dynamics. Additionally, \focus runs faster than mSSA and substantially faster than SyNBEATS under identical computational settings, with runtime growing scalably as $T$ increases, making \focus suitable for implementing in large-$T$ settings. In HeartSteps V1, we leverage the temporal association in the factors across consecutive suggestion slots. \focus, by capturing this temporal pattern, yields more accurate forecast of steps under intervention than mSSA.

\textbf{Organization.}~ We set up the problem under a simple setting and define our forecast estimand in Sec. \ref{sec:problem}. We illustrate our forecast algorithm \focus in Sec. \ref{sec:method}. The main results are stated in Sec. \ref{sec:theory}. We evaluate the performance of \focus on both simulated data and the real data set HeartSteps, as presented in Sec. \ref{sec:experiments}. The proofs, the general assumptions and the additional experiments are deferred in the supplement.

\textbf{Notation.}~ For $n\in \NN$, we denote $[n]={1,\dots,n}$. For $x\in \RR^n$, $\|x\|_2$ denotes the Eucledian $\ell_2$-norm of $x$. For any set $S$, $\sigma(S)$ denotes the smallest sigma-algebra containing $S$. For any event $E$, $\indic{E}=1$ if $E$ occurs, 0 otherwise. For a matrix $M$, $\|M\|=\sup_{\|x\|_2=1}\|Mx\|_2$; $\rho(M)$ is its spectral radius. We abbreviate almost surely to a.s., and independently and identically distributed to $\iid$. Convergence almost surely, in probability and distribution are denoted by $\xrightarrow{\mathrm{a.s.}}$, $\pconverge$ and $\dconverge$. For r.v.’s $X_n$ and sequence $a_n$, $X_n=\bigOP{a_n}$ means $X_n/a_n$ is bounded in probability, i.e. for every $\delta > 0$ there exists $c_\delta$ and $n_\delta$ such that $\pr\left( \frac{|X_n|}{a_n} > c_\delta \right) < \delta$ for all $n > n_\delta$. We use the terms \textit{unit} and \textit{time} interchangeably with \textit{row} and \textit{column} respectively.

\section{Problem set-up}\label{sec:problem}

We translate the counterfactual forecast problem into a forecast problem with missing entries in the panel. We build our forecast target based on the later.

\subsection{Modeling with dynamic latent factors}\label{subsec:linear_factor_model_missing_dynamics}

We denote the observed data by $(Y_{i,t}, W_{i,t}),~ (i,t)\in [N]\times[T]$, where $Y_{i,t}$ and $W_{i,t} \in \{0,1\}$ are the observed outcome and the binary treatment variable for unit $i$ and time $t$. When $W_{i,t} = 1$, we observe the treated outcome $Y_{i,t}(1)$, but not the control outcome $Y_{i,t}(0)$. Hence, $W_{i,t}$ indicates the observed entries in the potential outcomes $Y_{i,t}(1)$. Similarly, $1 - W_{i,t}$ is the observation indicator of $Y_{i,t}(0)$. We fix $w = 1$ and drop the notation $w$ from $Y_{i,t}(w), \theta_{i,t}(w), F_{i,t}(w)$ and $\Lambda_{i,t}(w)$. Restricting to the treatment panel with $W_{i,t}=1$, we rewrite \eqref{eq:counterfactual_model} with simplified notation as follows:
\begin{equation}\label{eq:factor_model_1}
    Y_{i,t} = \begin{cases}
        \theta_{i,t} + \varepsilon_{i,t} = \Lambda_i^\top F_t + \varepsilon_{i,t} & \text{ if } W_{i,t} = 1,\\
        \hspace{40pt}\star & \text{ if } W_{i,t} = 0.
    \end{cases}
\end{equation}
Here $\star$ denotes a missing entry. We impose a stationary vector autoregressive structure on the factors.
\vspace{5pt}
\begin{assumption}[Stationary VAR(1) factors]\label{asn:var1_factors}
    The latent factors $F_t$ follow a stationary $r$-dimensional VAR(1) model   \begin{equation}\label{eq:var1_factor}
    F_t = A F_{t-1} + \eta_t,    
    \end{equation}
    where $\rho(A)<1$ and $\eta_t$ is a noise process with mean 0 and covariance matrix $\Sigma_\eta$.
\end{assumption}
\vspace{5pt}

The VAR(1) structure is purely for the demonstration of our analysis, and can be generalized to more complex time series models. The coefficient matrix $A$ governs the structure of the autocorrelated factors. Additionally Wold's decomposition theorem \citet[Prop. 2.1]{lutkepohl2005new} guarantees that under mild regularity conditions, every stationary and purely nondeterministic process is well-approximated by a finite order VAR process. The condition $\rho(A) < 1$ ensures stationarity of the VAR(1) process \eqref{eq:var1_factor}.

\subsection{Forecast estimand}

The outcome mean at forecast horizon $h$ is $\theta_{i,T+h}$ that involves future realizations of the dynamic factor $F_{T+h}$. We define the forecast estimand as the mean future outcome implied by the latent factor dynamics. Although the factors are unobserved, all serial dependence in the panel operates through the factor process under independent outcome noise, while future factor noise $\eta_{T+h}$ and idiosyncratic noise $\varepsilon_{i,T+h}$ are mean-zero and unpredictable. Under the factor model \eqref{eq:factor_model_1}, the conditional mean of the future outcome given the latent factor history and loadings are
\begin{equation}\label{eq:forecast_target}
    \thetaith \defeq \Lambda_i^\top A^h F_T,
\end{equation}
which coincides with the minimum mean squared error linear predictor of $Y_{i,T+h}$ conditioned on the latent variables until time $T$. We refer to $\thetaith$ as our forecast estimand. In App. \ref{sec:identifiability}, we derive the expression in \eqref{eq:forecast_target} as the conditional mean given information until time $T$, and discuss more on the identifiability of $\thetaith$.

\section{Forecast method}\label{sec:method}

\focus addresses the task of estimating $\thetaith$ entry-by-entry by leveraging the dynamics of $F_t$. We illustrate the method under the simplified factor model \eqref{eq:factor_model_1} and Assum. \ref{asn:var1_factors}.

The inputs to \focus are the observed panel of outcome variables $Y \defeq (Y_{i,t})_{(i,t)\in [N]\times [T]}$ and the matrix of observation indicators $W \defeq (W_{i,t})_{(i,t)\in [N]\times [T]}$. For estimating the factors and the loadings, we use the PCA method of \citet{xiong2023large} interchanging the roles of unit and time. This method can be replaced by any algorithm that consistently estimates the factors and loadings under general observation pattern $W$. 

The steps of our proposed algorithm are the following.

\focus($Y, W$):
\vspace{-10pt}

\begin{enumerate}
    \item \label{item:step1} \textbf{Estimating the $\{F_t\}_{t\in [T]}$ and $\{\Lambda_i\}_{i \in [N]}$}:
    
    For each pair of columns $s, t\in [T]$, define the set
    \begin{equation}\label{eq:common_obs}
        \cQ_{s,t} \defeq \{i \in [N]: W_{i,s} = 1 \text{ and } W_{i,t} = 1\}.
    \end{equation}
    In words, $\cQ_{s,t}$ contains the rows for which both columns $s$ and $t$ are observed. We calculate the sample covariance matrix $\hat\Sigma$ for PCA with entries
    \begin{equation}\label{eq:sigma_hat}
    \hat{\Sigma}_{s,t} \defeq\frac{1}{|\cQ_{s,t}|} \sum_{i \in \cQ_{s,t}} Y_{i,t} Y_{i,s}, ~~~ |\cQ_{s,t}| > 0.
    \end{equation}
    Estimated factors are $\hat F \defeq \big[\hat F_1 :\ldots: \hat F_T\big]^\top$ where 
    \begin{equation}\label{eq:estimated_factors}
        \hat F = \sqrt{T} \times \text{First $r$ eigenvectors of } \hat\Sigma/T,
    \end{equation}
    Estimated loadings are $\hat \Lambda \defeq \big[\hat \Lambda_1 :\ldots: \hat \Lambda_N\big]^\top$ where
    \begin{equation}\label{eq:estimated_loadings}
        \hat \Lambda_i = \left(\sum_{t = 1}^T W_{i,t} \hat F_t \hat F_t^\top\right)^{-1}\left( \sum_{t = 1}^T W_{i,t} \hat F_t Y_{i,t} \right).
    \end{equation}
    \item \textbf{Forecasting with $\hat F$}:
    
    The ordinary least squares (OLS) estimator $\hat A$ of the VAR(1) coefficient matrix $A$ is calculated with $\{\hat F_t\}_{t\in [T]}$ as
    \begin{equation}\label{eq:A_hat}
        \hat A \defeq \left(\sum_{t=1}^{T-1} \hat F_{t+1}\hat F_t^\top\right) \left(\sum_{t=1}^{T-1} \hat F_t \hat F_t^\top\right)^{-1}.
    \end{equation}
    Then the plug-in estimator of $\theta_{i,T:T+h}$ is
    \begin{equation}\label{eq:forecast_target_est}
        \thetahatith = \hat \Lambda_i^\top \hat A^h \hat F_T.
    \end{equation}
\end{enumerate}

If the true factors $\{F_t\}_{t\in [T]}$ are observed, learning the autocorrelation of $F_t$'s provide an estimator of the best linear predictor of the future $F_{T+h}$ conditioned on their true past. \focus replaces $F_t$'s with their PCA estimators $\hat F_t$'s and the autocorrelation are learned from the estimates. The PCA algorithm of \citet{xiong2023large} provides consistent estimators $\hat \Lambda_i$'s and $\hat F_t$'s under a wide range of observation mechanism. We leverage their estimation approach in our forecasting method as a convenient baseline, and this modular step can be substituted by other consistent factor estimators without affecting the modeling of factor dynamics. When the latent process exhibits nonstationarity or a state-space structure, adapting a joint state-space posterior model \citep{doz2011two} is more appropriate.

In step \ref{item:step1} of \focus, we assume that $r$ i.e. the rank of the latent process is known to the algorithm. If $r$ is unknown, several works address the rank estimation methods \citep{bai2002determining, bai2019rank, choi2017selecting, wei2020determining} that can be applied here. Additionally $\hat \Sigma_{s,t}$ in \eqref{eq:sigma_hat} requires existence of at least one row with both columns $s$ and $t$ observed (see Rem. \ref{rem:observation} for more details). If $\cQ_{s,t}$ is empty, we can set $\hat \Sigma_{s,t} = 0$ and consistency of the estimated factors \citep{johnstone2009consistency} can still be ensured.

The algorithm \focus can be viewed as operating on a panel that has been denoised with respect to any observed covariates $X_{it}$. Covariates can be handled through methods such as interactive fixed effects \citep{bai2021matrix, choi2024matrix} and inverse propensity weighting of the entries of $\hat \Sigma$ \citep[Sec.~2]{xiong2023large}. Extensions incorporating covariates are left for future work.

\textbf{Forecasting control outcomes and treatment effects.}
We have used \focus to forecast future conditional mean outcomes from panels with partially observed entries. When both the treatment and control potential outcomes admit a latent factor structure as in \eqref{eq:counterfactual_model} and \eqref{eq:counterfactual_factor}, the two panels can be modeled separately under their respective latent dynamics. Conditioning on latent variables up to time $T$, we estimate the factor loadings, latent factors, and dynamic parameters $(\Lambda(w), F(w), A(w))$ for each treatment arm $w \in \{0,1\}$.

As discussed in Sec. \ref{subsec:linear_factor_model_missing_dynamics}, letting $W$ denote the treatment indicator, \focus estimates $\mathbb{E}[Y_{i,T+h}(w) \mid \Lambda(w), F(w)]$ by imposing factor dynamics on the panel $Y(w)$ and using $Ww + (1-W)(1-w)$ as the matrix of observation indicators. This allows us to forecast future treated and control outcomes under a factor model for each panel, and consequently to estimate future individual treatment effects with additional structure on the outcome model, e.g. within an interactive fixed effects structure or factor assumption on the treatments. The details are deferred to App. \ref{sec:ite_estimation}.

\section{Main results}\label{sec:theory}

Under regularity conditions on the factor model and the observation mechanism, we establish error bounds and asymptotic normality of the forecast estimator $\thetahatith$. The matrix of idiosyncratic noises is denoted by $\varepsilon \defeq (\varepsilon_{i,t})_{(i,t)\in [N]\times [T]}$. The covariance matrix of stationary $F_t$ is denoted by $\Sigma_F \defeq \EE[F_t F_t^\top]$.

\subsection{Forecast error bound}\label{subsec:forecast_error_rates}

We start with some sufficient regularity conditions on the factor model \eqref{eq:factor_model_1} for stating our main results.

\begin{assumption}[Factor model]\label{asn:factor}
    The factor noise $\eta_t$ in \eqref{eq:var1_factor} satisfies $\EE[\|\eta_t\|_2^4]<\infty$, with some positive definite $\Sigma_\eta$. The loadings $\Lambda_i$ are i.i.d. with mean 0 and positive definite covariance matrix $\Sigma_\Lambda$, with $\EE[\|\Lambda_i\|_2^4] < \infty$. Furthermore, the eigenvalues of $\Sigma_F \Sigma_\Lambda$ are distinct.
\end{assumption}

The i.i.d. loadings in \asn \ref{asn:factor} also appears in \citet{xiong2023large}. The distinct eigenvalues of $\Sigma_F \Sigma_\Lambda$ is a standard condition in the factor model literature \citep{xiong2023large, bai2003inferential, bai2021matrix}.

\begin{assumption}[Idiosyncratic noise]\label{asn:idiosyncratic}
    $\varepsilon_{i,t}$ are i.i.d. with mean $0$, variance $\sigma_\varepsilon^2 > 0$ and $\EE[|\varepsilon_{it}|^8] < \infty$.
\end{assumption}

\begin{assumption}[Mutual independence]\label{asn:independence}
    $\Lambda_i$, $\eta_t$ and $\varepsilon_{i,t}$ are mutually independent for all $i$ and $t$.
\end{assumption}

\asn\ref{asn:idiosyncratic} and \ref{asn:independence} are imposed to establish our results and hold under broader conditions (Sec. \ref{sec:general_asn}). Next we impose regularity conditions on the observation mechanism.

\begin{assumption}[Observation pattern]\label{asn:observation}
    The observation matrix $W$ is independent of the factor matrix $F$ and noise process $\varepsilon$, and there exists $\underline{q} \in (0,1] $ such that for $s,t\in [T]$,
    \begin{equation}\label{eq:q_lower_bound}
        |\cQ_{s,t}| \ge N\underline{q}~~ \text{almost surely}.
    \end{equation}
    
    Furthermore for $s,t,s',t'\in [T]$, there exist $\alpha_{s,t}, \beta_{s,t,s',t'}\in (0,1]$ such that for $N \to \infty$,
    \begin{equation}\label{eq:prop0}
        \frac{1}{N}|\mathcal{Q}_{s,t}| \xrightarrow{\mathrm{a.s}} \alpha_{s,t},~ \text{and}~~ \frac{1}{N}\left|\mathcal{Q}_{s,t} \cap \mathcal{Q}_{s',t'}\right| \xrightarrow{\mathrm{a.s}} \beta_{s,t,s',t'}.
    \end{equation}
\end{assumption}

\asn\ref{asn:observation} is similar to \citet[\asn S1]{xiong2023large}. The condition that $W$ is independent of $\varepsilon$ parallels the \textit{unconfoundedness} assumption in treatment effect identification \citep{rosenbaum1983central}. Requiring independence of $W$ from $F$ and $\varepsilon$ does not confine $W$ to static treatment assignments; $W$ can also capture time-varying treatment policies. We leave more exploration of sequential policies as future work.

\begin{remark}[Validity of \asn \ref{asn:observation}]\label{rem:observation}
The condition $|\cQ_{s,t}|>0$ requires presence of at least one treated unit at any two time points. Under many observation patterns, including \textit{missing completely at random} (MCAR) and \textit{staggered adoption design} (see Sec. \ref{subsec:one-factor_mcar} for details), this condition holds almost surely for large $N$. Similar arguments can be developed for other observation patterns, such as sequential \textit{missing at random} (MAR) design \citep[\asn 1]{dwivedi2022counterfactual}. The condition $|\cQ_{s,t}|>N\underline{q}$ and proportionality conditions \eqref{eq:prop0} can be relaxed to $|\cQ_{s,t}|>f(N)\underline{q}$ with $f(N)=o(N)$, at the cost of slower error bounds for $\thetahatith$ \citep[Sec. 9]{xiong2023large}.
\end{remark}

We denote $\dnt \defeq \min\left\{\sqrt{N}, \sqrt{T}\right\}$ that appears in the error bound and the asymptotic distributions. Equipped with the regularity conditions, we state our first result.

\vspace{5pt}
\begin{theorem}[Error bound for $\thetahatith$]\label{thm:forecast_error_bound}
    \hspace{5pt} Consider a factor model \eqref{eq:factor_model_1} with $N$ units and $T$ time points satisfying Assum. \ref{asn:var1_factors} to \ref{asn:observation}. Then for the \focus estimator $\thetahatith$ in \eqref{eq:forecast_target_est}, any fixed unit $i \in [N]$ and forecast horizon $h\ge 1$, the absolute error associated with $\thetahatith$ is bounded as 
    \begin{align}
        \left|\thetahatith -  \thetaith\right| = & \bigOP{\dnt^{-1}} \nonumber \\
        &\hspace{-50pt} + \bigOP{h\|A\|^{h-1} (N^{-1} + T^{-1/2})}. \label{eq:forecast_error_bound}
    \end{align}
\end{theorem}

The proof is deferred to App. \ref{pf:forecast_error_bound}. The result addresses high probability error bound for unit $\times$ horizon level estimates. $\bigOP{\dnt^{-1}}$ is similar to the unit $\times$ time-level rate for estimating the mean outcome \citep{bai2002determining, bai2021matrix}. The second term $\bigOP{h\|A\|^{h-1} (N^{-1} + T^{-1/2})}$ captures the estimation error of the coefficient matrix $\hat A$ with the fitted factors and forecasting with it. We address the roles of $A$ and $h$ in the error bound in Rem. \ref{rem:role_of_A} and \ref{rem:role_of_h} in App. \ref{pf:forecast_error_bound}.

\subsection{Asymptotic distribution of the forecast}\label{subsec:asymp_dist}

Before stating the asymptotic distribution result for $\thetahatith$, we make additional assumptions.

\begin{assumption}[Time limits of observation pattern]\label{asn:time_limit_observation}
    For $i \in [N]$, there exists positive definite $\Sigma_{F,i}$ such that 
    \begin{equation}\label{eq:lim_FFW}
        \hat\Sigma_{F,i}\defeq \frac{1}{T}\sum_{t=1}^T W_{it} F_t F_t^\top \pconverge \Sigma_{F,i}\quad \text{as }T \to \infty.
    \end{equation}
    For $s,t\in [T]$, $ \alpha_{s,T} \pconverge \nu_s \in (0,1]$ and $\beta_{s,T, t,T} \pconverge \rho_{s,t}\in (0,1]$ as $T\to \infty$.
    Furthermore, there exist $\omega_i\in (0,1),\ i = 1,2,3$ such that the following limits exist as $T\to \infty$,
    \begin{align*}
        & T^{-2} \sum_{s,t=1}^T \frac{\beta_{s,T,t,T}}{\alpha_{s,T} \alpha_{t,T}} \pconverge \omega_1, ~~ T^{-3} \sum_{s,s',t=1}^T \frac{\beta_{s,t,s',T}}{\alpha_{s,t} \alpha_{s',T}} \pconverge \omega_2,\\ & T^{-4} \sum_{s,t, s',t'=1}^T \frac{\beta_{s,t,s',t'}}{\alpha_{s,t} \alpha_{s',t'}} \pconverge \omega_3.
    \end{align*}
\end{assumption}

\eqref{eq:lim_FFW} captures that the factors are systematic over the observed entries. The impact of missing entries on the variance is captured by $\omega_i, \ i = 1,2,3$. Similar conditions also appear in \citet[Assum.S3]{xiong2023large}.

\begin{theorem}[Asymptotic normality of $\thetahatith$]\label{thm:asymptotic_normality}
Consider the setup from Thm. \ref{thm:forecast_error_bound} and suppose Assum.\ref{asn:time_limit_observation} holds. Then for any fixed unit $i \in [N]$ and forecast horizon $h\ge 1$, $\thetahatith$ in \eqref{eq:forecast_target_est} satisfies
\begin{equation}\label{eq:asymptotic_normality}
    \frac{\dnt}{\sigma_{i,T,h}}~\left(\thetahatith - \thetaith\right) \dconverge \normal{0}{1},
\end{equation}
where the asymptotic variance is denoted by $\sforest \defeq \sest + \sfor$, with $\sest$ and $\sfor$ defined in \eqref{eq:variance_est} and \eqref{eq:variance_for} in the Appendix.
\end{theorem}

The proof is deferred to App. \ref{pf:asymptotic_normality}. Thm. \ref{thm:asymptotic_normality} guarantees $\delta_{NT}$-consistency and asymptotically normalilty of $\thetahatith$. Relative to \citet{xiong2023large}, the extension to temporally dependent factors follows from explicitly checking that the general CLT assumptions remain valid under autocorrelation with new Lindeberg and Martingale CLT results, as established in Lem. \ref{lem:imply_5}. In App. \ref{sec:unit_root}, we additionally outline the extension of Thm. \ref{thm:asymptotic_normality} to nonstationary autoregressive factor processes with unit roots.

We note that $\sforest$ has two components-- (a) The asymptotic variance due to estimating the factors from the panel with missing entries is captured by $\sest$, and is similar to the variance of unit $\times$ time-level mean outcome estimation under missing observations in \citet[Cor. 1]{xiong2023large}, and (b) $\sfor$ captures the uncertainty solely due to the forecasting task, and is asymptotically independent of other leading terms (App.~\ref{pf:asymptotic_normality}). When $N/T \to 0$, $\sfor \to 0$ and forecast uncertainty comes only from factor estimation.

\subsubsection{Forecast confidence interval}\label{subsubsec:ci}

If we replace $\sforest$ in Thm. \ref{thm:asymptotic_normality} with its consistent estimator, we can obtain confidence interval for $\thetaith$, and quantify uncertainty of the point forecasts. Using the consistent variance estimators described in Sec. \ref{subsec:hac_estimators}, we state the result for forecast confidence intervals without proof.

\vspace{5pt}
\begin{corollary}[Asymptotic C.I. for $\thetahatith$]\label{cor:ci} Consider the setup of Thm. \ref{thm:asymptotic_normality}. Let $\shatforest \defeq \shatest + \shatfor$ be the asymptotic variance estimator in App. \ref{subsec:hac_estimators}. Then given a level of significance $\alpha \in (0,1)$, the following holds
\begin{align*}
\lim_{N,T\rightarrow\infty} \pr\left(\thetaith \in \Big[\thetahatith ~\mp~  z_{1-\alpha/2} \cdot\frac{\hat\sigma_{i,T,h}}{\delta_{NT}}\Big] \right) & \\ & \hspace{-60pt} = 1 -\alpha,
\end{align*}
where $z_{1-\alpha/2}$ is the $(1-\alpha/2)$-quantile of standard normal distribution.
\end{corollary}
Hence, we provide a $100(1-\alpha)\%$ confidence interval for $\thetaith$ under standard regularity assumptions.

\subsection{Examples}\label{subsec:one-factor_mcar}

We demonstrate our results with a \textit{one-factor model} for the outcome variable. Under \asn \ref{asn:var1_factors}, the factors follow a first order autoregressive (AR(1)) process. We consider some commonly used treatment designs in causal inference as our examples for observation mechanism, with $W$ being independent of other sources of randomness.

\textbf{MCAR.}
$W$ has entries \textit{missing completely at random} (MCAR) if $W_{i,t} \overset{\iid}{\sim} \text{Bernoulli}(p), p \in (0,1]$, i.e. each entry is observed independently with probability $p$. Since in this case, $|\cQ_{s,t}|/N$ is average of $\text{Bernoulli}(p^2)$ random variables, MCAR satisfies $|\cQ_{s,t}| > N \underline{q}$ almost surely as $N \to \infty$ for $0 < \underline{q} < p^2$ (App. \ref{sec:mcar_as}). Thus the condition \eqref{eq:q_lower_bound} in \asn \ref{asn:observation} is reasonable in MCAR setting.

\textbf{Staggered adoption.} The observation matrix $W$ satisfies
$$ W_{i,t} = 1 \text{ for some } t\in [T] \implies W_{i,t'} = 1 \text{ for all } t' \ge t, $$
i.e. an unit that once becomes an adopter stays compliant to the treatment for the rest of the observation period. Staggered adoption design is used for evaluating policy interventions that are implemented at different times across units \citep{ben2021augmented, athey2022design}.

We note that staggered adoption is a special case of \textit{missing at random} (MAR) designs \citep[Sec. 2.2]{xiong2023large} where for each $t$, the cross-sectional units are equally likely to be observed. Similar to MCAR, $|\cQ_{s,t}|/N$ for staggered adoption for unit-wise independent adoption times is an average of Bernoulli random variables, implying that staggered adoption design with independent adoption times across units satisfies the condition \eqref{eq:q_lower_bound} in \asn \ref{asn:observation} for $0 < \underline{q} < \pr(W_{i,1} = 1)$, almost surely as $N \to \infty$.

\begin{corollary}[\focus under MCAR and staggered adoption]\label{cor:mcar}
Consider a stationary one-factor model with AR(1) factors with coefficient $\phi$ satisfying $|\phi| < 1$. Assume that $W \indep \Lambda, F, \varepsilon$, and the observations are subject to MCAR design with observation probability $p$, or staggered adoption design with independent adoption time satisfying $\lim_{t\to \infty} \pr(W_{i,t} = 1) = 1$. Then for every unit $i \in [N]$ and horizon $h \ge 1$, under \asn \ref{asn:factor}-\ref{asn:independence} and provided that condition \eqref{eq:q_lower_bound} holds, the forecast error bound in \eqref{eq:forecast_error_bound} holds. Furthermore, the asymptotic normality in \eqref{eq:asymptotic_normality} holds, with the asymptotic variance given by \eqref{eq:variance_est1} and \eqref{eq:variance_for1}.
\end{corollary}

The proof is deferred to App. \ref{pf:mcar}. In both cases, \asn \ref{asn:observation} and \ref{asn:time_limit_observation} also translate to the asymptotic variance as in App. \ref{pf:asn56_mcar} and \ref{pf:asn56_staggered}. For staggered adoption design, the condition $\lim_{t\to \infty} \pr(W_{i,t} = 1) = 1$ implies that every unit eventually adopts the treatment under almost surely. As $p$ decreases from $1$ to $0$, $\sest$ in \eqref{eq:variance_est1} grows and the asymptotic variance of $\thetahatith$ increases with the probability of missing observations.

\begin{figure*}[!t]
    \centering
    \subfloat[]{\includegraphics[width=0.25\textwidth]{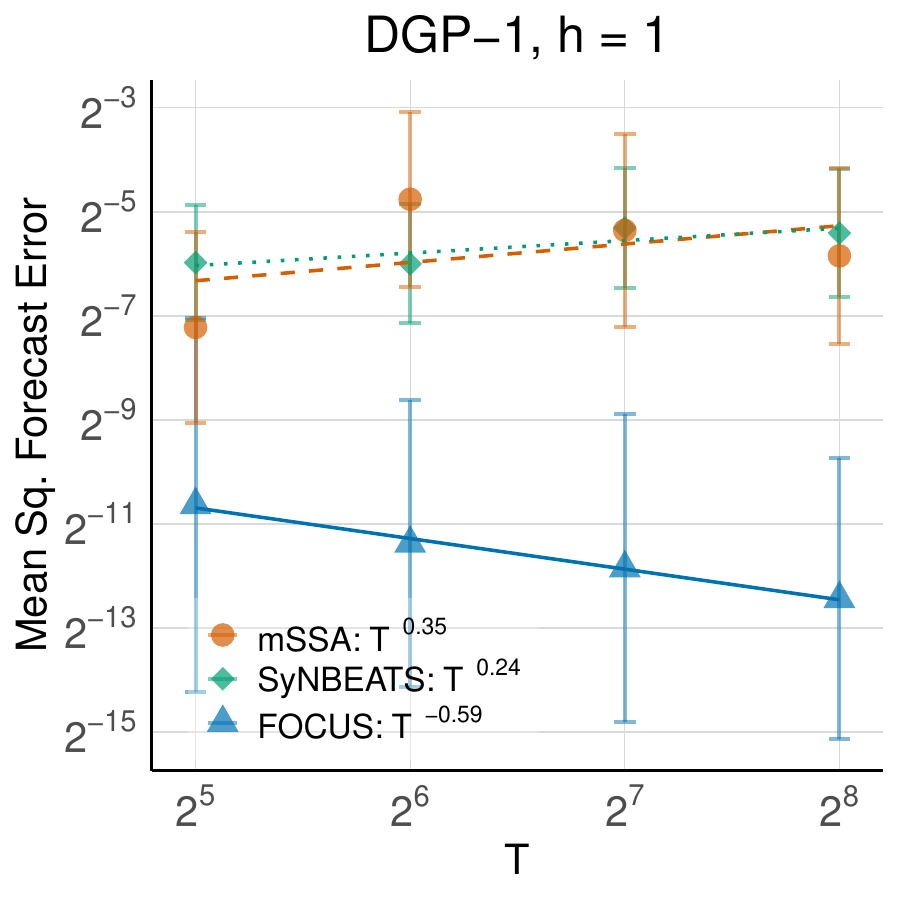}}
    \hfill
    \subfloat[]{\includegraphics[width=0.25\textwidth]{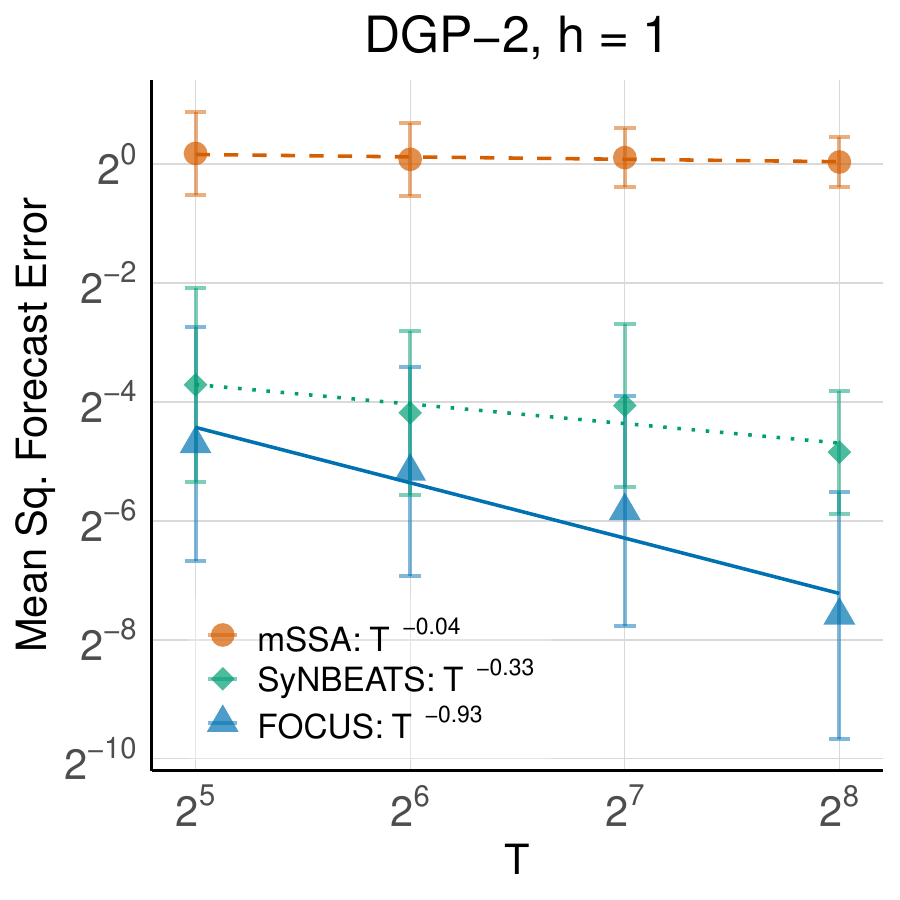}}
    \hfill
    \subfloat[]{\includegraphics[width=0.25\textwidth]{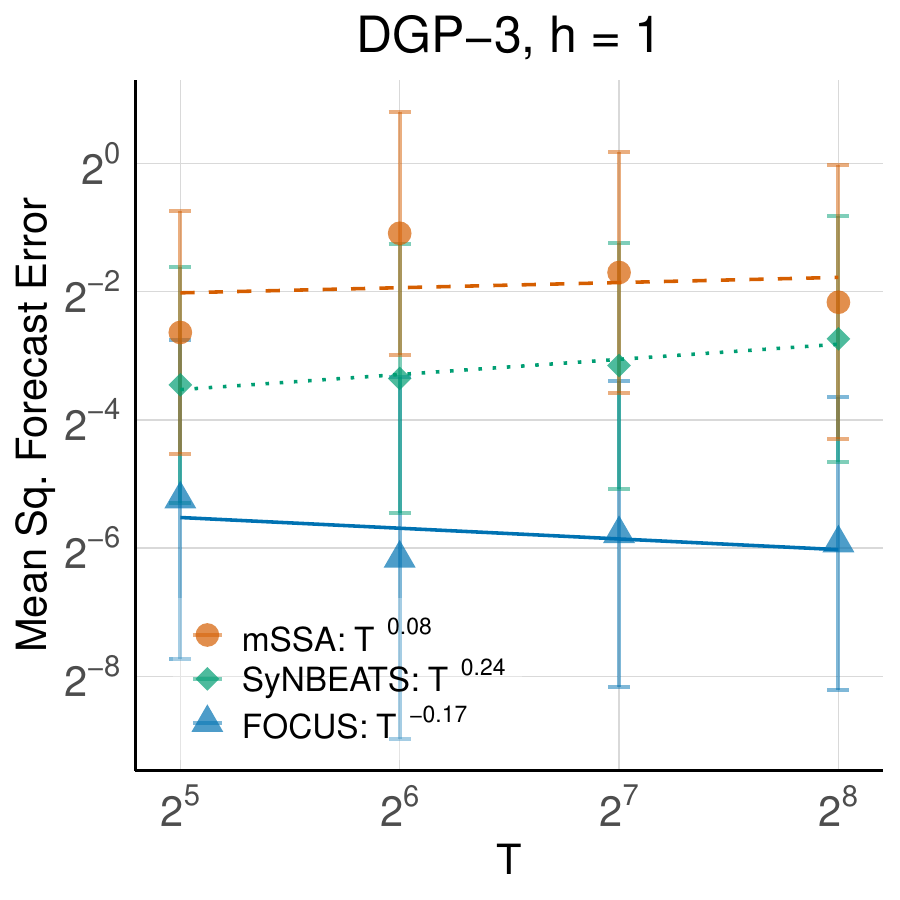}}
    \subfloat[]{\includegraphics[width=0.25\textwidth]{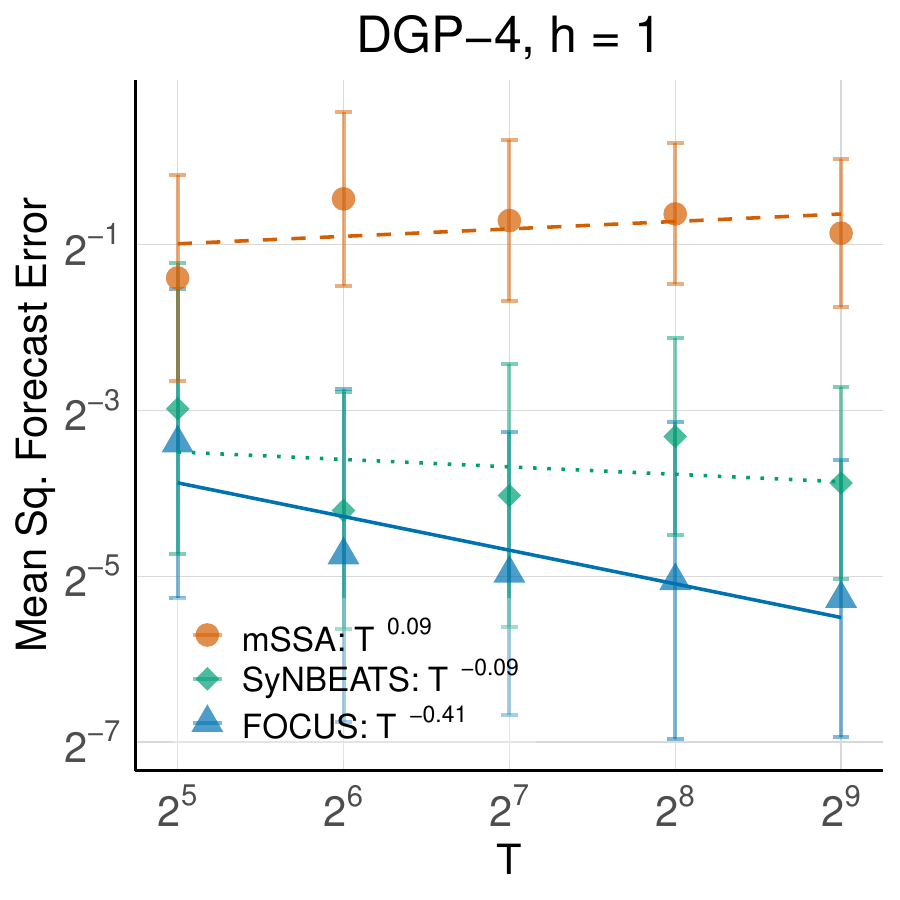}}\\
    \subfloat[]{\includegraphics[width=0.26\textwidth]{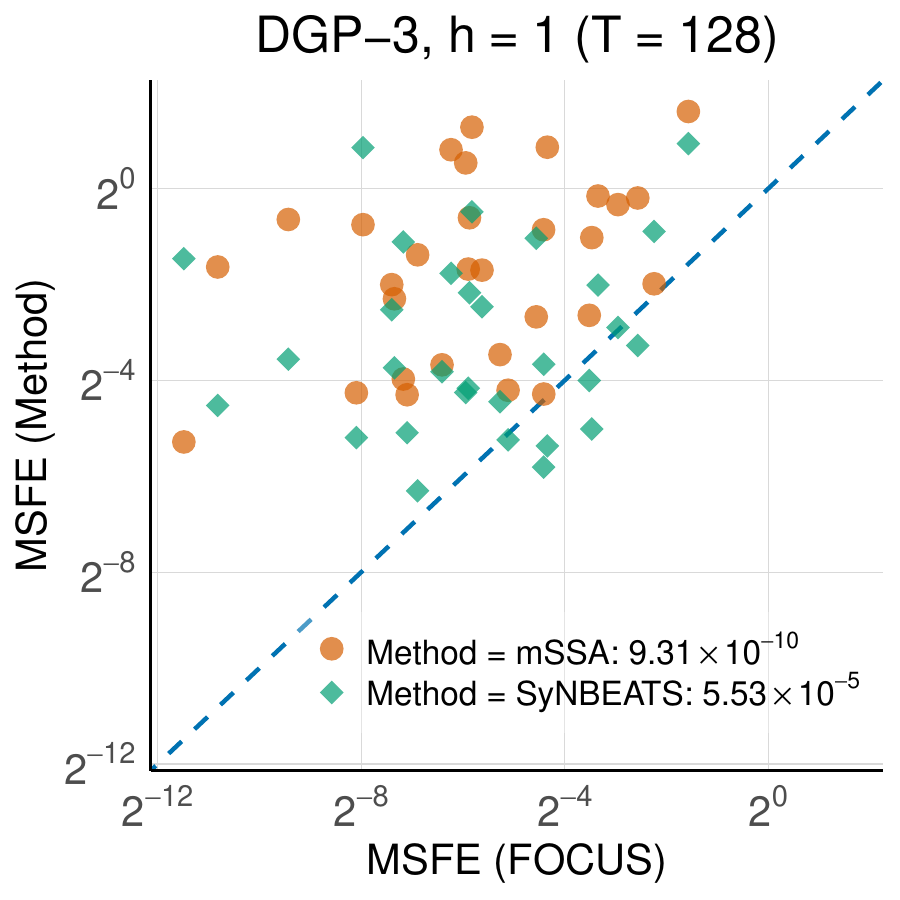}}\hfill
    \subfloat[]{\includegraphics[width=0.26\textwidth]{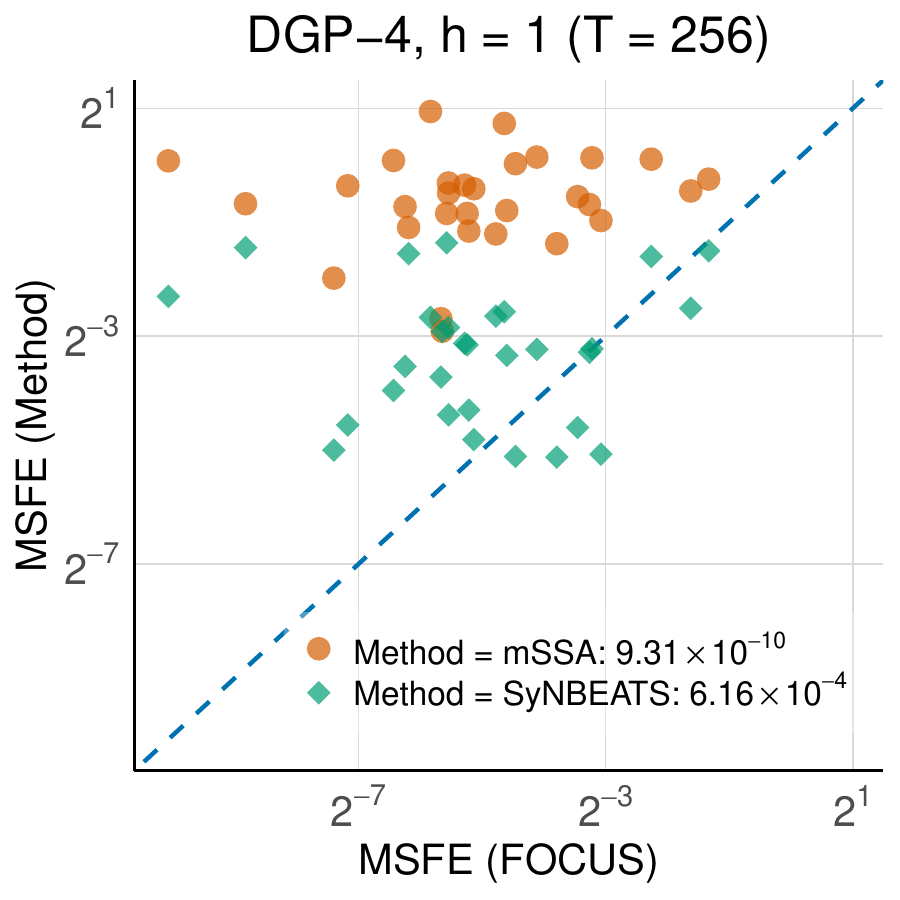}}\hfill
    \subfloat[]{\includegraphics[width=0.48\textwidth]{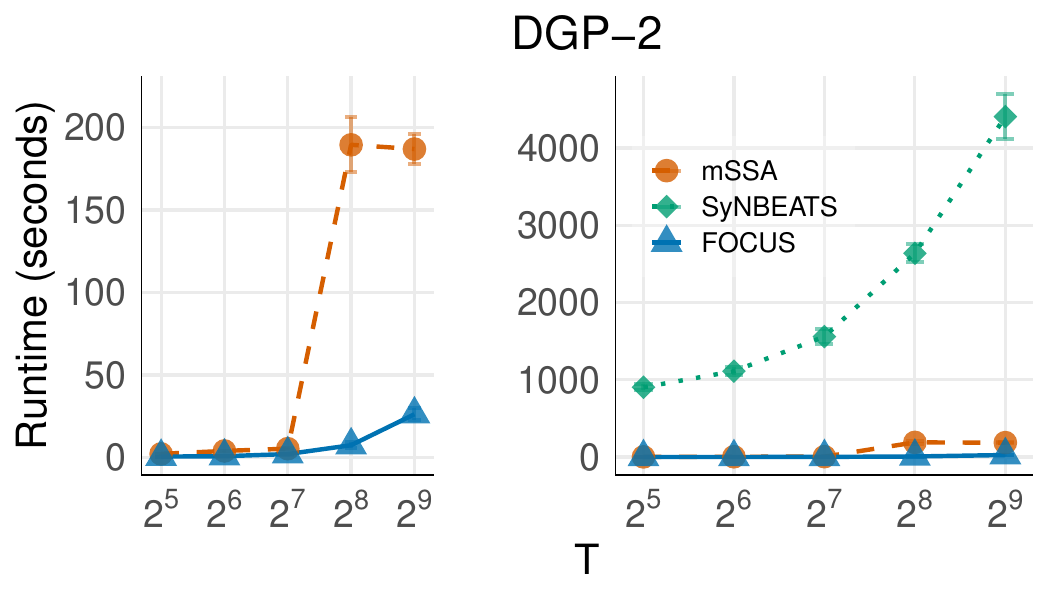}}
    \caption{\textbf{Mean Squared Forecast Error (MSFE) at $h = 1$ and runtime (median across 30 trials, in seconds) across the benchmarks for $N = 64$ and 4 generative models.} Panels (a)-(d) present the average MSFE (across 30 trials) of \focus (blue triangle), mSSA (orange circle) and SyNBEATS (green diamond) across $T \in \{2^5,\ldots, 2^8\}$, and the vertical lines mark the one standard deviation error bars. As comapared to SyNBEATS and mSSA, \focus has lower average MSFE that decreases faster with $T$ (empirical rates in the legends). Panels (e)-(f) present scatter plots of MSFE comparing \focus vs the benchmarks for DGP-3 at $T = 128$ and DGP-4 at $T = 256$. The errors of \focus are significantly lower (p-values of Wilcoxon's one-sided pairwise test in legends $<$ 0.01), resulting the scatter plots concentrated in $y>x$ region. Panel~(g) presents median (across 30 trials) runtime in DGP-2 for $N=64$ and $T\in\{2^5,\ldots,2^9\}$, with vertical bars representing the median absolute deviation of runtimes. The right plot shows all methods, and the left plot zooms in on \focus and mSSA. Across all $T$, \focus (blue triangles) is faster than mSSA (orange circles) and substantially faster than SyNBEATS (green diamonds).}
    \label{fig:err_plots_vs_T}
\end{figure*}

\begin{figure*}[t]
    \centering
    \subfloat[]{\includegraphics[width=0.25\textwidth]{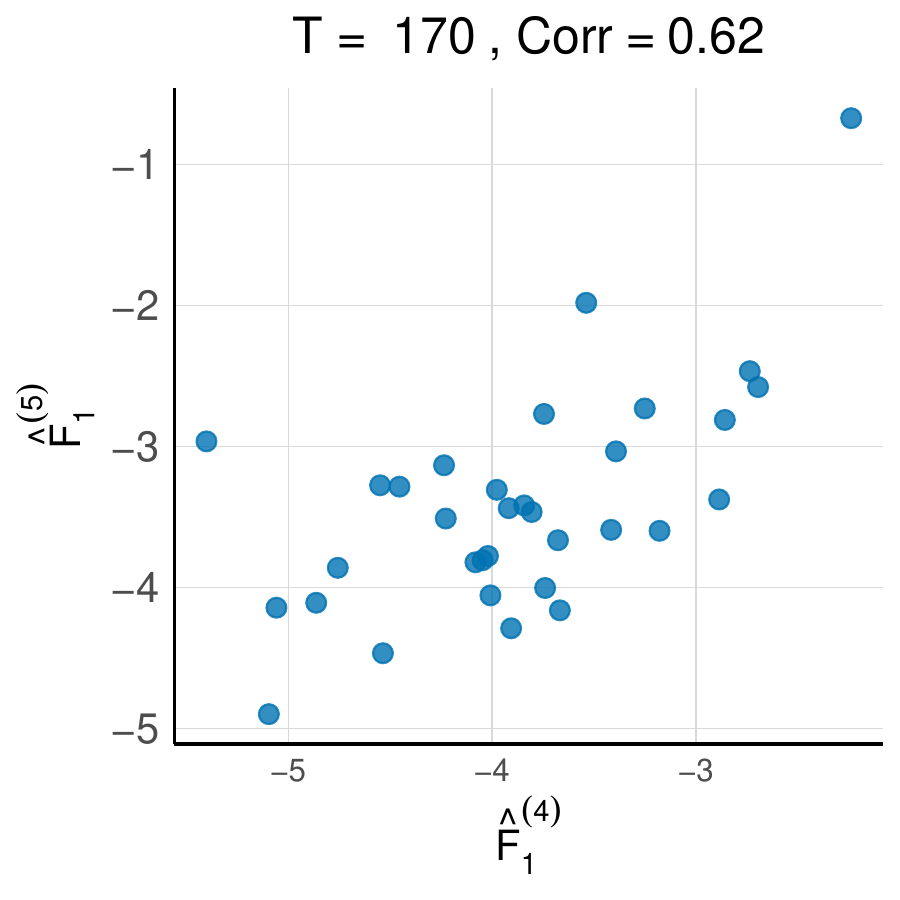}}
    \hfill
    \subfloat[]{\includegraphics[width=0.25\textwidth]{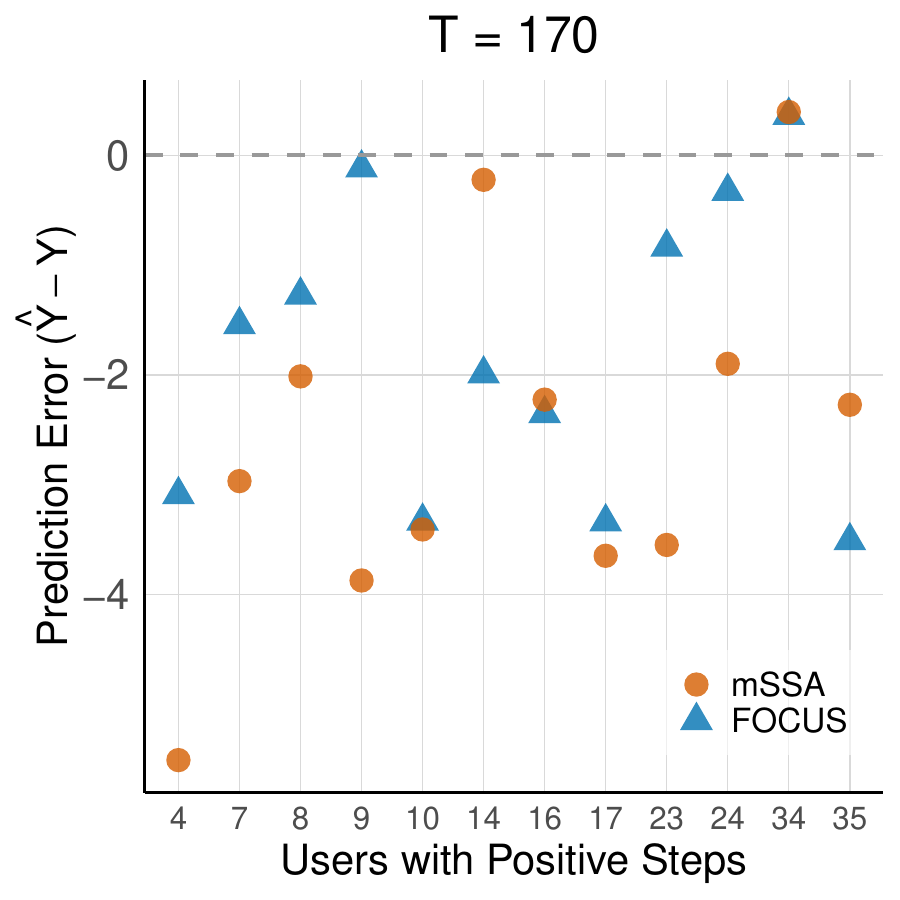}}\hfill
    \subfloat[]{\includegraphics[width=0.25\textwidth]{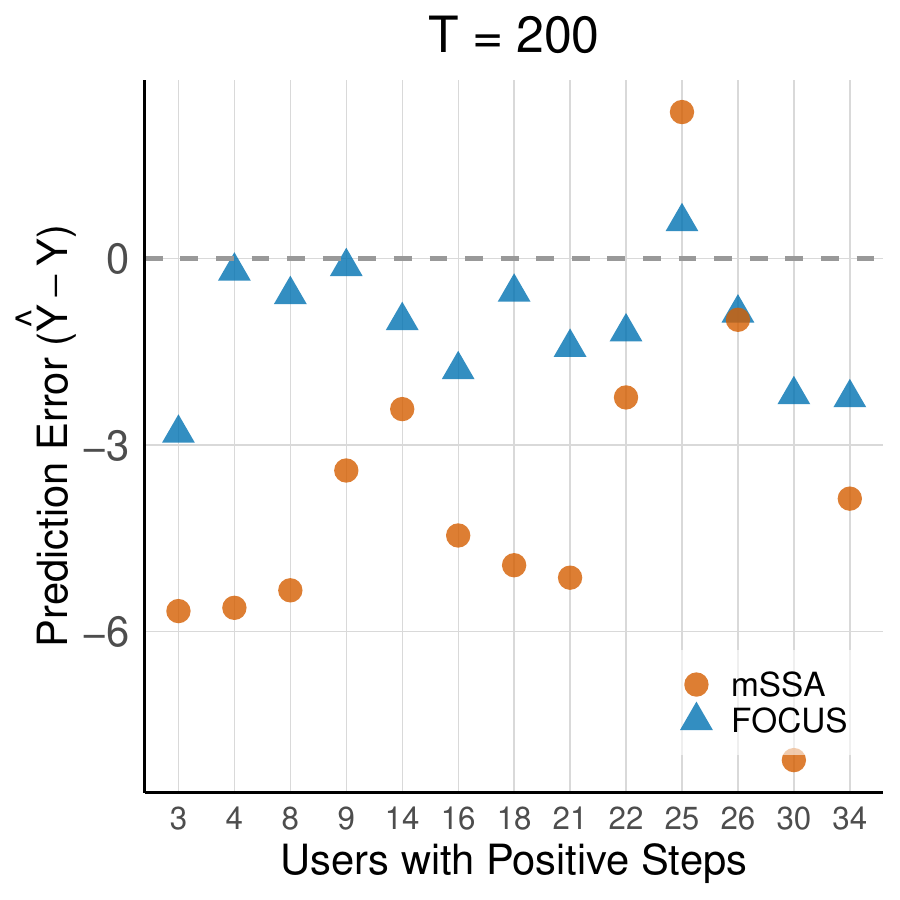}}
    \hfill
    \subfloat[]{\includegraphics[width=0.25\textwidth]{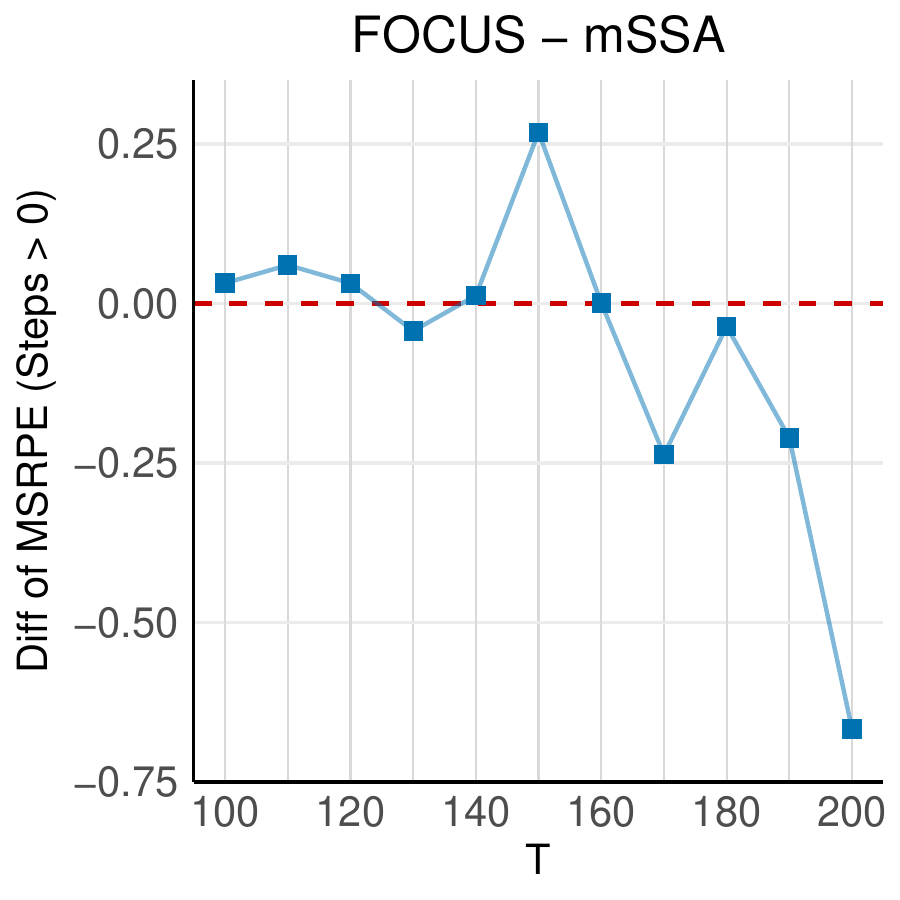}}
    \caption{\textbf{Results of \focus and mSSA for HeartSteps data.} Panel (a) presents a scatter plot of estimated factors for slot pair (4, 5) that shows strong temporal correlation 0.62 for $T = 170$, highlighting the predictive power leveraged by \focus. Panels (b) and (c) present the point-wise prediction error for mSSA and \focus at $T = 170, 200$ and for users with positive steps at $T+h$. For most users, \focus outputs (blue triangle) are closer to the horizontal line at 0 than mSSA (orange circle), yielding more accurate prediction. We note that for both methods, most users have negative prediction errors. Panel (c) presents difference of MSRPE (\focus\ - mSSA) for users with positive steps at $T+h$ and across different $T$. As $T$ grows, the difference stays under the horizontal line ar 0 and decreases with $T$ -- indicating an empirically better performance of \focus with increasing $T$.}
    \label{fig:heartsteps_results}
\end{figure*}

\section{Experiments}\label{sec:experiments}
We conduct experiments on simulated data and the HeartSteps data to compare our proposed method \focus against the benchmark methods when the factors have temporal correlation, particularly an autoregressive component. 
All experiments were conducted on a BioHPC server, with simulations run using 30 CPU cores.

\subsection{Simulation studies} \label{subsec:simulation}
We compare the Mean Squared Forecast Error (MSFE) of \focus with mSSA \citep{agarwal2020multivariate} and SyNBEATS \citep{goldin2022forecasting}. mSSA conducts forecasting by imposing a deterministic latent dynamics, and SyNBEATS requires a fully observed panel up to $T$ and control observations in post treatment period for other units. The details of the simulation setting, performance metric and additional results are deferred to App. \ref{sec:simulation_plus}\footnote{Codes available at: \href{https://github.com/navonildeb/FOCUS}{github.com/navonildeb/FOCUS}} .

\textbf{Setup.} We consider four data generating processes: (a) DGP-1 involves factors from an AR(1) process and the observation pattern $W$ is MCAR(0.7), (b) DGP-2 has factors as a sum of AR(1) process and a quadratic trend, with the observation pattern being simultaneous adoption, (c) DGP-3 has factors with ARMA(3,1) process with a deterministic quadratic trend and MCAR observation pattern with observation probability 0.7, and (d) DGP-4 has 2-dimensional factors with VARMA(1,1) process and a periodic component. We consider DGP-3 and DGP-4 to demonstrate the performance of \focus under model misspecification. The experiments are conducted for $N = 64$, $T \in \{2^5, \ldots, 2^9\}$ and across 30 trials.

\textbf{Results.} As shown in Fig. \ref{fig:err_plots_vs_T}(a)-(d), \focus consistently achieves the lowest MSFE among all benchmarks across varying $T$, DGPs and observation patterns. \focus also exhibits the fastest empirical decay rates for MSFE. Under DGP-2, the observed rate of squared errors is approximately consistent with the theoretical $\bigO{T^{-1/2}}$ in Thm. \ref{thm:forecast_error_bound}. To assess statistical significance, Fig. \ref{fig:err_plots_vs_T}(e)–(h) reports pairwise comparisons at $N=64$ across experiment trials, showing that \focus outperforms all benchmarks with Wilcoxon's test p-values $<10^{-2}$. Importantly, \focus maintains its advantage under model misspecification (DGP-3 and 4). Even when SyNBEATS is granted additional information by including pre-treatment and post-treatment outcomes, \focus still delivers significantly more accurate forecasts. 

As additionally shown in Fig \ref{fig:err_plots_vs_T}(g), \focus yields comparable or better forecast trajectories faster than mSSA, and substantially faster than SyNBEATS under identical computational environments, with the runtime of \focus growing much more scalably with $T$ than the benchmarks. This experiment highlights the efficiency and suitability of \focus for settings with autocorrelated dynamics, short horizon and large $T$, without compromising the forecast accuracy at all. We report the additional simulation results in Fig. \ref{fig:err_vs_T_extra}-\ref{fig:scatterplot_dgp4}.

\subsection{HeartSteps case study}\label{subsec:data}

We illustrate our method with the HeartSteps V1 dataset coming from a 6-week mHealth intervention study involving 37 sedentary adults \citep{heartstepsv1, liao2020personalized}. Participants received walking notifications at five daily decision points via fitness trackers (e.g., Jawbone or Google Fit) that are considered as interventions. The outcome of our interest is the log-step counts 30 minutes  after each suggestion, denoted by \texttt{jbsteps30}. In particular for a user $i$ and decision time $t$, $Y_{i,t} = \log(1 + \texttt{jbsteps30}_{i,t})$. The observation matrix $W$ has entries $W_{i,t} = 1$ if and only if a user is available and receives a nudge, and 0 otherwise. 

We evaluate forecast performance by training on the first $T$ time points and testing on periods $T+1$ to $T+5$. 
For each user $i$, we estimate $\theta_{i, T:T+5}$ i.e. 5-step-ahead mean potential steps under intervention with \focus. Performance is measured by the Mean Squared Relative Prediction Error (MSRPE), computed for users with positive steps. Since SyNBEATS require observations for post-treatment period and fully-observed panel for pre-treatment period, we only consider mSSA as our feasible benchmark. More details regarding the data preprocessing, performance metric and the implementation are deferred to App. \ref{sec:heartsteps_add}.

\textbf{Informative slot pairs.}~ We designate slot pair $(4,5)$ as \textit{informative} due to strong correlations in the estimated factors across consecutive suggestion slots (Fig. \ref{fig:heartsteps_results}). Certain factor components also show high cross-temporal correlation between slots 4 and 5 (Fig. \ref{fig:hs_extra_plot}(a)-(c) in Appendix). We exploit this association by predicting the outcome under a nudge at slot 5 using the estimated factors at slot 4. In line with autoregressive frameworks, where predictions at a given time are based on the immediately preceding time, we model the outcome at slot 5 under a nudge as depending on the outcome at slot 4. For simplicity, we assume that steps under a nudge at slot 4 on a day are independent of steps under a nudge at slot 5 on the previous day.

\textbf{Results.}~ Figure~\ref{fig:heartsteps_results} shows that \focus consistently outperforms mSSA in forecasting future steps under intervention. For $T=170,200$, \focus achieves lower MSRPE, and the performance gap widens as $T$ increases (Figure~\ref{fig:heartsteps_results}(d)). Additional results for $T=190$ (Fig.~\ref{fig:hs_extra_plot} in Appendix) confirm the same pattern. Taken together, these findings demonstrate that explicitly leveraging temporal latent dynamics is critical for accurate counterfactual forecasting, and that \focus provides a scalable advantage over the benchmark mSSA.

\section{Discussion}\label{sec:conclusion}

We propose a counterfactual forecasting method for panel data with low-rank structure and stochastic, time-varying latent factors. Under standard assumptions, we derive error bounds for the forecast estimator and construct valid confidence intervals. Empirically, our method consistently outperforms benchmarks when factors follow autoregressive dynamics, and a proof-of-concept on the HeartSteps V1 dataset highlights its practical utility in capturing temporal patterns in walking behavior.

A limitation of this work is the reliance on stationarity. Extending the framework to more flexible non-stationary settings, e.g. state space models with hidden Markov or Markov switching dynamics, is an important direction for future work. Another future direction is to build \textit{doubly robust} forecast estimators that combine the factor estimators with matching-type estimator, e.g. synthetic controls and nearest-neighbor estimators -- ensuring consistency and robust performance if either the matching assumption or the latent factor dynamics is correctly specified. Further extensions include incorporating covariates for capturing unit-wise heterogeneity, and adapting the method to dynamic treatment regimes and sequential policies.

\bibliography{causalTE}
\bibliographystyle{icml2026}
\newpage
\onecolumn
\appendix

\begin{center}
    \Large
    \textbf{Appendix}
\end{center}
\vspace{10pt}

\paragraph{Additional notation} For any matrix $M$, $\vecop{M}$ denotes the column-wise vectorization of $M$. For matrices $A \in \mathbb{R}^{m \times n}$ and $B \in \mathbb{R}^{p \times q}$, $A \otimes B \in \mathbb{R}^{mp \times nq}$ denotes the Kronecker product defined blockwise as $A \otimes B = [A_{i,j} B]_{i\in [m], j \in [n]}$. We write $X_n = \littleOP{a_n}$ if $\frac{X_n}{a_n} \pconverge 0$ as $n \to \infty$.

\section{Proof of Thm. \ref{thm:forecast_error_bound}: Error bound for $\thetahatith$} \label{pf:forecast_error_bound}

We denote $H \defeq \frac{1}{NT} \hat D^{-1} \hat F^\top \hat F \Lambda^\top \Lambda$. In words, $H$ is a non-singular matrix that aligns $\hat F_t$, $\hat\Lambda_i$ with $F_t$ and $\Lambda_i$ respectively. For $i \in [N]$ and $h\ge 1$, we also denote:
\begin{align*}
    \dload \defeq~ & \hat \Lambda_i - (H^\top)^{-1}\Lambda_i,\\
    \dfactor \defeq~ & \hat F_T - H F_T,\\
    \dcoeff \defeq~ & \hat A^h - H A^h H^{-1} = \hat A^h - (HAH^{-1})^h.
\end{align*}

The $h$-step forecast error is decomposed as 
\begin{align}
   \thetahatith - \thetaith = &~~ \hat \Lambda_i^\top \hat A^h\hat F_T - {\Lambda_i}^\top A^h F_T\nonumber\\
   = &~~ \left[\dload + (H^\top)^{-1} \Lambda_i\right]^\top \left[\dcoeff + H A^h H^{-1}\right] \left[\dfactor + H F_T\right] - {\Lambda_i}^\top A^hF_t\nonumber\\
   = &~~ \dload^\top \dcoeff \dfactor + \dload^\top \dcoeff (H F_T) + \dload (H A H^{-1})^h \dfactor\nonumber\\
   &~~ + ((H^\top)^{-1}\Lambda_i)^\top \dcoeff\dfactor + \dload^\top (HAH^{-1})^h(HF_T) \nonumber\\
   &~~ + ((H^\top)^{-1}\Lambda_i)^\top \dcoeff (HF_T) + ((H^\top)^{-1}\Lambda_i)^\top(HAH^{-1})^h\dfactor. \label{eq:sum_difference_terms}
\end{align}

By Lem. \ref{lem:assumption_lemma}, \asn \ref{asn:var1_factors} to \asn \ref{asn:observation} imply that the general assumptions \asn \ref{gasn:factor_model} to \ref{gasn:moment} hold. From \citet[Lemma A.2]{bai2003inferential}, under Assum. \ref{asn:observation} and \asn \ref{gasn:factor_model} to \ref{gasn:moment} we can show:
\begin{align*}
    \|\dfactor\|_2 = &~~ \|\hat F_T - HF_T\|_2 \\
    = & ~~\bigOP{\frac{1}{\sqrt{T}\delta_{NT}}} + \bigOP{\frac{1}{\sqrt{N}\delta_{NT}}} + \bigOP{\frac{1}{\sqrt{N}}} + \bigOP{\frac{1}{\sqrt{N}\delta_{NT}}} \\
     = & ~~\bigOP{\delta_{NT}^{-1}}.
\end{align*}

Additionally, similar to \citet[Thm. 2.2]{xiong2023large}, we can show that $\|\dload\|_2 = \bigOP{\delta_{NT}^{-1}}$. Lem. \ref{lem:coeff} implies, $\|\Delta_{A,1}\| = \bigOP{N^{-1}} + \bigOP{T^{-1/2}}$. For any horizon $h$, Taylor's expansion implies
\begin{align*}
    \left\|\hat A^h - H A^h H^{-1}\right\| & =~~ \left\|(\dcoeff + H A H^{-1})^h - (H A H^{-1})^h\right\|\\
    & \le~~ \left\|h \dcoeff H A^{h-1} H^{-1}\right\| + \littleOP{1} \\
    & =~~ \bigOP{h \|A\|^{h-1} N^{-1} } + \bigOP{ h \|A\|^{h-1} T^{-1/2}}.
\end{align*}

From \citet[Lemma A.3]{bai2003inferential}, $\|H\| = \bigOP{1}$. By Assum. \ref{gasn:factor_model} $\|\Lambda_i\|_2 = \bigOP{1}$, and $\|F_T\|_2 = \bigOP{1}$. Therefore, the $h$-step forecast error is
$$ \left|\thetahatith - \thetaith\right| = \bigOP{\delta_{NT}^{-1}} + \bigOP{\delta_{NT}^{-1}} + \bigOP{h \|A\|^{h-1} N^{-1}} + \bigOP{h \|A\|^{h-1} T^{-1/2}},$$
and we are done. \pfend

\begin{remark}[Role of $A$ in the error bound] \label{rem:role_of_A}
    As $\rho(A)$ approaches 1, $F_t$ exhibits slower noise decay in its MA representation \eqref{eq:ma_rep}, leading to longer memory and reduced stability \citep[Prop.~2.2]{basu2015regularized}. Since $\|A\| \ge\rho(A)$, the error bound in \eqref{eq:forecast_error_bound} deteriorates accordingly.
\end{remark}

\begin{remark}[Role of $h$] \label{rem:role_of_h}
In practice, $\thetahatith$ is more informative for small to moderate $h$. Lem.~\ref{lem:A_sum} shows--
\[
|\thetaith| < \|\Lambda_i\|_2 \|F_T\|_2 \left( \tfrac{1+\rho(A)}{2} \right)^h, \quad h \ge \underline{N},
\]
where $\underline{N}$ depends on $A$. Since $\|\Lambda_i\|_2$ and $\|F_T\|_2$ are $\bigOP{1}$ by Assum.~\ref{asn:factor} and Lem.~\ref{lem:imply_1} respectively, $\thetaith$ decays exponentially. Hence the second term in \eqref{eq:forecast_error_bound}, which captures forecast error in the estimated factors, vanishes with $h$. For large horizons, both $\thetahatith$ and $\thetaith$ are small, and \eqref{eq:forecast_error_bound} primarily reflects factor model estimation error.
\end{remark}

The proof invokes the following lemmata.

\begin{lemma}[Error bound on empirical autocovariance]\label{lem:imply_2mp_autocov_rate}
Consider a factor model \eqref{eq:factor_model_1} with $N$ units and $T$ time points satisfying Assum. \ref{asn:observation} and \asn \ref{gasn:factor_model} to \ref{gasn:moment}. Define the following matrices:
$$ \tilde \Gamma(\ell) \defeq \frac{1}{T-\ell} \sum_{t = \ell + 1}^{T} HF_{t} (HF_{t-\ell})^\top,\quad
\hat \Gamma(\ell) \defeq \frac{1}{T-\ell} \sum_{t = \ell + 1}^{T} \hat F_{t} \hat F_{t-\ell}^\top, ~~ \ell \in \{0,1\}. $$
Then we have $\left\|\hat \Gamma(\ell) - \tilde \Gamma(\ell)\right\| = \bigOP{\delta_{NT}^{-2}}$.
\end{lemma}

\vspace{20pt}
\begin{lemma}[Error bound on $\hat A$] \label{lem:coeff}
Consider the setup of Lem. \ref{lem:coeff}. Then the following holds: 
$$\| \hat A - H A H^{-1} \| = \bigOP{N^{-1}} + \bigOP{T^{-1/2}}.$$
\end{lemma}

\begin{remark}
    Similar to Lem. \ref{lem:coeff}, $\bigOP{N^{-1} + T^{-1/2}}$ rate is obtained when $\hat A$ is estimated from fully observed panel with PCA-estimated factors under VAR(1) assumption\citep[Prop. 3]{doz2011two}.
\end{remark}

We now prove Lem. \ref{lem:imply_2mp_autocov_rate} and \ref{lem:coeff}.

\subsection{Proof of Lem. \ref{lem:imply_2mp_autocov_rate}: Error bound on empirical autocovariance}\label{pf:emp_autocov_rate}

We have the following:
\begin{align*}
    & ~~ \left\|\hat \Gamma(\ell) - \tilde\Gamma(\ell)\right\|\\
    \le &~~ \left\| \frac{1}{T-\ell} \sum_{t = \ell + 1}^{T} \hat F_{t} \hat F_{t-\ell}^\top - \frac{1}{T-\ell} \sum_{t = \ell + 1}^{T} HF_{t} (HF_{t-\ell})^\top \right\| \\ 
    \le &~~ \left\| \frac{1}{T-\ell}\sum_{t=\ell + 1}^T (\hat F_t - HF_t)(\hat F_{t-\ell} -HF_{t-\ell})^\top \right\| + 2\left\| \frac{1}{T-\ell}\sum_{t = \ell + 1}^T (\hat F_t - HF_t) F_{t-\ell}^\top \right\|\\
    \le &~~ \left[\frac{1}{T-\ell}\sum_{t=\ell + 1}^T \left\|\hat F_t - HF_t\right\|_2^2\right]^{1/2} \left[\frac{1}{T-\ell}\sum_{t=\ell + 1}^T \left\|\hat F_{t-\ell} - HF_{t-\ell}\right\|_2^2\right]^{1/2} + 2\|H\|~\left\|\frac{1}{T-\ell}\sum_{t=\ell + 1}^T(\hat F_t - HF_t) F_{t-\ell}^\top\right\|.
\end{align*}

Using \citet[Lemma A.3]{bai2003inferential}, $\|H\| = \bigOP{1}$. From \citet[Lemma A.1]{bai2003inferential},
$$ \frac{1}{T-\ell}\sum_{t=\ell + 1}^T \left\|\hat F_t - HF_t\right\|_2^2 = \bigOP{\delta_{NT}^{-2}},\quad \frac{1}{T-\ell}\sum_{t=\ell + 1}^T \left\|\hat F_{t-\ell} - HF_{t-\ell}\right\|_2^2 = \bigOP{\delta_{NT}^{-2}}, $$
and \citet[Lemma A.2]{bai2003inferential} implies $\frac{1}{T-\ell}\sum_{t=\ell + 1}^T(\hat F_t - HF_t) F_{t-\ell}^\top = \bigOP{\delta_{NT}^{-2}}$.Hence combining the three rates, we complete the proof. \pfend

\subsection{Proof of Lem. \ref{lem:coeff}: Error bound on $\hat A$}\label{pf:coeff}
We use the fact that for any two invertible matrices $A$ and $B$ of the same dimension,  
$$A^{-1} - B^{-1} = -A^{-1}(A-B)B^{-1} \quad \implies\quad \|A^{-1} - B^{-1}\| \le \|A^{-1}\|~\|A-B\|~\|B^{-1}\|. $$ 

We use the notation of Lem. \ref{lem:imply_2mp_autocov_rate}. Denote 
\begin{equation}\label{eq:A_tilde}
    \tilde A \defeq \tilde \Gamma(1) \tilde\Gamma(0)^{-1} = \left(\sum_{t=1}^{T-1} F_{t+1} F_t^\top\right) \left(\sum_{t=1}^{T-1} F_t F_t^\top\right)^{-1}.
\end{equation}
We have
\begin{align*}
\left\|\hat A - H\tilde A H^{-1} \right\| = &~ \left\| \hat \Gamma(1) \hat\Gamma(0)^{-1} - \tilde \Gamma(1) \tilde\Gamma(0)^{-1} \right\|\\
= &~ \left\| \hat \Gamma(1) \hat \Gamma(0)^{-1} -  \tilde \Gamma(1) \hat \Gamma(0)^{-1} + \tilde \Gamma(1) \hat \Gamma(0)^{-1} - \tilde \Gamma(1) \tilde\Gamma(0)^{-1} \right\|\\
\le &~ \left\|\hat \Gamma(1) - \tilde \Gamma(1)\right\| \left\|\hat \Gamma(0)^{-1}\right\| + \left\|\tilde \Gamma(1)\right\| \left\|\hat \Gamma(0)^{-1} - \tilde\Gamma(0)\right\|\\
\le &~ \left\|\hat \Gamma(1) - \tilde \Gamma(1)\| \|\hat \Gamma(0)^{-1}\| + \|\tilde \Gamma(1)\|\|\hat \Gamma(0)^{-1}\| \|\hat \Gamma(0) - \tilde \Gamma(0)\| \|\tilde \Gamma(0)^{-1}\right\|.
\end{align*}

Assum. \ref{gasn:factor_model} implies that there exists a positive definite matrix $\Sigma_F^{(\ell)}$ such that $\tilde \Gamma(\ell) = H \Sigma_F^{(\ell)} H^\top + \littleOP{1}$. In addition, using Lem. \ref{lem:imply_2mp_autocov_rate} we have
$$\hat \Gamma(\ell) = \tilde \Gamma(\ell) + \bigOP{\delta_{NT}^{-2}} = H \Sigma_F^{(h)} H^\top + \littleOP{1}.$$ 

Thus under \asn \ref{gasn:factor_model}, $\|\tilde \Gamma(h)^{-1}\| = \bigOP{1}$ and $\|\hat \Gamma(h)^{-1}\| = \bigOP{1}$, which implies 
$$\left\|\hat A - H\tilde AH^{-1}\right\| = \bigOP{\delta_{NT}^{-2}}.$$

From \citet[Prop. 3.1]{lutkepohl2005new}, $\|\tilde A - A\| = \bigOP{T^{-1/2}}$. Therefore, combining all the terms and using $H = \bigOP{1}$ \citep[Lem. A.1]{bai2003inferential}, the result holds. \pfend

\vspace{10pt}

\section{Proof of Thm. \ref{thm:asymptotic_normality}: Asymptotic normality of $\thetahatith$}\label{pf:asymptotic_normality}

Before we stat the proof, following we introduce some necessary shorthands for formulating the asymptotic variance of $\thetahatith$. For $i \in [N]$, we denote
\begin{align*}
    {\cal V}_\Lambda & \defeq \EE[\vecop{\Lambda_i \Lambda_i^\top - \Sigma_\Lambda}~ \vecop{\Lambda_i \Lambda_i^\top - \Sigma_\Lambda}^\top],\\
    \Psi_T & \defeq (F_T\otimes \Sigma_F)~ (\Sigma_F^{-1} \Sigma_{\Lambda}^{-1}),\\
    \Upsilon_i & \defeq (\Sigma_{F,i}\otimes \Sigma_{\Lambda}^{-1}\Lambda_i)\Sigma_{F,i}^{-1}.
\end{align*}

We define the following covariance matrices that will be needed to express the asymptotic variance of $\thetahatith$:
\begin{align}
    \SigmaFobs & \defeq \sigma_\varepsilon^2 \Sigma_{\Lambda}^{-1},\quad 
    \SigmaFmiss\defeq \Psi_T^\top {\cal V}_{\Lambda}\Psi_T,\quad
    \SigmaLobs \defeq \sigma_\varepsilon^2 \Sigma_{F,i}^{-1}, \label{eq:variance_shorthand_1}\\
    \SigmaLmiss & \defeq \Upsilon_i^\top {\cal V}_{\Lambda}\Upsilon_i, \quad
    \SigmaFLmisscov \defeq \Psi_T^\top {\cal V}_{\Lambda}\Upsilon_i. \label{eq:variance_shorthand_2}
\end{align}

Similar to \eqref{eq:sum_difference_terms}, we write $\delta_{NT}(\thetahatith - \thetaith)$ where each term is a product of at most 3 error terms. The product terms containing only one error term dominates the distribution and the rest of the terms are stochastically negligible. Therefore,
\begin{align*}
    \delta_{NT}(\thetahatith - \thetaith) = & \underbrace{\delta_{NT}((H^\top)^{-1}\Lambda_i)^\top (\hat A^h - (HAH^{-1})^h)(HF_T)}_{\text{(I)}} \\
    & \vspace{10pt} +  \underbrace{\delta_{NT}((H^\top)^{-1}\Lambda_i)^\top (HAH^{-1})^h (\hat F_T - HF_T)}_{\text{(II)}}\\
    &  \vspace{10pt} +\underbrace{\delta_{NT}(\hat \Lambda_i - (H^\top)^{-1}\Lambda_i)^\top (HAH^{-1})^h(HF_T)}_{\text{(III)}}~~~ + \littleOP{1}.
\end{align*}

From Lem. \ref{lem:clt_varcoef},
\begin{equation}\label{eq:clt_term_1}
    \text{(I)} = ~\delta_{NT} \Lambda_i^\top H^{-1} \left(\hat A^h - (HAH^{-1})^h\right) HF_T = \delta_{NT} \Lambda_i^\top \left(H^{-1}\dcoeff H \right) F_T + \littleOP{1}.
\end{equation}

From Lem. \ref{lem:asymp_indep},
\begin{equation}\label{eq:clt_term_2}
    \text{(II)} = ~\delta_{NT} \Lambda_i^\top A^h (H^{-1}\hat F_T - F_T) = \delta_{NT} ((A^\top)^h\Lambda_i)^\top (e_{F, T}^{(1)} + e_{F, T}^{(2)}),
\end{equation}

and
\begin{equation}\label{eq:clt_term_3}
    \text{(III)} = ~\delta_{NT}(H^\top\hat \Lambda_i - \Lambda_i)^\top (A^hF_T)
    = \delta_{NT}(e_{\Lambda,i}^{(1)} + e_{\Lambda,i}^{(2)})^\top(A^hF_T) + \littleOP{1}.
\end{equation}

Additionally by Lem. \ref{lem:asymp_indep}, every term in (I), (II) and (III) are asymptotically independent except $e_{F, T}^{(2)}$ and $e_{\Lambda, i}^{(2)}$. By Lem. \ref{lem:lim_cross_cov}, the limiting asymptotic covariance term is $-2(\omega_2-1)\SigmaFLmisscov$. Therefore, combining \ref{lem:clt_varcoef}, Lem. \ref{lem:factor_loading_clt} and an asymptotic argument similar to \citet[Eq. (C.1)]{bai2003inferential}, we obtain
$$ \delta_{NT}\Omega_{i,T,h}^{-1}(\thetahatith - \thetaith) \xrightarrow{d} \normal{0}{1},$$
where the asymptotic variance term is $\Omega_{i,T,h} = \sest + \sfor$ with
\begin{align}
\sest \defeq & \Lambda_i^\top A^h \Omega_T^{(F)} (A^h)^\top \Lambda_i + F_T^\top (A^h)^\top \Omega_i^{\Lambda} A^h F_T- 2(\omega_2-1)F_T^\top (A^h)^\top \SigmaFLmisscov (A^h)^\top \Lambda_i \nonumber\\
= & \frac{\dnt^2}{N}\bigg[~ \Lambda_i^\top A^h~ \left\{ \omega_1 ~ \SigmaFobs + (\omega_1 -1)~\SigmaFmiss\right\}~(A^h)^\top\Lambda_i - 2(\omega_2 - 1)~ F_T^\top (A^h)^\top~ \SigmaFLmisscov~ (A^h)^\top \Lambda_i \label{eq:variance_est}\\
&\hspace{10pt} + (\omega_3 - 1) F_T^\top (A^h)^\top \SigmaLmiss A^h F_T \bigg] + \frac{\dnt^2}{T}F_T^\top (A^h)^\top \SigmaLobs A^h F_T, \nonumber
\end{align}
and
\begin{equation}\label{eq:variance_for}
\sfor \defeq \Omega^{(A)}_{\Lambda_i, F_T}
= \frac{\dnt^2}{T}\sum_{k,\ell = 0}^{h-1} \left(\Lambda_i^\top(A^{h-1-k})^\top \Sigma_F^{-1} A^{h-1-\ell} \Lambda_i\right)\left(F_T^\top A^k \Sigma_\eta (A^\ell)^\top F_T\right),
\end{equation}

where $\SigmaFobs$, $\SigmaFmiss$, $\SigmaLobs$, $\SigmaLmiss$ and $\SigmaFLmisscov$ are defined in \eqref{eq:variance_shorthand_1} and \eqref{eq:variance_shorthand_2}. Hence the proof is complete. \pfend

\vspace{5pt}
Below we state some asymptotic normality results used in the proof of Thm \ref{thm:asymptotic_normality}.

\vspace{5pt}
\begin{lemma}[Asymptotic normality of $\hat A$]\label{lem:clt_varcoef}
Consider a factor model \eqref{eq:factor_model_1} with $N$ units and $T$ time points satisfying Assum. \ref{asn:var1_factors} to \ref{asn:observation}. Then for $\sqrt{T}/N \rightarrow 0$, $h \in \NN$ and $u, v\in \RR^r$,
$$ \delta_{NT} {\Omega_{u,v}^{(A)}}^{-1/2}~ u^\top (H^{-1} \hat A^h H - A^h) H v \xrightarrow{d}\normal{0}{1}, $$

where $\Omega_{u,v}^{(A)} := \frac{\delta_{NT}^2}{T}\sum_{k,\ell = 0}^{h-1} \left(v^\top [A^{h-1-k}]^\top \Sigma_F^{-1} A^{h-1-\ell} v\right)\cdot \left( u^\top A^k \Sigma_\eta [A^\ell]^\top u \right) $.
\end{lemma}

\begin{lemma}[Asymptotic independence]\label{lem:asymp_indep}
Consider a factor model \eqref{eq:factor_model_1} with $N$ units and $T$ time points satisfying Assum. \ref{gasn:factor_model} to \ref{gasn:limdist}. Then for $\sqrt{N}/T \rightarrow 0$ and $\sqrt{T}/N \to 0$, 
$$ H^{-1} F_T = e_{F, T}^{(1)} + e_{F, T}^{(2)}, \quad H^\top \hat\Lambda_i - \Lambda_i = e_{\Lambda,i}^{(1)} + e_{\Lambda, i}^{(2)}, $$

with $e_{F, T}^{(1)}$ and $e_{F, T}^{(2)}$ being asymptotically independent, $e_{\Lambda, i}^{(1)}$ and $e_{\Lambda, i}^{(2)}$ being asymptotically independent, and
$$ \dnt^2 \cov(e_{\Lambda, i}^{(2)}, e_{F,T}^{(2)}) \to \Sigma_{\Lambda}^{-1} \Sigma_F^{-1} [g_{i,T}^\mathrm{cov}(u_i)]^\top \Sigma_{F,i}^{-1}, $$
where $g_{i,T}^\mathrm{cov}(u_i)$ is defined in \eqref{eq:limdist3} in \asn \ref{gasn:limdist}. Furthermore, each of $e_{F, T}^{(1)}$, $e_{F, T}^{(2)}$, $e_{\Lambda, i}^{(1)}$ and $e_{\Lambda, i}^{(2)}$ is asymptotically independent with $\Delta_{A, h}$.
\end{lemma}

\begin{lemma}[Simplified form of asymptotic covariance]\label{lem:lim_cross_cov}
    Under \asn \ref{asn:var1_factors} to \ref{asn:time_limit_observation}, $$\Sigma_{\Lambda}^{-1}\Sigma_F^{-1} [g_{i,T}^\mathrm{cov}(u_i)]^\top \Sigma_{F,i}^{-1} = -2(\omega_2-1)\SigmaFLmisscov.$$
\end{lemma}

\begin{lemma}[Asymptotic normality of $\hat F_T$ and $\hat \Lambda_i$]\label{lem:factor_loading_clt}
Under the assumptions of Lem. \ref{lem:asymp_indep},
$$\delta_{NT}{\Omega_{T}^{(F)}}^{-1/2}(H^{-1}\hat  F_T - F_T) \xrightarrow{d}\normal{0}{I_r}\quad \text{and}\quad \delta_{NT}{\Omega_{i}^{(\Lambda)}}^{-1/2}(H^\top\hat\Lambda_i - \Lambda_i) \xrightarrow{d}\normal{0}{I_r}$$ 

where
$\Omega_{T}^{(F)} := \frac{\delta_{NT}^2}{N}\Sigma_{\Lambda}^{-1}\Sigma_F^{-1}[\Gamma_T^\mathrm{obs} + \Gamma_T^\mathrm{miss}]\Sigma_F^{-1}\Sigma_\Lambda^{-1}$ and $\Omega_{i}^{(\Lambda)} := \Sigma_{F,i}^{-1}\left[ \frac{\delta_{NT}^2}{T}\Phi_i^\mathrm{obs} + \frac{\delta_{NT}^2}{N} \Phi_i^\mathrm{miss} \right] \Sigma_{F,i}^{-1}$. Under Assum. \ref{asn:var1_factors} to \ref{asn:time_limit_observation}, $\Omega_{t}^{(F)}$ and $\Omega_{i}^{(\Lambda)}$ are simplified as
\begin{equation}\label{eq:cov_f_simple}
    \Omega_T^{(F)} = \frac{\delta_{NT}^2}{N}\left[\omega_1 \SigmaFobs + (\omega_1 - 1) \SigmaFmiss\right],\quad \Omega_i^{(\Lambda)} = \frac{\delta_{NT}^2}{T}\SigmaLobs + \frac{\delta_{NT}^2}{N}(\omega_3 - 1)\SigmaLmiss,
\end{equation}
where $\SigmaFobs$, $\SigmaFmiss$, $\SigmaLobs$ and $\SigmaLmiss$ are defined in \eqref{eq:variance_shorthand_1} and \eqref{eq:variance_shorthand_2}.
\end{lemma}

The proofs of Lemmata \ref{lem:asymp_indep}, \ref{lem:lim_cross_cov} and \ref{lem:factor_loading_clt} follow from \citet[Thm. 1 and Thm. 2]{bai2003inferential} and \citet[Thm. 2, Corollary 1]{xiong2023large}. The proof techniques involve the roles of $T$ and $N$ exchanged in order to maintain the structure of \citet[Thm. 1 and Thm. 2]{bai2003inferential} and conduct our theoretical analysis. The final statement of Lem. \ref{lem:asymp_indep} follows since the randomness in $\dcoeff$ is generated by the time average of $\eta_t, t = 2,\ldots, T$ and hence independent from the source of randomness of the other error terms. Hence a mutual asymptotic independence of the three terms hold. We only provide a proof of Lem. \ref{lem:clt_varcoef}.

\subsection{Proof of Lemma \ref{lem:clt_varcoef}: Asymptotic normality of $\hat A$} \label{pf:clt_varcoef}

Denote the following matrices $F^{(-T)} \defeq [F_1:\ldots: F_{T-1}]^\top_{(T-1) \times r}$ and $U \defeq [\eta_2: \ldots: \eta_T]$. From \citet[Prop. 3.1]{lutkepohl2005new}, we have
\begin{align*}
    \tilde A - A = \left[ \left( \frac{1}{T-1} {F^{(-T)}}^\top F^{(-T)}\right) \otimes I_r \right] \frac{1}{\sqrt{T-1}} \left({F^{(-T)}}^\top \otimes I_r\right) \text{vec}(U)
\end{align*}
and the following holds:
\begin{equation}\label{eq:clt_A}
    \sqrt{T}\ \text{vec}(\tilde A - A) \xrightarrow{d} \normal{0}{\Sigma_F^{-1} \otimes \Sigma_\eta},
\end{equation}

where $\tilde A$, defined in \eqref{eq:A_tilde}, is the OLS estimator of $A$ with the true rotated factors $F_t$ up to a rotation with $H$.

First we fix $h = 1$. Following \ref{lem:coeff},
$$ \text{vec}(H^{-1}\hat A H -  A) = \text{vec}(H^{-1}\hat A H -  \tilde A) + \text{vec}(\tilde A - A) = \text{vec}(\tilde A - A) + \bigOP{N^{-1}} + \bigOP{T^{-1/2}}. $$

If $\sqrt{T}/N \rightarrow 0$, the second term dominates and contributes to the asymptotic distribution i.e.
$$ \sqrt{T}\ \text{vec}(H^{-1}\hat A H - A) \xrightarrow{d} \normal{0}{\Sigma_F^{-1} \otimes \Sigma_\eta}. $$

Next we consider the matrix valued operator $f(A) = \text{vec}(A^h),~~~ A\in \RR^{r\times r}$. Following \citet[Thm. 3.1]{al2009computing}, $A \mapsto A^h$ has a Fr\'{e}chet derivative in the directional matrix $E$ as  $D_{A^h}(A;E) = \sum_{k=0}^{h-1} A^k E A^{h-1-k}$, and hence the vectorized gradient operator is $\nabla f(A) = \sum_{k=0}^{h-1} \left((A^{h-1-k})^\top \otimes A^k\right)$. Therefore, applying multivariate Delta-method, the asymptotic variance term becomes
$ \left[\sum_{k=0}^{h-1} \left((A^{h-1-k})^\top \otimes A^k\right)\right] (\Sigma_F^{-1}\otimes \Sigma_\eta)\left[\sum_{k=0}^{h-1} \left((A^{h-1-k})^\top \otimes A^k\right) \right]^\top $, and the rest follows from the standard properties of Kronecker products. \pfend

\subsection{HAC estimators of $\sforest$}\label{subsec:hac_estimators}

A consistent estimator of the asymptotic forecast variance $\sforest$ is $\shatforest = \shatest + \shatfor$ where the individual components are described in the following:

\paragraph{Estimation of $\sest$} We estimate $\shatest$, a consistent estimator of $\sest$, via the plug-in approach. In \eqref{eq:variance_est}, the plug-in estimators $\hat\Lambda_i$, $\hat F_T$, and $\hat A$ are obtained from \focus($Y,W$) as described in Sec.~\ref{sec:method}. Consistent estimators of $\SigmaFobshat$ and $\SigmaLobshat$ are constructed using HAC estimators \citep{newey1987hypothesis}, while $\SigmaFmisshat$, $\SigmaLmisshat$, $\SigmaFLmisscovhat$, and $\hat\omega_i$ rely on the same plug-in principle. Details can be found in \citet[Sec. 8]{xiong2023large}

\paragraph{Estimation of $\sfor$} For estimating $\sfor$, we can again use a plug-in approach similar to \citet[Eq. 3.2.17, Prop. 3.2]{lutkepohl2005new} by replacing the true unknown parameters in \eqref{eq:variance_for} with corresponding plug-in estimators. We can use the estimators $\hat \Sigma_F = \frac{1}{T}\sum_{t=1}^T \hat F_t \hat F_t^\top$, $\hat \Sigma_\eta = \frac{1}{T-1}\sum_{t=2}^{T}(\hat F_t - \hat A \hat F_{t-1})(\hat F_t - \hat A \hat F_{t-1})^\top$, and the variance estimator can be calculated as
$$ \shatfor := \frac{\dnt^2}{T}\sum_{k,\ell = 0}^{h-1} \left(\hat \Lambda_i^\top(\hat A^{h-1-k})^\top \hat \Sigma_F^{-1} \hat A^{h-1-\ell} \hat \Lambda_i\right)\left(\hat F_T^\top \hat A^k \hat \Sigma_\eta (\hat A^\ell)^\top \hat F_T\right). $$

A detailed discussion of consistency of the estimators $\hat \Sigma_F$ and $\hat \Sigma_\eta$ and the translation of asymptotic normality with the estimated variance can be found in \citet[Prop. 3.2, Cor. 3.2.1]{lutkepohl2005new}.

\section{Proof of Cor. \ref{cor:mcar}: \focus under MCAR and staggered adoption}\label{pf:mcar}

We consider the one factor model $Y_{i,t} = \Lambda_i F_t + \varepsilon_{i,t}$ with AR(1) factors $F_t$ with AR coefficient $\phi$ satisfying $|\phi| < 1$. Thus \asn \ref{asn:var1_factors} is satisfied. We introduce the notation $\sigma_\eta^2 \defeq \var(\eta_t), \sigma_F^2 \defeq \var(F_t)$ and $\sigma_\Lambda^2 \defeq \var(\Lambda_i)$. Next we verify that \asn \ref{asn:observation} and \ref{asn:time_limit_observation} hold under MCAR and staggered adoption.

\subsection{\asn \ref{asn:observation} and \ref{asn:time_limit_observation} in MCAR}\label{pf:asn56_mcar}

For any $i\in [N]$ and $s,t\in [T]$, $W_{i,s} W_{i,t}\sim \text{Bernoulli}(p^2)$ and are independent across $i$. Therefore $\frac{1}{N}|\cQ_{s,t}| = \frac{1}{N}\sum_{i = 1}^N W_{i,s} W_{i,t}$ is an average of $N$  $\text{Bernoulli}(p^2)$ random variables. Applying strong law of large numbers, Assum.~\ref{asn:observation} is satisfied with
\begin{equation}\label{eq:mcar_q2}
    \alpha_{s,t} = \begin{cases}
    p & \text{if } s= t,\\
    p^2 & \text{otherwise}.
\end{cases}
\end{equation}
Similarly, $\frac{1}{N}|\cQ_{s,t}\cap \cQ_{s',t'}| = \frac{1}{N}\sum_{i = 1}^N W_{i,s} W_{i,t} W_{i,s'} W_{i,t'}$. Again applying strong law of large numbers,
\begin{equation}\label{eq:mcar_q4}
    \beta_{s,t,s',t'} = p^{\, \#\text{distinct}(\{s,t,s',t'\})},
\end{equation}
where for any set $S$, $\#\text{distinct}(S) \defeq$ Number of distinct elements in $S$. We now show that \asn \ref{asn:time_limit_observation} holds under MCAR. We fix any $s, t\in [T]$. Using \eqref{eq:mcar_q2} and \eqref{eq:mcar_q4}, Assum. \ref{asn:time_limit_observation} is satisfied with 
\begin{align*}
    \nu_s & = p^2,\quad
    \rho_{s,t} = \begin{cases}
        p^2 & \text{if } s = t,\\
        p^2 & \text{otherwise},
    \end{cases}\\
    \omega_1 & = 1/p, \quad
    \omega_2 = \omega_2 = 1.
\end{align*}

Furthermore since $|\phi|< 1$, for any fixed $i\in [N]$ we can write \eqref{eq:lim_FFW} as
$$ \frac{1}{T}\sum_{t=1}^T W_{i,t} F_t^2 = \underbrace{\frac{1}{T}
\sum_{t=1}^T (W_{i,t} - p)F_t^2}_{\defeq A_T} + \underbrace{\frac{p}{T}\sum_{t=1}^T F_t^2}_{\defeq B_T}. $$

Following from \citet[Thm. 11.2.2 and Eq. 11.2.6]{brockwell1991time}, $B_T \pconverge p\sigma_F^2$. Additionally, we have $\EE[A_T] = 0$. Denote $Z_{t}\defeq (W_{i,t} - p)F_t^2$. Thus, for any $s, t\in [T]$
\begin{align*}
    \cov(Z_s, Z_t) = \EE[(W_{i,s} - p)(W_{i,t} - p)F_s^2 F_t^2 ]
    = \begin{cases}
        p(1-p)\EE[F_t^4] & \text{if } s = t\\
        0 & \text{otherwise}.
    \end{cases}
\end{align*}
Therefore under \asn \ref{asn:factor},
$$ \var(A_T) = \frac{1}{T^2}\sum_{s,t=1}^T \cov(Z_s, Z_t) =  \frac{1}{T^2}\sum_{t=1}^T \var(Z_t) = \frac{1}{T}p(1-p)\EE[F_1^4] = \bigO{T^{-1}}. $$

Using \textit{Chebychev's inequality} for any $\delta > 0$, $\pr(A_T > \delta)\le \var(A_T)/\delta^2 = \bigO{T^{-1}} \to 0$ as $T\to \infty$. Therefore $A_T \pconverge 0$, and we obtain $\Sigma_{F,i} = p\sigma_F^2$ -- showing that MCAR satisfies Assum. \ref{asn:time_limit_observation} for AR(1) factors.

\subsection{\asn \ref{asn:observation} and \ref{asn:time_limit_observation} in staggered adoption}\label{pf:asn56_staggered}

In staggered adoption treatment design, we denote for each $i$,
$$ \tau_i := \begin{cases}
\min\{t \in [T]: W_{i,t} = 1\}, & \\
\quad \infty & \hspace{-50pt}\text{ if } W_{i,t} = 0 \text{ for all } t \in [T]
\end{cases}$$
as the \textit{treatment adoption time.} Since the adoption time $\tau_i$'s are independent for the units, we denote the distribution function of the adoption time as $G_\tau(t) := \pr(\tau_i \le t) = \pr(W_{i,t} = 1)$. Thus,
$$ \frac{1}{N}\sum_{i = 1}^N W_{i,s} W_{i,t} = \frac{1}{N} \sum_{i = 1}^N \indic{\tau_i \le \min\{s, t\}} \asconverge G_\tau(\min\{s, t\}), $$

and similarly $\frac{1}{N}\sum_{i = 1}^N W_{i,s} W_{i,t} W_{i,s'} W_{i,t'} \asconverge G_\tau(\min\{s, t, s', t'\})$. Therefore \asn \ref{asn:observation} is satisfied with $\alpha_{s, t} = G_\tau(\min\{s, t\})$ and $\beta_{s, t, s', t'} = G_\tau(\min\{s, t, s', t'\})$. Furthermore for $s, t \in [T]$, $\beta_{s,T,t,T} = G_\tau(\min\{s,t\})$ and $\alpha_{t,T} = G_\tau(t)$, which implies

$$ \frac{1}{T^2} \sum_{s, t = 1}^T \frac{\beta_{s,T,t,T}}{\alpha_{s,T} \alpha_{t, T}} = \frac{1}{T^2} \sum_{s,t=1}^T \frac{1}{G_\tau(\max\{s,t\})} = \frac{1}{\sum_{t=1}^T (2t-1)} \sum_{t = 1}^T \frac{2t-1}{G_\tau(t)}. $$

Applying Stolz-Cesaro theorem \citep[Theorem 1.22]{muresan2009concrete}, we obtain $\frac{1}{T^2} \sum_{s, t = 1}^T \frac{\beta_{s,T,t,T}}{\alpha_{s,T} \alpha_{t, T}} \to 1$ as $T \to \infty$. Similarly
$$ \frac{1}{T^3} \sum_{s, t, s' = 1}^T \frac{\beta_{s,t,s',T}}{\alpha_{s,t} \alpha_{s', T}} \to 1,\quad \text{and}\quad \frac{1}{T^4} \sum_{s, t, s', t' = 1}^T \frac{\beta_{s,t,s',t'}}{\alpha_{s,t} \alpha_{s', t'}} \to 1 $$

as $T \to \infty$. Next we denote $L_T := \frac{1}{T}\sum_{t = 1}^T W_{i,t} F_t^2 = \frac{1}{T}\sum_{t = \tau_i}^T F_t^2$, and $\Fktsq := \frac{1}{T-k+1}\sum_{t = k}^T F_t^2$ for $k \le T$. Therefore, $L_T = \frac{T-\tau_i + 1}{T} \overline{F}^2_{\tau_i:T}$. Since $\varepsilon_t \sim \normal{0}{\sigma_\varepsilon^2}$, $\cov(F_t^2, F_s^2) = \Gamma\phi^{2|s-t|}$ where $\Gamma := 2\left(\frac{\sigma_\varepsilon^2}{1-\phi^2}\right)^2$. Thus,
\begin{align*}
    \var(\Fktsq) =\ & \frac{1}{(T-k+1)^2} \sum_{t = 1}^T \sum_{s = 1}^T \cov(F_t^2, F_s^2) \\
    =\ & \frac{1}{(T-k+1)^2}\sum_{h = -(T-k)}^{T-k} (T - k +1 - |h|)\Gamma\phi^{2|h|}\\
    \le\ & \frac{\Gamma}{T - k + 1}.
\end{align*}

Hence for any $M \in \NN$ and $\varepsilon > 0$, we apply Chebychev's inequality on $\Fktsq$ and get
$$ \sup_{1\le k \le M} \pr(|\Fktsq - \sigma_F^2| > \varepsilon) \le \frac{\Gamma}{\varepsilon^2(T-M+1)}, $$

By triangle inequality, $|L_T - \sigma_F^2| \le | \overline{F}^2_{\tau_i:T} - \sigma_F^2| + \sigma_F^2 \frac{\tau_i - 1}{T} $. Therefore on the set $\{\tau_i < \infty\}$,
\begin{align*}
    \pr(|L_T - \sigma_F^2| > \varepsilon) \le & \sum_{k=1}^{\lfloor\sqrt{T}\rfloor} \pr(|L_T - \sigma_F^2| > \varepsilon \mid  \tau_i = k) \pr(\tau_i = k) + \pr(\tau_i > \lfloor \sqrt{T} \rfloor)\\
    \le & \frac{\Gamma}{\left( \varepsilon - \sigma_F^2 \frac{\lfloor \sqrt{T} \rfloor - 1}{T}\right)^2 (T - \lfloor \sqrt{T} \rfloor + 1)} \sum_{k = 1}^{\infty} \pr(\tau_i = k) + \pr(\tau_i > \lfloor \sqrt{T} \rfloor)\\ 
    \to\ & 0 \text{ as } T \to \infty.
\end{align*}
Hence $L_T \pconverge \sigma_F^2$, and \asn \ref{asn:time_limit_observation} is satisfied. 

\subsection{Asymptotic variance for MCAR and staggered adoption}
The asymptotic variances in \eqref{eq:variance_est} and \eqref{eq:variance_for} become
\begin{align}
    \sest = & \begin{cases}
       \frac{\dnt^2 \phi^{2h}\sigma_\varepsilon^2}{p}\left[\frac{\Lambda_i^2}{N \sigma_\Lambda^2} + \frac{F_T^2}{T \sigma_F^2}\right] + \frac{\dnt^2 \phi^{2h}\Lambda_i^2 F_T^2}{N\sigma_\Lambda^2}\left( \frac{1}{p} - 1 \right) \EE\left[\left(\frac{\Lambda_i^2}{\sigma_\Lambda^2} - 1 \right)^2 \right] & \text{ for MCAR design},\\
        \dnt^2 \phi^{2h}\sigma_\varepsilon^2\left[\frac{\Lambda_i^2}{N \sigma_\Lambda^2} + \frac{F_T^2}{T \sigma_F^2}\right]& \hspace{-50pt}\text{ for staggered adoption design},
    \end{cases}\label{eq:variance_est1}\\
    \sfor = & \frac{\dnt^2}{T}~h^2 \phi^{2h-2}(1-\phi^2)\Lambda_i^2 F_T^2 \label{eq:variance_for1}.
\end{align}

Hence the proof is complete. \pfend

\begin{remark}
When the AR(1) factors are fully observed, $\sfor$ in \eqref{eq:variance_for1} captures the uncertainty of forecasting with information up to time~$T$ \citep[Sec.~3.4]{shumway2000time}
\end{remark}

\section{General assumptions}\label{sec:general_asn}

We now state the general set of assumptions on the factor model \ref{eq:factor_model_1}. The general assumptions impose general structure on the temporal and cross sectional dependence of the factors, loadings and the idiosyncratic errors. Under these assumptions, we prove Thm. \ref{thm:forecast_error_bound}. In all the following, $M, M'$ denote universal positive constants that is allowed to change values from one line to another.

\begin{assumption}[\textbf{Factor model with general structure}]\label{gasn:factor_model}
For every $i\in [N],t \in [T]$, $\EE[\|F_t\|_2^4] \le M$ and $\EE[\|\Lambda_i\|_2^4] \le M$. For fixed $0 \le h < T$, there exists a positive definite matrix $\Sigma_F^{(h)}$ such that 
\begin{equation}\label{eq:factor_second_moment}
   \EE\left[\left\|\sqrt{T-h}\left(\frac{1}{T-h}\sum_{t=h+1}^T F_t F_{t-h}^\top - \Sigma_{F}^{(h)}\right)\right\|^2\right] \le M. 
\end{equation}

There exists a positive definite matrix $\Sigma_\Lambda$ such that 
\begin{equation}\label{eq:loading_second_moment0}
  \EE\left[\left\|\sqrt{N}\left(\frac{1}{N}\sum_{i = 1}^N \Lambda_i \Lambda_i^\top - \Sigma_\Lambda\right)\right\|^2\right] \le M,  
\end{equation}

and for every $\cQ_{s,t}\subset[N]$,
\begin{equation}\label{eq:loading_second_moment1}
    \EE\left[\left\|\sqrt{|\cQ_{s,t}|}\left(\frac{1}{|\cQ_{s,t}|}\sum_{i\in \cQ_{s,t}}^N \Lambda_i \Lambda_i^\top - \Sigma_\Lambda\right)\right\|^2\right] \le M.
\end{equation}

The eigenvalues of $\Sigma_F^{(0)} \Sigma_\Lambda$ are distinct and strictly positive.
\end{assumption}

\vspace{10pt}
\begin{assumption}[\textbf{Idiosyncratic errors with general strucuture}]\label{gasn:idio}
For $i \in [N], t \in [T]$, $\EE[\varepsilon_{i,t}] = 0$, $\EE|\varepsilon_{i,t}|^8 \le M$. The autocvariance function $\gamma(s,t) \defeq \EE\left[\frac{1}{|\cQ_{s,t}|}\varepsilon_t^\top \varepsilon_s \right] = \EE\left[\frac{1}{|\cQ_{s,t}|}\sum_{i \in \cQ_{s,t}} {\varepsilon_{i,t}}{\varepsilon_{i,s}}\right]$, satisfies  $|\gamma(s, s)| \le M,\quad \sum_{s=1}^T |\gamma(s, t)| \le M$. For every $s,t \in [T]$, 
\begin{equation} \label{eq:idio_moment}
    \EE\left[\left|\frac{1}{\sqrt{|\cQ_{s,t}|}} \sum_{i \in \cQ_{s,t}}[\varepsilon_{i,s}\varepsilon_{i,t} - \EE(\varepsilon_{i,s}\varepsilon_{i,t})]\right|^4\right] \le M.
\end{equation}
\end{assumption}

\vspace{10pt}
\begin{assumption}[\textbf{Moment conditions}]\label{gasn:moment}
For every $s \in [T]$ and $0 \le h < T$,
\begin{align}
    & \EE\left[\left\| \frac{1}{\sqrt{|\cQ_{s,t}|}}\sum_{i\in \cQ_{s,t}} \Lambda_i \varepsilon_{i,t} \right\|_2^2\right] \le M, \label{eq:moment_0}\\
    & \EE\left[\left\|\frac{1}{\sqrt{T}} \sum_{t=1}^T F_{t-h} \frac{1}{\sqrt{|\cQ_{s,t}|}} \sum_{i \in \cQ_{s,t}} (\varepsilon_{i,t}\varepsilon_{i,s} - \EE[\varepsilon_{i,t}\varepsilon_{i,s}])\right\|^2\right] \le M, \label{eq:moment_1}\\
    & \EE\left[\left\|\frac{1}{\sqrt{T}} \sum_{t=1}^T \frac{1}{\sqrt{|\cQ_{s,t}|}} \sum_{i \in \cQ_{s,t}} \Lambda_i F_{t-h}^\top \varepsilon_{i,t} \right\|^2\right] \le M. \label{eq:moment_2}
\end{align}
\end{assumption}

\vspace{10pt}
\begin{assumption}[\textbf{Asymptotic distributions}]\label{gasn:limdist}
There exists a positive definite matrix $\Sigma_{F,i}$ such that \eqref{eq:lim_FFW} holds. In addition, there exists a positive definite matrix $\Gamma_T^\mathrm{obs}$ such that
\begin{equation}\label{eq:limdist1}
    \frac{\sqrt{N}}{T}\sum_{s=1}^T F_s F_s^\top\frac{1}{|\cQ_{s,T}|}\sum_{i \in \cQ_{s,T}} \Lambda_i \varepsilon_{i,T} \dconverge \normal{0}{\Gamma^\mathrm{obs}_T}.
\end{equation}

For every $i\in [N]$, there exists a positive definite matrix $\Phi_i^\mathrm{obs}$ such that
\begin{equation}\label{eq:limdist2}
    \frac{1}{\sqrt{T}}\sum_{t = 1}^T W_{i,t} F_t \varepsilon_{i,t} \dconverge \normal{0}{\Phi_i^\mathrm{obs}}.
\end{equation}

Denote the quantities $u_i \defeq \Sigma_F^{-1}\Sigma_\Lambda^{-1} \Lambda_i$, $J_t \defeq \frac{1}{T}\sum_{s=1}^T F_s F_s^\top \left( \frac{1}{|\cQ_{s,t}|}\sum_{i \in \cQ_{s,t}} \Lambda_i \Lambda_i^\top - \frac{1}{N}\sum_{i = 1}^N \Lambda_i \Lambda_i^\top \right)$, and $R_i \defeq \frac{1}{T}\sum_{t=1}^T W_{i,t} J_t F_t F_t^\top$. Then
\begin{equation}\label{eq:limdist3}
\sqrt{N}{\Sigma_{i,T}^{(J,R)}}^{-1/2}\begin{bmatrix}
J_T F_T\\ R_i u_i
\end{bmatrix} \xrightarrow{d} \normal{0} {I_r},
\end{equation}

where $\Sigma_{i,T}^{(J,R)} \defeq \begin{bmatrix}
    \Sigma_{J,T} & g_{i,T}^\mathrm{cov}(u_i)^\top\\
    g_{i,T}^\mathrm{cov}(u_i) & h_i(u_i)
\end{bmatrix}$ for positive definite matrix $\Sigma_{J,T}$, and functions $g_{i,T}^\mathrm{cov}$ and $h_i$.
\end{assumption}

Under the model assumptions in the paper, the general assumptions follow. We use the general assumptions to prove Thm. \ref{thm:forecast_error_bound}. The implication of the simplified assumptions by more general set of assumptions is stated as the following lemma:

\begin{lemma}[\textbf{Implication by simplified assumptions}] \label{lem:assumption_lemma}
    The general assumptions \ref{gasn:factor_model} to \ref{gasn:limdist} are satisfied by the simplified model assumptions \ref{asn:var1_factors} to \ref{asn:time_limit_observation}.
\end{lemma}

\subsection{Proof of Lemma \ref{lem:assumption_lemma}: Implication by simplified assumptions}

We state some sub-lemmata to validate that each of the general assumptions are satisfied under the simplified assumptions.

\begin{lemma}[\textbf{Power norm bounds under spectral radius condition}]\label{lem:A_sum}
Let $A \in \mathbb{R}^{r \times r}$ be a square matrix with spectral radius $\rho(A) < 1$. Then there exist an integer $\underline{N} \ge 1$, both depending on $A$, such that
$$
\|A^n\| < \left( \frac{1+\rho(A)}{2} \right)^n \quad \text{for all } n \ge \underline{N}.
$$
In particular, the series $\sum_{n=0}^\infty \|A^n\|$ and $\sum_{n=0}^\infty n\|A^n\|$ converge.
\end{lemma}

\begin{lemma}\label{lem:imply_1}
    \asn \ref{asn:var1_factors}, \ref{asn:factor} and \ref{asn:observation} $\implies$ \asn \ref{gasn:factor_model}.
\end{lemma}

\begin{lemma}\label{lem:imply_3}
    \asn \ref{asn:idiosyncratic} $\implies$ \asn \ref{gasn:idio}
\end{lemma}

\begin{lemma}\label{lem:imply_4}
    \asn \ref{asn:var1_factors} to \ref{asn:independence} $\implies$ \asn \ref{gasn:moment}
\end{lemma}

\begin{lemma}\label{lem:imply_5}
    \asn \ref{asn:var1_factors} to \ref{asn:time_limit_observation} $\implies$ \asn \ref{gasn:limdist}
\end{lemma}

Now we provide proofs of the stated lemmata.

\subsection{Proof of Lemma \ref{lem:A_sum}: Power norm bounds under spectral radius condition}
This lemma uses Gelfand's formula from matrix algebra \citep[Cor. 5.6.14]{horn2012matrix} that states the following: $\rho(A) = \lim_{n\rightarrow\infty} \|A^n\|^{1/n}$. Thus for any $\epsilon > 0$, there exists $N(\epsilon) \in \NN$ such that 
$$ \left| \|A^n\|^{1/n} - \rho(A) \right| < \epsilon~~ \text{ for } n > N(\epsilon) \quad \implies \quad \|A^n\| < (\rho(A) + \epsilon)^n ~~ \text{ for } n > N(\epsilon).$$
Choose $\epsilon = \epsilon(A) \triangleq (1 - \rho(A))/2$ and define $\tilde\rho(A) \triangleq (1 + \rho(A))/2 < 1$. We also denote $\underline{N}\triangleq N\left(\epsilon(A)\right)$. Therefore,
$$ \|A^n\| < \tilde\rho(A)^n~~~ \text{ for } n > \underline{N}. $$
Hence,
\begin{align*}
    \sum_{n=0}^\infty \|A^n\| = \sum_{n \le \underline{N}} \|A^n \| + \sum_{n \le \underline{N}} \|A^n \|
    < & \sum_{n \le \underline{N}} \|A^n \| + \sum_{n > \underline{N}} \tilde \rho(A)^n\\
    = & \sum_{n \le \underline{N}} \|A^n \| + \frac{\tilde\rho(A)^{\underline{N} + 1}}{1 - \tilde\rho(A)} < \infty,
\end{align*}

since the first term is a finite sum, and for the the second term $\tilde\rho(A) < 1$. Similarly,
$$ \sum_{n=\underline{N}}^\infty n \|A^n\| < \frac{\tilde\rho(A)^{\underline{N}}}{1 - \tilde \rho(A)} \left(\underline{N}(1 - \tilde \rho(A)) + \tilde\rho(A) \right) < \infty. $$

Hence the proof is done. \pfend

\subsection{\texorpdfstring{Proof of Lem. \ref{lem:imply_1}: \asn \ref{asn:var1_factors}, \ref{asn:factor} and \ref{asn:observation} $\implies$ \asn \ref{gasn:factor_model}}{Proof of Lem. \ref{lem:imply_1}: \asn \ref{asn:var1_factors}, \ref{asn:factor} and \ref{asn:observation} implies \asn \ref{gasn:factor_model}}} \label{pf:imply1}

From \asn \ref{asn:var1_factors}, $\rho(A) < 1$. Lem. \ref{lem:A_sum} implies that $F_t$ has the following moving average (MA) representation 
\begin{equation}\label{eq:ma_rep}
    F_t = \sum_{j=0}^\infty A^j \eta_{t-j}.
\end{equation}

Additionally we can bound the fourth moment by Minkowski's inequality as follows
$$ (\EE\|F_t\|_2^4)^{1/4} = \left( \EE \left[ \bigg \|\sum_{j=0}^\infty A^j \eta_{t-j}\bigg\|_2^4\right] \right)^{1/4} = \sum_{j=0}^\infty \left(\EE\left[\|A^j \eta_{t-j}\|_2^4\right]\right)^{1/4} \le \sum_{j=0}^\infty \left(\|A^j\| \EE\left[\|\eta_{t-j}\|_2^4\right]\right)^{1/4} \le M \sum_{j=0}^\infty \|A^j\|. $$

Lem. \ref{lem:A_sum} implies that the infinite sum converges. Hence there exists $M > 0$ such that $\EE[\|F_t\|_2^4] \le M$. \asn \ref{asn:factor} also poses $\EE[\|\Lambda_i\|_2^4] \le M$ for every $i\in [N]$ that is satisfied by \asn \ref{gasn:factor_model}.

Denote the autocovariance function of the factors $F_t$ as $\Gamma_F(h)\defeq \cov(F_t, F_{t-h})$. Then by \citet[Thm. 11.2.2 and Eq. 11.2.6]{brockwell1991time}, we have 
$$\frac{1}{T-h}\sum_{t=h+1}^T F_t F_{t-h}^\top = \Gamma_F(h) + \littleOP{T^{-1/2}},$$

Hence \eqref{eq:factor_second_moment} holds with  $\Sigma_F^{(h)} = \Gamma_F(h)$. As $\Lambda_i$ are $\iid$ in \asn \ref{asn:factor}, \eqref{eq:loading_second_moment0} holds. Additionally \eqref{eq:loading_second_moment1} holds under \eqref{eq:q_lower_bound} in \asn \ref{asn:observation} that completes the proof. \pfend

\subsection{\texorpdfstring{Proof of Lem. \ref{lem:imply_3}:  \asn \ref{asn:idiosyncratic} $\implies$ \asn \ref{gasn:idio}}{Proof of Lem. \ref{lem:imply_3}: \asn \ref{asn:idiosyncratic} implies \asn \ref{gasn:idio}}} \label{pf:F}

\asn \ref{asn:idiosyncratic} directly implies $\EE[\varepsilon_{it}] = 0$ and $\EE|\varepsilon_{it}|^8 \le M$. For $s\ne t$, $\gamma(s, t) = |\cQ_{s,t}|^{-1}\EE[\varepsilon_t^\top \varepsilon_s] = 0$. Furthermore,
$\gamma(s, s) = |\cQ_{s,s}|^{-1}\EE[\varepsilon_s^\top \varepsilon_s] = N^{-1}  \sum_{i=1}^N \EE[\varepsilon_{is}^2] = \sigma_\varepsilon^2 < \infty$ by \asn \ref{asn:idiosyncratic}. In addition,  $\sum_{s=1}^T \gamma(s, t) =\gamma(s,s) = \sigma_\varepsilon^2$ which is finite. For fixed $s, t \in [T]$, denote $v_i^{(s,t)} \defeq \left|\varepsilon_{i,s} \varepsilon_{i,t} - \EE(\varepsilon_{i,s}\varepsilon_{i,t})\right|$. Therefore $\EE[v_i^{(s,t)}] = 0$ and $\EE[({v_i^{(s,t)}})^4] < \infty$ by \asn \ref{asn:idiosyncratic}. We also have the following
\begin{align*}
    \EE\left[\left| \frac{1}{\sqrt{|\cQ_{s,t}|}}\sum_{i \in \cQ_{s,t}} [\varepsilon_{i,s} \varepsilon_{i,t} - \EE(\varepsilon_{i,s}\varepsilon_{i,t}) ] \right|^4\right]
    = &~~ \frac{1}{|\cQ_{s,t}|^2} \EE\left(\sum_{i \in \cQ_{s,t}} v_i^{(s,t)} \right)^4\\
    = &~~ \frac{1}{|\cQ_{s,t}|^2}\left( |\cQ_{s,t}| \EE\left[\left(v_1^{(s,t)}\right)^4 \right] + 3 |\cQ_{s,t}|(|\cQ_{s,t}| - 1)\EE\left[\left(v_1^{(s,t)}\right)^2\right]^2 \right)\\
    \le &~~ |\cQ_{s,t}|^{-1} M + 3(1 - |\cQ_{s,t}|^{-1}) M \le 3M.
\end{align*}
Hence the proof is done. \pfend

\subsection{\texorpdfstring{Proof of Lemma \ref{lem:imply_4}: \asn \ref{asn:var1_factors} to \ref{asn:independence} $\implies$ \asn \ref{gasn:moment}}{Proof of Lemma G: \asn \ref{asn:var1_factors} to \ref{asn:independence} implies \asn \ref{gasn:moment}}}\label{pf:imply_4}
Next we show that \eqref{eq:moment_0} holds. We have the following
\begin{align*}
    \EE\left[\left\| \frac{1}{\sqrt{|\cQ_{s,t}|}}\sum_{i\in \cQ_{s,t}} \Lambda_i \varepsilon_{i,t} \right\|_2^2\right] = &~~ \frac{1}{|\cQ_{s,t}|}\EE\left[\left(\sum_{i \in \cQ_{s,t}} \Lambda_i^\top \varepsilon_{i,t} \right) \left(\sum_{i \in \cQ_{s,t}} \Lambda_i \varepsilon_{i,t} \right) \right]\\
    = &~~ \frac{1}{|\cQ_{s,t}|}\EE\left[ \sum_{i,j\in \cQ_{s,t}} \Lambda_i^\top \Lambda_j \varepsilon_{i,t} \varepsilon_{j,t} \right]\\
    = &~~\frac{1}{|\cQ_{s,t}|}\cdot |\cQ_{s,t}|\ \EE[\|\Lambda_1\|_2^2] \EE[\varepsilon_{1,t}^2] ~ \le~ \sigma_\varepsilon^2 M,
\end{align*}

where we use \asn \ref{asn:factor} and \ref{asn:idiosyncratic} in the last step. Next we show that \eqref{eq:moment_1} holds. We additionally have
\begin{align}
    & \EE \left[\left\| \frac{1}{\sqrt{T}} \sum_{t=1}^T F_t \frac{1}{\sqrt{|\cQ_{s,t}|}} \sum_{i \in \cQ_{s,t}} (\varepsilon_{i,s}\varepsilon_{i,t} - \EE[\varepsilon_{i,s}\varepsilon_{i,t}]) \right\|_2^2\right] \nonumber \\
    = &\frac{1}{T} \EE\left[ \left( \sum_{t=1}^T F_t^\top \frac{1}{\sqrt{|\cQ_{s,t}|}} \sum_{i \in \cQ_{s,t}} v_i^{(s,t)} \right)\left( \sum_{t=1}^T F_t \frac{1}{\sqrt{|\cQ_{s,t}|}} \sum_{i \in \cQ_{s,t}} v_i^{(s,t)} \right) \right] \nonumber\\
    = & \frac{1}{T}\sum_{t,t'=1}^T \EE[F_t^\top F_{t'}]\cdot \frac{1}{\sqrt{|\cQ_{s,t}||\cQ_{s,t'}|}}~ \EE\left[\sum_{i \in \cQ_{s,t},~ j \in \cQ_{s,t'}} v_i^{(s,t)} v_i^{(s,t')} \right] \label{eq:cross_error_terms}
\end{align}

Under the general assumption \ref{gasn:factor_model}, there exists $M > 0$ such that $\EE[F_t^\top F_{t'}]\le \EE[\|F_t\|_2^2]\le M$. We denote 
$A_{s,t,t'} \defeq \EE[\sum_{i \in \cQ_{s,t}, j \in \cQ_{s,t'}} v_{i,s,t} v_{i,s,t'}]$. We consider the following cases:

\textbf{Case 1.} If $t = t' = s$,
\begin{align*}
 A_{s,t,t'} = & \EE\sum_{i,j \in \cQ_{ss}} (\varepsilon_{is}^2 - \sigma_\varepsilon^2)(\varepsilon_{js}^2 - \sigma_\varepsilon^2)\\
 = & \EE\left[ \sum_{i=1}^N (\varepsilon_{is}^2 - \sigma_\varepsilon^2) \right]^2 = \EE\left[ \sum_{i=1}^N (\varepsilon_{is}^2 - \sigma_\varepsilon^2)^2 \right] \le NM'
\end{align*}
for some $M' > 0$.

\textbf{Case 2.} If $t = t' \ne s$, 
$$A_{s,t,t'} = \sum_{i,j\in \cQ_{s,t}} \EE[\varepsilon_{i,s}\varepsilon_{j,s}] \EE[\varepsilon_{i,t}\varepsilon_{j,t}] = |\cQ_{s,t}| \sigma_\varepsilon^4. $$

\textbf{Case 3.} If $t=s,t'\ne s$, 
$$A_{s,t,t'} = \EE\left[ \sum_{1\le i \le N, j\in \cQ_{s,t'}} (\varepsilon_{i,s}^2 - \sigma_\varepsilon^2) \varepsilon_{j,s}\varepsilon_{j,t'} \right] = 0.$$
\textbf{Case 4.} If $t\ne s, t'\ne s$, 
$$A_{s,t,t'} = \EE\left[ \sum_{i\in \cQ_{s,t}, j \in \cQ_{s,t'}} \varepsilon_{is}\varepsilon_{i,t}\varepsilon_{j,s}\varepsilon_{j,t'} \right] = 0.$$

Combining all cases in \eqref{eq:cross_error_terms}, we obtain the following for some $M$ and $M' >0$:
\begin{align*}
 \EE \left[\left\| \frac{1}{\sqrt{T}} \sum_{t=1}^T F_t \frac{1}{\sqrt{|\cQ_{s,t}|}} \sum_{i \in \cQ_{s,t}} (\varepsilon_{i,s}\varepsilon_{i,t} - \EE[\varepsilon_{i,s}\varepsilon_{i,t}]) \right\|_2^2\right]
 \le &~~  \frac{1}{T} M\cdot \frac{1}{N}\cdot M'N + \frac{1}{T} M \cdot \sigma_\varepsilon^4 (T-1) \\
 \le &~~ \frac{1}{T} M [M' + \sigma_\varepsilon^4(T-1)]\\
 \le &~~ M [M' + \sigma_\varepsilon^4].
\end{align*}

Finally we show that \eqref{eq:moment_2} holds by the mutual independence \asn \ref{asn:independence} as follows:
\begin{align*}
    & \EE\left[\left\|\frac{1}{\sqrt{T}} \sum_{t=1}^T \frac{1}{\sqrt{|\cQ_{s,t}|}} \sum_{i \in \cQ_{s,t}} \Lambda_i F_t^\top \varepsilon_{i,t}\right\|^2\right] \\
    \le &~~ \EE\left[\left\| \frac{1}{\sqrt{T}} \sum_{t=1}^T \frac{1}{\sqrt{|\cQ_{s,t}|}} \sum_{i \in \cQ_{s,t}} \Lambda_i F_t^\top \varepsilon_{i,t} \right\|_F^2 \right] \\
    = &~~ \EE\left[ \text{Tr}\left\{ \left(\frac{1}{\sqrt{T}} \sum_{t=1}^T \frac{1}{\sqrt{|\cQ_{s,t}|}} \sum_{i \in \cQ_{s,t}} \Lambda_i F_t^\top \varepsilon_{i,t}\right) \left(\frac{1}{\sqrt{T}} \sum_{t=1}^T \frac{1}{\sqrt{|\cQ_{s,t}|}} \sum_{i \in \cQ_{s,t}} F_t \Lambda_i^\top \varepsilon_{i,t}\right) \right\} \right]\\
    = &~~ \text{Tr}\left\{ \frac{1}{T} \sum_{t, t'=1}^T \frac{1}{\sqrt{|\cQ_{s,t}||\cQ_{s,t'}|}} \sum_{i \in \cQ_{s,t}, j \in \cQ_{s,t'}} \EE[\|F_t\|_2^2] \Sigma_\Lambda \EE(\varepsilon_{i,t}\varepsilon_{j,t'}) \right\}\\
    \le &~~ \frac{1}{T}\times T \times \text{Tr}(\Sigma_\Lambda)\EE[\|F_t\|_2^2] \sigma_\varepsilon^2\ \le\ r^2\|\Sigma_\Lambda\|~ \EE[\|F_t\|_2^2]~ \sigma_\varepsilon^2.
\end{align*}

By \asn \ref{asn:factor} and \ref{asn:idiosyncratic}, there exists $M > 0$ such that $\|\Sigma_\Lambda\| \le M$, and $\sigma_\varepsilon^2 \le M$. Additionally, \asn \ref{gasn:factor_model} implies that $\EE[\|F_t\|_2^2] \le M$. Since $r$ is fixed, we are done with the proof. \pfend

\subsection{\texorpdfstring{Proof of Lemma \ref{lem:imply_5}:  \asn \ref{asn:var1_factors} to \ref{asn:time_limit_observation} $\implies$ \asn \ref{gasn:limdist}}{Proof of Lemma \ref{lem:imply_5}:  \asn \ref{asn:var1_factors} to \ref{asn:time_limit_observation} implies \asn \ref{gasn:limdist}}}\label{pf:imply_5}

\eqref{eq:lim_FFW} in \asn \ref{asn:time_limit_observation} is already implied in \asn \ref{gasn:limdist}. We now show the asymptotic distributions in \asn \ref{gasn:limdist} hold under the simple conditions.

\subsubsection{Derivation of \eqref{eq:limdist1}}

First we denote the sigma-algebra ${\cal F}_T \defeq \sigma\left(\{F_t\}_{t\in [T]}, \{\alpha_{t,T}\}_{t\in [T]}, \{\beta_{s,T,t,T}\}_{s,t\in [T]} \right)$  where $\alpha_{s,t}$ and $\beta_{s,t,s',t'}$ are defined for $s,t,s',t'\in [T]$ in \asn \ref{asn:observation}. For $s\in [T]$, 
$$Z_{s,T} \defeq \frac{1}{\sqrt{|\cQ_{s,T}|}}\sum_{i \in \cQ_{s,T}} \Lambda_i \varepsilon_{i,T} = \frac{1}{\sqrt{|\cQ_{s,T}|}} \sum_{i=1}^N W_{i,s} W_{i,T} \Lambda_i \varepsilon_{i,T}.$$ 
Using \textit{Lindeberg's central limit theorem}, we show asymptotic normality of $Z_{s,T}$ stably on ${\cal F}_T$ as $N\to \infty$. We verify that the conditions of Lindeberg's CLT are satisfied:

\begin{enumerate}
\item \textit{Zero mean.}~ Using independence of $\Lambda_i$'s and $\varepsilon_{i,T}$'s in \ref{asn:independence},
\begin{equation}\label{eq:Z_zero_exp}
    \EE[Z_{s,T}\mid {\cal F}_T] = \sum_{i\in \cQ_{s,T}} \EE\left[\frac{1}{\sqrt{|\cQ_{s,T}|}}\Lambda_i\varepsilon_{i,T} \bigg| {\cal F}_T\right] = \sum_{i\in \cQ_{s,T}} \frac{1}{\sqrt{|\cQ_{s,T}|}}\EE[\Lambda_i] \EE[\varepsilon_{i,T}] = 0.
\end{equation}
\item \textit{Bounded covariance matrix.}~ The covariance matrix is 
$$\cov(Z_{s,T}\mid {\cal F}_T) = \sum_{i \in \cQ_{s,T}}\cov\left(\frac{1}{\sqrt{|\cQ_{s,T}|}}\Lambda_i\varepsilon_{i,T}\bigg| {\cal F}_T\right) = \cov(\Lambda_i \varepsilon_{i,T}) = \sigma_\varepsilon^2 \Sigma_\Lambda, $$ 
which is bounded with $T$.
\item \textit{Lindeberg's condition.}~ Under \asn \ref{asn:observation} and for any $\delta > 0$, 
$$\indic{\frac{1}{\sqrt{|\cQ_{s,T}|}} W_{i,s} W_{i,T} \|\Lambda_i \varepsilon_{i,T}\|_2 > \delta } \le \indic{\|\Lambda_i \varepsilon_{i,T}\|_2 > \delta\sqrt{\underline{q} N}}.$$

In addition, there exists $M > 0$ such that
$$ \sum_{i=1}^N \pr\left(\|\Lambda_i \varepsilon_{i,T}\|_2 > \delta\sqrt{\underline{q} N}\right) \le \sum_{i=1}^N \frac{\EE[\|\Lambda_i \varepsilon_{i,T}\|_2^2]}{\delta^2 \underline{q} N} = \frac{\EE[\|\Lambda_i\|_2^2 \sigma_\varepsilon^2]}{\delta^2\underline{q}^2} < \infty, $$
where the first inequality follows from \textit{Markov inequality}. Therefore, \textit{Borel-Cantelli Lemma} implies 
$$\indic{\frac{1}{\sqrt{|\cQ_{s,T}|}}W_{i,s} W_{i,T}\|\Lambda_i \varepsilon_{i,T}\|_2 > \delta} \xrightarrow{\text{a.s}} 0.$$ 

Further using \textit{Dominated Convergence Theorem}, we obtain the following
$$ \sum_{i=1}^N \EE\left[ \frac{1}{|\cQ_{s,T}|} W_{i,s} W_{i,T} \|\Lambda_i \varepsilon_{i,T}\|^2 \indicator\left\{\frac{1}{\sqrt{|\cQ_{s,T}|}}W_{i,s} W_{i,T}\|\Lambda_i \varepsilon_{i,T}\|_2 > \delta \right\} \right] \rightarrow 0. $$
\end{enumerate}

Therefore, Lindeberg's CLT holds and using \asn \ref{asn:observation} and \ref{asn:time_limit_observation} we have the following for $N\rightarrow \infty$:
$$Z_{s,T} \xrightarrow{d} \normal{0}{\sigma_\varepsilon^2 \Sigma_\Lambda},\quad \sqrt{{\frac{N}{|\cQ_{s,T}|}}}Z_{s,T} \xrightarrow{d} \normal{0}{\frac{1}{\alpha_{s,T}}\sigma_\varepsilon^2\Sigma_{\Lambda}} \quad \text{stably on } {\cal F}_T,$$
and the vector $\left( \frac{\sqrt{N}}{|\cQ_{1,T}|}\sum_{i \in \cQ_{1,T}} \Lambda_i^\top \varepsilon_{i,T},\ldots, \frac{\sqrt{N}}{|\cQ_{T,T}|}\sum_{i \in \cQ_{T,T}} \Lambda_i^\top \varepsilon_{i,T} \right)^\top$
is jointly asymptotically normal stably on ${\cal F}_T$. For each $s, s'\in [T]$, the asymptotic covariance matrix of $\frac{1}{T}\sum_{s=1}^T F_s F_s^\top Z_{s,T}$ consists of the following terms

\begin{align*}
    \EE\left[ \sqrt{\frac{N}{|\cQ_{s,T}|}} Z_{s,T} \cdot \sqrt{\frac{N}{|\cQ_{s'T}|}} Z_{s'T}^\top \bigg| {\cal F}_T \right] = & \frac{N}{|\cQ_{s,T}||\cQ_{s',T}|}\sum_{i \in \cQ_{s',T}}\sum_{j \in \cQ_{s',T}} \EE\left[ \Lambda_i \varepsilon_{i,T} \Lambda_j^\top \varepsilon_{j,T} \right]\\
    = & \frac{N|\cQ_{s,T,s',T}|}{|\cQ_{s,T}| |\cQ_{s',T}|} \sum_{i \in \cQ_{s,T,s',T}} \EE[\Lambda_i \Lambda_i^\top] \EE[\varepsilon_{i,T}^2] \\
    = & U^{(s,s')} \sigma_\varepsilon^2 \Sigma_\Lambda,
\end{align*}
where $U^{(s,s')} \defeq \frac{N|\cQ_{s,T,s',T}|}{|\cQ_{s,T}| |\cQ_{s',T}|}$. From \asn \ref{asn:observation}, as $N\to \infty$
$$U^{(s,s')} \xrightarrow{\mathrm{a.s.}} \frac{\beta_{s,T, s',T}}{\alpha_{s,T} \alpha_{s',T}} \le \frac{1}{\alpha_{s,T}} \le \frac{1}{\underline{q}} < \infty,$$
and using \asn \ref{asn:time_limit_observation}
$$\frac{1}{T^2}\sum_{s = 1}^T\sum_{s'=1}^T U^{(s,s')}~~ \xrightarrow[N\to \infty]{\mathrm{a.s.}}~~ \frac{1}{T^2}\sum_{s = 1}^T\sum_{s'=1}^T \frac{\beta_{s,T, s',T}}{\alpha_{s,T} \alpha_{s',T}} ~~\xrightarrow[T\to \infty]{\pr}~~ \omega_1.  $$

Therefore,
\begin{align*}
    \cov\left( \frac{1}{T}\sum_{s=1}^T F_s F_s^\top \sqrt{\frac{N}{|\cQ_{s,T}|}} Z_{s,T} \bigg| {\cal F}_T \right) = \frac{1}{T^2}\sum_{s = 1}^T\sum_{s' = 1}^T \sigma_\varepsilon^2 U^{(s,s')} F_s F_s^\top \Sigma_\Lambda F_{s'} F_{s'}^\top
    \xrightarrow{\mathrm{a.s.}} {\cal D}_T
\end{align*}
as $N \to \infty$, where 
$$ {\cal D}_T \defeq \frac{\sigma_\varepsilon^2}{T^2}\sum_{s = 1}^T\sum_{s' = 1}^T \frac{\beta_{s,T, s',T}}{\alpha_{s,T} \alpha_{s',T}} F_s F_s^\top \Sigma_\Lambda F_{s'} F_{s'}^\top. $$

In order to show unconditional CLT of $\frac{1}{T}\sum_{s=1}^T F_s F_s^\top Z_{s,T}$, we invoke \textit{Chebychev's inequality} and show that ${\cal D}_T$ converges in probability. We denote
$V^{(s,s')} \defeq \text{vec}(F_s F_s^\top \Sigma_\Lambda F_{s'} F_{s'}^\top)$. Therefore for any $m \in [r^2]$,
\begin{equation}\label{eq:chebychev_factor}
\EE\left[\frac{1}{T^2}\sum_{s=1}^T \sum_{s'=1}^T U^{(s,s')}\left(V^{(s,s')}_{m} - \EE[V^{(s,s')}_{m}]\right) \right]^2
= \frac{1}{T^4} \sum_{s,s',t,t'} U^{(s,s')} U^{(t,t')} \cov\left(V^{(s,s')}_m, V^{(t,t')}_m\right),
\end{equation}

We handle the sum in \eqref{eq:chebychev_factor} case by case. When at least one of the two equalities $s = t$ and $s' = t'$ holds, the sum has at most $\bigO{T^3}$ terms with finite second moments. When $s \neq s', t\neq t'$, there exists $(u,v)\in [r]\times [r]$ depending on $m$ such that
\begin{align}
    \cov\left(V_m^{(s,s')}, V_m^{(t,t')}\right) = & \cov\left( F_{s,u} F_s^\top \Sigma_\Lambda F_t F_{t,v}, F_{s',u} F_{s'}^\top \Sigma_\Lambda F_{t'} F_{t',v}  \right)\nonumber \\
    = & \sum_{a,b,a',b'} (\Sigma_{\Lambda})_{a,b} (\Sigma_{\Lambda})_{a',b'} \cov\left( F_{s,u} F_{s,a} F_{t,b} F_{t,v}, F_{s',u} F_{s',a'} F_{t',b'} F_{t',v} \right). \label{eq:8moment}
\end{align}

By \asn \ref{asn:idiosyncratic} the errors $\eta_t$ are Gaussian, and by \asn \ref{asn:var1_factors} $F_t$ is a VAR(1) process. Therefore $\|\EE[F_s F_t^\top]\| = \bigO{\|A^{|s-t|}\|}$ and for $s\ne t$, $s'\ne t'$, we have $\cov\left(V_m^{(s,s')}, V_m^{(t,t')}\right) = \bigO{\|A^h\|}$, where $h \defeq \min\{|s-s'|, |t-t'|, |s-t'|, |s'-t|\}$. Fixing $s,t$ and summing \eqref{eq:8moment} we obtain
$$ \sum_{s',t': s\ne t, s'\ne t'} \cov\left(V_m^{(s,s')}, V_m^{(t,t')}\right) = \sum_{h=1}^T \bigO{Th} \bigO{\|A^h\|} < \bigO{T} \sum_{h = 1}^\infty h\|A^h\| = \bigO{T}, $$

where the last step follows from Lem. \ref{lem:A_sum}. Hence \eqref{eq:chebychev_factor} can be bounded by $T^{-4} \bigO{T^3}= \bigO{T^{-1}}$. Hence we apply Chebychev's inequality to obtain
${\cal D}_T \pconverge \omega_1 \sigma_\varepsilon^2 \Sigma_F \Sigma_\Lambda \Sigma_F$ as $T \to \infty$. Using tower property of conditional expectation and \eqref{eq:Z_zero_exp},
$$ \lim_{N, T\to \infty} \cov\left( \frac{1}{T}\sum_{s=1}^T F_s F_s^\top \sqrt{\frac{N}{|\cQ_{s,T}|}} Z_{s,T} \right) = \omega_1 \sigma_\varepsilon^2\Sigma_F \Sigma_{\Lambda} \Sigma_{F}. $$

Therefore,
$$ \frac{1}{T}\sum_{s=1}^T F_s F_s^\top \sqrt{\frac{N}{|\cQ_{sT}|}} Z_{sT} \xrightarrow{d} \normal{0}{\omega_1 \sigma_\varepsilon^2\Sigma_F \Sigma_{\Lambda} \Sigma_{F}}. $$

\subsubsection{Derivation of \eqref{eq:limdist2}}

Denote $Z_{i,t} \defeq W_{i,t} F_t$. By \asn \ref{asn:time_limit_observation}, $\frac{1}{T}\sum_{t=1}^T W_{i,t} F_t F_t^\top = \frac{1}{T}\sum_{t=1}^T Z_{i,t} Z_{i,t}^\top \pconverge \Sigma_{F,i}$. Additionally we denote, $U_{i,t,T} \defeq \frac{1}{\sqrt{T}} Z_{i,t} \varepsilon_{i,t}$, and the filtration, $\mathcal{F}_{i,t} \defeq \sigma(\{A_{i,s}\}_{s\in [t+1]}, \{\varepsilon_{i,s}\}_{s\in [t]})$. Therefore,
$$ \EE[U_{i,t,T} \mid \mathcal{F}_{i,t-1}] = \frac{1}{\sqrt{T}}Z_{i,t} \EE[\varepsilon_{i,t} \mid \mathcal{F}_{i,t-1}] = 0,  $$

since $\varepsilon_{i,t}$ is independent of the history until time $t-1$ and $Z_{i,t}$, and it has mean zero. Therefore $(U_{i,t,T}, \mathcal{F}_{i,t,T})_{t\in [T]}$ is a Martingale difference array in $\RR^r$. Additionally,

$$ \sum_{t=1}^T \EE[U_{i,t,T}~ U_{i,t,T}^\top \mid \mathcal{F}_{i,t-1}] = \sum_{t=1}^T \frac{1}{T} \sigma_{\varepsilon}^2 Z_{i,t} Z_{i,t}^\top \pconverge \sigma_{\varepsilon}^2 \Sigma_{F,i}. $$

Hence adapted to ${\cal F}_{i,t-1}$, the sum of $U_{i,t,T}$'s has finite second moment. Also, $\|U_{i,t,T}\|_2 = T^{-1/2} |\varepsilon_{i,t}|\cdot\|Z_{i,t}\|_2$. Therefore for any $\delta > 0$,
\begin{align*}
    & \sum_{t=1}^T \EE\left[\|U_{i,t,T}\|_2^2~ \indic{\|U_{i,t,T}\|_2 > \delta} \mid \mathcal{F}_{i,t-1}\right] \\ 
    \le~~ & \sum_{t=1}^T \frac{1}{T}\EE\left[|\varepsilon_{i,t}|^2 \|Z_{i,t}\|_2^2~ \indic{|\varepsilon_{i,t}|\cdot \|Z_{i,t}\|_2 \ge \sqrt{T}\delta }\right]\\
    \le~~ & \left( \EE[|\varepsilon_{i,t}|^4] \EE[\|Z_{i,t}\|_2^4] \right)^{1/2} ~~\left(\pr\left(|\varepsilon_{i,t}|\cdot \|Z_{i,t}\|_2 \ge \sqrt{T}\delta\right)\right)^{1/2},
\end{align*}

where the last step follows from Cauchy-Schwarz inequality. From \asn \ref{asn:var1_factors}, \ref{asn:factor} and \ref{asn:idiosyncratic}, $\varepsilon_{i,t}$ and $Z_{i,t}$ have finite fourth moments and are bounded in probability. Thus, 
$$\pr\left(|\varepsilon_{i,t}|\cdot \|Z_{i,t}\|_2 \ge \sqrt{T}\delta\right) \rightarrow 0, ~~~ \text{as }T\rightarrow \infty.$$ 

Therefore the \textit{Lindeberg condition} is satisfied. Using \textit{Martingale central limit theorem} \citep[Thm. 3.2]{hall2014martingale}, \ref{eq:limdist2} holds with $\Phi_i^\mathrm{obs} = \sigma_\varepsilon^2 \Sigma_{F,i}$.

\subsubsection{Derivation of \eqref{eq:limdist3}}

The proof proceeds analogously to step 5 in the proof of Proposition 3 in \citet{xiong2023large}. We denote
\begin{align*}
    V_{s,T} \defeq \frac{1}{|\cQ_{s,T}|}\sum_{i \in \cQ_{s,T}} \Lambda_i \Lambda_i^\top - \frac{1}{N}\sum_{i=1}^N \Lambda_i \Lambda_i^\top = & \frac{1}{\sqrt{N}}\Bigg[\left( \sqrt{\frac{N}{|\cQ_{s,T}|}} - \sqrt{\frac{|\cQ_{s,T}|}{N}} \right)\frac{1}{\sqrt{|\cQ_{s,T}|}}\sum_{i \in \cQ_{s,T}} \Lambda_i  \Lambda_i^\top\\
    & \hspace{30pt} - \sqrt{\frac{N - |\cQ_{s,T}|}{N}}\frac{1}{\sqrt{N - |\cQ_{s,T}|}}\sum_{i \in \cQ_{s,T}^c} \Lambda_i \Lambda_i^\top\Bigg].
\end{align*}

Therefore invoking \asn \ref{asn:factor} and \ref{asn:observation}, $\sqrt{N}\text{vec}(V_{s,T}) \xrightarrow{d} \left( \frac{1}{\sqrt{\gamma_s}} - \sqrt{\gamma_s} \right) Z_1 - \sqrt{1-\gamma_s}Z_2$, where $Z_1, Z_2\sim \normal{0}{\Theta_\Lambda}$ independent, with $\Theta_\Lambda = \EE[\text{vec}(\Lambda_i \Lambda_i^\top-\Sigma_\Lambda) (\text{vec}(\Lambda_i \Lambda_i^\top - \Sigma_\Lambda))^\top]$. Simplifying the variance expression,
$$ \sqrt{N}\text{vec}(V_{s,T}) \xrightarrow{d} \normal{0}{\left(\frac{1}{\gamma_s}-1\right)\Theta_{\Lambda}}. $$
Similar to \citet[Prop. 3]{xiong2023large}, the covariance term for $s,s'\in [T]$ is
$$ \lim_{T,N\rightarrow\infty}\cov(\sqrt{N}\text{vec}(V_{s,T}), \sqrt{N}\text{vec}(V_{s',T})) = \left( \frac{\gamma_{s,s'}}{\gamma_s \gamma_{s'}} - 1 \right)\Theta_{\Lambda}. $$

\eqref{eq:limdist3} consists of $J_TF_T$ where $J_T = \frac{1}{T}\sum_{s=1}^T F_s F_s^\top V_{s,T}$. Following the route of \citet[Prop. 3.1(b), step 5.2]{xiong2023large}, and invoking \asn \ref{asn:time_limit_observation} the sum has a limiting covariance term
$$\lim_{T, N \to \infty} \cov(\sqrt{N}\text{vec}(J_T)) = (\omega_1-1)(I_r \otimes \Sigma_F) \Theta_\Lambda (I_r \otimes \Sigma_F).$$

Using the stable convergence in law \citep[Thm. 6.1]{hausler2015stable}, $\sqrt{N} \Sigma_{J,T}^{-1/2} J_T F_T \xrightarrow{d} \normal{0}{I_r}$ where $\Sigma_{J,T} = (\omega_1-1)(F_T^\top \otimes \Sigma_F)\Theta_\Lambda (F_T\otimes \Sigma_F)$. Similarly the proof for asymptotic normality of $R_i$ is similar to \citet[Prop. 3.1(b), step 5.4]{xiong2023large} with the asymptotic covariance of $R_i u_i$ being $h_i(u_i)= (\omega_3 - 1) (\Sigma_{F,i} \otimes \Sigma_F) \Theta_{\Lambda}(\Sigma_{F,i} \otimes \Sigma_{F})$. Additionally the asymptotic covariance term between $J_T F_T$ and $R_i u_i$ being $g_{i,T}^\text{cov}(u_i)^\top = (\omega_2 - 1)(\Sigma_{F,i}\otimes \Sigma_F) \Theta_{\Lambda} (I_r\otimes \Sigma_F)$. Thus the proof is complete. \pfend

\section{Deferred details on the construction and identifiability of the forecast estimand}\label{sec:identifiability}

Under the VAR representation Assum. \ref{asn:var1_factors}, we express the factor at horizon $h$ as $$F_{T+h} = A^h F_T + \sum_{j=1}^h A^{h-j}\eta_{T+j}.$$
We denote $\mathcal{F}_T = \sigma(\{F_t : t \le T\})$ as the filtration encoded by the factor history. Conditioned on this filtration, $\EE[F_{T+h}\mid \mathcal{F}_T] = A^h F_T$ that is the best linear predictor of $F_{T+h}$ among all $\cal{F}_T$ measurable functions.

Now under \asn \ref{asn:idiosyncratic} and \ref{asn:independence},
\begin{align*}
    & \EE[Y_{i,T+h} \mid \sigma(\{\varepsilon_{i,t}: i \in [N], t\in [T]\} \cup \{\Lambda_i:i \in [N]\} \cup \cal{F}_T)]\\
    = &\EE[\Lambda_i^\top F_{T+h} + \varepsilon_{i,T+h} \mid  \sigma(\{\varepsilon_{i,t}: i \in [N], t\in [T]\} \cup \{\Lambda_i:i \in [N]\} \cup \cal{F}_T)]\\
    = &\EE[\Lambda_i^\top \EE[ F_{T+h}\mid \mathcal{F}_T ] \mid \{\Lambda_i: i \in [N]\}]\\
    = & \Lambda_i^\top A^h F_T = \thetaith,
\end{align*}
which matches the forecast estimand $\thetaith$ in \eqref{eq:forecast_target}, and is the best linear predictor of $Y_{i,T+h}$ conditioned on the latent dynamics until time $T$. In out algorithm \focus, we leverage the estimation procedure of the latent components to estimate $\thetaith$.

\paragraph{Identifiability under rotation} In \eqref{eq:factor_model_1}, factors and loadings are identifiable only up to a non-singular rotation. For any invertible $H$, the rotated representation yields the factors $G_t = H F_t$, loadings $(H^\top)^{-1}\Lambda_i$, and coefficient matrix $HAH^{-1}$. The forecast target remains invariant, since $\theta_{i,T:T+h}^{(H)} = \big((H^\top)^{-1}\Lambda_i\big)^\top (HAH^{-1})^h G_T = \thetaith$. Hence identification assumption determines $H$ but not affecting the forecast target.

\subsection{Construction of forecast estimand under moving average dynamics}
The forecast estimand can be calculated similarly for more complex factor dynamics. We illustrate with a zero mean VARMA(1, 1) process, which is a VAR(1) process augmented with a first order moving average (MA) component, as
$$ F_t = AF_{t-1} + \eta_t + B \eta_{t-1}, $$
where $B \in \RR^{r\times r}$ is the MA coefficient matrix with $\rho(B) < 1$. For $h \ge 1$,
$$ \EE[F_{T+1}\mid {\cal F}_T] = A^hF_T + A^{h-1}B \eta_T^*, $$
where $\eta_t^* \defeq \EE[\eta_t \mid {\cal F}_t]$ and can be recursively calculated as 
$$\eta_t^* = F_t - AF_{t-1} - \eta_{t-1}^*.$$ 
For practical purposes, we can estimate $\eta_0^*$ as $\tilde\eta_0 = 0$. Convexity of the conditional expectation and an algebra with the recursion yields a sequence of the factor innovations $\tilde \eta_t$ satisfying
\begin{equation}\label{eq:eta_diff}
    \EE[\|\eta_t^* - \tilde\eta_t\|_2] = \EE[\|(-B)^t \EE[\eta_0 \mid {\cal F}_T]\|_2] \le C \left(\frac{1+\rho(B)}{2}\right)^t \EE\|\eta_0\|_2,
\end{equation}
where we use Lem. \ref{lem:A_sum}. Hence for large enough $T$, $\EE[\|\eta_T^* - \tilde \eta_T\|_2] \to 0$ in \eqref{eq:eta_diff} and we can work with the forecast estimand
$$ \thetaith \defeq A^hF_T + A^{h-1}B \tilde \eta_T, ~~ \text{where } \tilde \eta_t = F_t - AF_{t-1} - \tilde \eta_{t-1}~ \text{with } \tilde \eta_0 = 0. $$
Similarly we can write the expressions of the forecast estimand for higher order VARMA processes that we also use to calculate the population target of DGP-3 in Sec. \ref{subsec:simulation}. The identifiability of the forecast estimand can also be ensured similar to the VAR(1) dynamics.

\section{Estimation of individual treatment effects using \focus}\label{sec:ite_estimation}
Factor model structure on $Y(1)$ and $Y(0)$, or additional structural assumption on the treatment effects enables forecasting appropriately defined individual future treatment effects, and average treatment effects. For example, the factor model 
\begin{equation}\label{eq:fixed_effects}
    Y_{i,t} = \tau_{i,t} W_{i,t} + \Lambda_i^\top F_t + \varepsilon_{i,t}.
\end{equation}
is a special case of an interactive fixed effects model \citep{bai2021matrix}, denoising the covarite effects. Two-way fixed effects model \citep{de2020two}, widely used in difference-in-differences estimator \citep{goodman2021difference, arkhangelsky2021synthetic} and similar econometric applications, is a special case of \eqref{eq:fixed_effects}. Forecasting requires structure on $\tau_{i,t}$, or equivalently on $Y(1)$ and $Y(0)$. Under fixed-effects model \eqref{eq:fixed_effects} or factor structure on $\tau_{i,t}$, 
$$Y_{i,t}(w) = \Lambda_i^\top(w) F_t(w) + \varepsilon_{i,t}$$
for $w = 0,1$ and the ranks of factor models satisfy $r(0) \le r(1)$. Conditioning on all latent variables up to time $T$, and with VAR(1) dynamics on factor process $F_t(w)$, the mean individual effect conditioned on information until time $T$ is
\begin{align*}
    & E[\tau_{i,T+h}\mid F(w), \Lambda(w), w = 0,1]\\
    = & E[Y_{i,T+h}(1) - Y_{i,T+h}(0) \mid F(w), \Lambda(w), w = 0,1]\\
    = & \Lambda_i(1)^\top A(1)^h F_T(1) - \Lambda_i(0)^\top A(0)^h F_T(0).
\end{align*}
We can now use \focus to separately estimate $\Lambda(w), F(w), A(w)$ for $w = 0,1$, and forecast the treatment effects. Under more complex factor dynamics, we can inherit the target treatment effect as the forecast estimand construction in Sec. \ref{sec:identifiability}.

\section{An almost sure guarantee of \asn \ref{asn:observation} under MCAR}\label{sec:mcar_as}

In this section we show that the condition \eqref{eq:q_lower_bound} in \asn \ref{asn:observation} regarding the observation $W$ holds with high probability when the observations are missing completely at random.

\begin{lemma}[MCAR]
    Suppose $W$ has MCAR observations. Then for $T = \bigO{e^N}$ and $N \to \infty$, \eqref{eq:q_lower_bound} holds almost surely.
\end{lemma}

\begin{proof}

We have $|\cQ_{s,t}|/N = \frac{1}{N}\sum_{i = 1}^N W_{i,s} W_{i,t}$ that is a sum of independent $\text{Bernoulli}(p^2)$ random variables. Hence we can apply Hoeffding's inequality \citep[Thm. 2.2.5]{vershynin2010introduction} for $\underline{q} \in (0, p^2) $ and $s, t \in [T]$ as follows
$$ \pr\left(|\cQ_{s,t}| < N\underline{q} \right) = \pr\left( \frac{|\cQ_{s, t}|}{N} < \underline{q} \right) \le \pr\left( \left|\frac{|\cQ_{s, t}|}{N} - p^2\right| > p^2 - \underline{q} \right) \le 2 e^{-2N(p^2 - \underline{q})^2}. $$

Union bound over $s,t \in [T]$ implies
\begin{equation}\label{eq:union_bound}
    \pr\left(|\cQ_{s,t}| < N\underline{q} \text{ for some } s,t \in [T] \right) \le 2T^2 e^{-2N(p^2 - \underline{q})^2} = 2 e^{-2[N(p^2 - \underline{q})^2 - \log T]}.
\end{equation}

Hence in the regime $T = \bigO{e^N}$, we have $\log T \le C N$ for some $C > 0$. Consequently,
$$ \sum_{N = 1}^\infty \pr\left(|\cQ_{s,t}| < N\underline{q} \text{ for some } s,t \in [T] \right) \le \sum_{N=1}^N e^{-2N[(p^2 - \underline{q})^2 - C]} < \infty. $$

Hence by first Borel-Cantelli lemma, $\pr(|\cQ_{s,t}| < N\underline{q} \text{ infinitely often for some } s, t ) = 0$. Equivalently, $|\cQ_{s,t}| \ge N\underline{q}$ almost surely as $N \to \infty$. Hence the proof is complete.
\end{proof}

\section{Asymptotic normality under unit root}\label{sec:unit_root}

We can extend the asymptotic guarantees of Thm. \ref{thm:asymptotic_normality} when $\rho(A)=1$, i.e. the factors have common stochastic trends \citep{phillips1988testing, bai2004estimating}. For clarity, we illustrate the argument for a random walk, i.e. $A=I_r$, and no missing entry. \asn \ref{gasn:factor_model} is then replaced by the law of iterted logarithm-type condition \citep[\asn A.2]{bai2004estimating} 
$$\liminf_{T\to\infty}\frac{\log\log T}{T^2}\sum_{t=1}^T F_tF_{t-h}^\top = D^{(h)}$$
for $h=0,1$ and positive definite $D^{(h)}$. Using this condition, we can show similar to Lem. \ref{lem:imply_2mp_autocov_rate} and \ref{lem:coeff},
$$\tilde\Gamma(h) = \bigOP{\frac{1}{\log\log T}} + \bigOP{\frac{T}{\sqrt{N}\log\log T}}.$$
Thus, if $N(\log\log T)^2/T\to 0$, the OLS estimator $\hat A=\hat\Gamma(1)\hat\Gamma(0)^{-1}$ is super-consistent \citep{phillips1988testing}:
$$T(\hat A - I_r) \overset{d}{\to} \Big(\int B_\eta B_\eta^\top\Big)^{-1}\Big(\int B_\eta dB_\eta^\top\Big),$$
where $B_\eta$ is an $r$-dimensional Brownian motion with covariance $\Omega_\eta \defeq \lim_{T\to\infty} T^{-1}\sum_{s,t}E[\eta_s\eta_t^\top]$. Using the same decomposition as Thm. \ref{thm:asymptotic_normality},
$$\hat\theta_{i,T:T+h}-\theta_{i,T:T+h}=\Lambda_i^\top(H^{-1}\hat A^h-I_r)F_T+\Lambda_i^\top(\hat F_T-HF_T)+(\hat\Lambda_i-(H^\top)^{-1}\Lambda_i)^\top F_T.$$
Next we can apply Slutsky’s theorem and a Taylor expansion of $\hat A^h$. The limits 
$F_T/\sqrt{T}\dconverge Z_\eta\sim N(0,\Omega_\eta)$ together with 
$T^{-1}\sum_{t\le T}F_t\varepsilon_{i,t}\dconverge \int B_\eta dB_\varepsilon^{(i)}$ \citep[\asn H]{bai2004estimating}, yields the following under $N/T\to\kappa$ and $N(\log\log T)^2/T\to 0$:
$$\sqrt{N}\,(\hat\theta_{i,T:T+h}-\theta_{i,T:T+h})\dconverge N(0,V_{i,t}) +
\sqrt{\kappa}\Big[h\Lambda_i^\top\Big(\int B_\eta B_\eta^\top\Big)^{-1} \int B_\eta dB_\eta^\top\, Z_\eta+Z_\eta^\top\Big(\int B_\eta B_\eta^\top\Big)^{-1}\int B_\eta dB_\varepsilon^{(i)}\Big],$$
where $V_{i,t}$ is the asymptotic variance as in \citep{bai2004estimating} and can be estimated by a HAC estimator. Hence inference for $\theta_{i,T:T+h}$ is feasible when $N/T\to 0$. We can extend this argument to the following generalization:
\begin{itemize}
    \item \textbf{Cointegrated factors}:~ Block-diagonalizing $A$ in canonical form, i.e. $A=S^{-1}\mathrm{diag}(I_{r_1},A_{2,2})S$ with invertible $S$ and $\rho(A_{2,2}) < 1$, yields super-consistency for the unit-root block and $\sqrt{T}$-consistency for the stable block.

    \item \textbf{Missing data}:~ \asn \ref{gasn:factor_model} and \ref{gasn:limdist} require modifications so that the Brownian-motion limits incorporate missing-entry contributions; the resulting adjustments are standard but omitted for brevity.
\end{itemize}

\section{Additional details of simulation studies in Sec. \ref{subsec:simulation}}\label{sec:simulation_plus}

In this section, we defer the experimental details of Sec. \ref{subsec:simulation}. The benchmarks mSSA and SyNBEATS implemented from the publicly available repositories \url{https://github.com/AbdullahO/mSSA} and \url{https://github.com/Crabtain959/SyNBEATS} respectively.

\textbf{Performance metric.}~ We evaluate $h$-step Mean Squared Forecast Error for first 32 rows
$$\mathrm{MSFE} = \frac{1}{32}\sum_{i=1}^{32}(\thetahatith - \thetaith)^2,$$
with $\thetahatith$ of different benchmarks. Due to the computational cost of SyNBEATS, we restrict the number of validation rows to 32 for ensuring computational feasibility across methods.

\textbf{Simultaneous adoption observation pattern}.~ We implement simultaneous adoption pattern in DGP-2 that is generated as follows. Unit-specific characteristics are generated as $X_i = \indic{\Lambda_i \ge 0}$. For the units with $X_i = 1$, 25\% randomly selected units have missing entries onward $\lceil 0.75 T \rceil$, and the remaining 75\% have all the entries observed. For the units with $X_i = 0$, 62.5\% randomly selected units have missing entries onward $\lceil 0.375 T \rceil$, and the remaining 37.5\% have all the entries observed.

\subsection{Data generative models} 
The outcomes are generated from the following factor model 
$$Y_{i,t} = \Lambda_i F_t + \varepsilon_{i,t},\quad \varepsilon_{i,t} \overset{\iid}{\sim}\normal{0}{0.1^2}.$$ 
We describe the four DGPs as follows--
\begin{enumerate}
    \item \textbf{DGP-1 (Purely autoregressive process).} 
    In this simulation setting, we set $r = 1$ and $\Lambda_i \overset{\iid}{\sim} \normal{0}{0.5^2}$. The factors are generated from an AR(1) process as $F_t = 0.5 F_{t-1} + \eta_t$ with $\eta_t \overset{\iid}{\sim} \normal{0}{(0.5)^2}$. The observations are MCAR with observation probability $0.7$.
    
    \item \textbf{DGP-2 (Autoregressive process with quadratic component).} The outcomes are generated from the one-factor model of DGP-1. The factors are generated as a sum of quadratic function $F_t^{(1)} = 2t^2/T^2$ and the AR(1) component $F_{t}^{(2)}$ as DGP-1, i.e. $F_t = F_t^{(1)} + F_t^{(2)}$. The observation mechanism is simultaneous adoption.
    
    \item \textbf{DGP-3 (ARMA process and quadratic component).} The outcomes are generated from the one-factor model of DGP-1. The quadratic component $F_t^{(1)}$ of the factors are generated similar to DGP-2. The autpregressive moving average component i.e. ARMA(1,1) has the following generative model:
    $$ F_t^{(2)} = 0.5 F_{t-1}^{(2)} - 0.4 F_{t-2}^{(2)} + 0.2 F_{t-2}^{(2)} + \eta_t + 0.5 \eta_{t-1}, $$
    with $\eta_t \overset{\iid}{\sim}\normal{0}{(0.7)^2}$. The observation mechanism is MCAR with observation probability $0.7$.

    \item \textbf{DGP-4 (VARMA process and periodic component).} In this simulation setting, we choose $r = 2$ for the outcome model and $\Lambda_{i,1}, \Lambda_{i,2} \overset{\iid}{\sim} \normal{0}{0.5^2}$. Each coordinate of the periodic component $F_t^{(1)}$ is generated in a similar spirit of \citet[App. B.1]{agarwal2020multivariate}, as mixture of two periodic components as follows:
    $$ p_t = 3\alpha \cos(10t/T) + 6\alpha \cos(20t/T), ~~ \text{and}~~ F_t^{(1)} = (p_t, p_t, \ldots, (r\text{ times}))^\top, $$
    where $\alpha > 0$ is a scalar constant. The stochastic part $F_t^{(2)}$ is generated from a VARMA(1,1) model as:
    \begin{equation}\label{eq:varma_dgp4}
        F_t^{(2)} = A F_{t-1}^{(2)} + \eta_t + B \eta_{t-1},~~  \eta_t \overset{\iid}{\sim}\normal{0}{(0.5)^2 I_r}.
    \end{equation}
    with the following choices of $A$ and $B$:
    $$ A = \begin{bmatrix}
        0.5 & 0.3\\
        -0.2 & 0.5
    \end{bmatrix},~~ B = \begin{bmatrix}
        0.3 & 0\\
        0 & 0.3
    \end{bmatrix}. $$
    
    \textbf{Choice of $\alpha$.} We tune $\alpha$ to control the signal-to-noise ratio (SNR) between $F_t^{(1)}$ and $F_t^{(2)}$. In \eqref{eq:varma_dgp4}, $\Sigma_F \defeq \var(F_t^{(2)}) = \sigma_\varepsilon^2 I_r + S$, with
    $$ C \defeq (A+B)^\top \Sigma (A+B)^\top,~~ \text{and}~~ S \defeq \sum_{\ell=0}^\infty A^\ell C (A^\ell)^\top $$
    satisfying the discrete Lyapunov equation
    $$ S = C + A S A^\top, ~~ \text{i.e.}~~ \vecop{S} = (I_{r^2} - A \otimes A)^{-1} \vecop{C}. $$
    Finally we set
    $$ \alpha = \left[ \frac{\frac{1}{r}\mathrm{trace}(\Sigma_F)}{\frac{1}{T-1} \sum_{t=1}^T \left(p_t - \overline{p}_t\right)^2 } \right]^{1/2},~~ \overline{p}_t \defeq \frac{1}{T} \sum_{t=1}^T p_t. $$
    
\end{enumerate}

\subsection{Implementation of \focus for DGP-2 to 4}

We implement \focus as described in Section~\ref{sec:method}, with model specifications adapted to the data-generating mechanisms of DGP-2 through DGP-4. For DGP-1 to DGP-3, the latent factor dynamics are modeled using autoregressive (AR) processes, with the lag order selected by AIC. For DGP-4, where the estimated noise component exhibits multivariate dependence, we instead fit a VAR(1) model, as detailed in Section~\ref{subsubsec:dgp4_details}. In addition, for DGP-2 to DGP-4 the estimated latent factors exhibit periodic behavior. Hence we employ an extended implementation of \focus that models deterministic periodic components.

\subsubsection{Tuning the penalized smoothing splines for DGP-2 and DGP-3} 
The penalized smoothing splines is a method for estimating functions with nonparametric regression \citep{green1993nonparametric, hastie2017generalized}. We implement the penalized spline with penalty tuned with 10-fold block-cross validation \citep{racine2000consistent,roberts2017cross} and block size $\lceil \sqrt{T} \rceil$ for selecting the spline penalty, with previous blocks training the temporally next blocks on a rolling basis.

\subsubsection{Extraction of periodicity in DGP-4}\label{subsubsec:dgp4_details}
We decompose the estimated factor process $\hat F_t \in \mathbb{R}^2$, $t = 1,\dots,T$ as follows,
$$\hat F_t = D_t + U_t,$$
where $D_t$ is a low-dimensional deterministic component and $U_t$ is a stochastic component. We model $D_t$ as
$$D_t = B^\top z_t,
\quad
z_t = \bigl(1,\ t,\ \{\cos(\omega t/T),\ \sin(\omega t/T)\}_{\omega \in \Omega}\bigr),$$
where $\Omega$ is a small set of candidate frequencies. Given $\Omega$, the coefficient matrix $B$ is estimated by ridge regression,
$$
\hat B(\Omega) = \arg\min_{B \in \RR^{(|\Omega|+2)\times r}} \sum_{t=1}^T \|\hat F_t - B^\top z_t\|_2^2 + \lambda \|B\|_F^2,$$
where $\|B\|_F = \sqrt{\sum_{i,j} B_{i,j}^2}$ denotes the Fr\"{o}benius norm of $B$. Candidate frequencies are obtained by first fitting a VAR(1) model to $\hat F_t$ and extracting residuals, followed by selecting the dominant peaks of their periodograms (mapped to the frequency range $\omega \in [\omega_{\min}, \omega_{\max}]$). Models with no frequency, one frequency, or two frequencies are considered, and the final frequency set $\hat\Omega$ is selected by minimizing a Gaussian BIC based on the residual sum of squares. Additionally, the ridge parameter $\lambda$ is fixed at a small value (e.g., $\lambda = 10^{-6}$) and serves only as a numerical regularization to avoid ill-conditioning of the deterministic design matrix; it is not tuned, since the complexity of the deterministic component is controlled via BIC-based selection of $\Omega$.

Given $\hat D_t = \hat B(\hat\Omega)^\top z_t$, we define $\hat U_t = \hat F_t - \hat D_t$ and fit a VAR(1) model and generate $h$-step-ahead forecasts $\hat U_{T+1:T+h}$. The final forecast is obtained by extrapolating the deterministic component and recombining,
$$\hat X_{T+1:T+h} = \hat D_{T+1:T+h} + \hat U_{T+1:T+h}.$$
This procedure yields a parsimonious decomposition that captures smooth periodic structure while preserving short-run dependence through the VAR component.

\subsection{Runtime comparison details}
We compare the runtime of \focus against mSSA and SyNBEATS. To establish a fair comparison, all methods are executed using CPU resources only, i.e. in particular, SyNBEATS is run without GPU acceleration. All experiments are conducted on the BioHPC cluster using up to 30 CPU cores, with identical resource constraints and replicates of input data across methods. Runtime is measured in seconds and includes model fitting and forecasting, excluding data generation and disk input and output. For robustness, we report median runtime across repeated runs, and quantify their dispersion using the median absolute deviation (MAD).

\section{Additional details of the HeartSteps case study in Sec. \ref{subsec:data}} \label{sec:heartsteps_add}

The HeartSteps data has 37 users and the maximum number of decision points across users is 315. Users were considered unavailable (i.e., not nudged) when driving, offline and in similar circumstances. The readers are referred to \url{https://github.com/klasnja/HeartStepsV1} for the data source, and \citet{heartstepsv1, liao2020personalized} for details.

\subsection{Preprocessing the data} 

We consider three variables in the data-- the binary variable \texttt{available} that indicates user availability at the underlying decision slot, the binary variable \texttt{send} that tracks whether a user was nudged with a notification, and \texttt{jbsteps30} that measures the number of steps accomplished by the user 30 minutes after sending the activity prompt. The outcome variable is $\log(1 + \texttt{jbsteps30})$, and the treatment variable is \texttt{available*nudge}.

For $i = 1,\ldots, 37$, the $i$\textsuperscript{th} row of the outcome matrix $Y$ and the observation matrix $W$ are constructed by stacking $\log(1 + \texttt{jbsteps30})$ and \texttt{available*sent} indicators for user $i$ at each quintuplets for a single decision day. Next, we discard user 31 from the analysis due to consistently low nudges across time points, resulting to $N = 36$. Out of the 315 total time points, the horizon $T$ for model training is considered in the range $[100, 200]$ in multiples of 10 for obtaining enough time points to extract the factor dynamics, as well as to rule out the last few intervention periods when many users discontinued the experiment.

\subsection{\focus on slot pair (4,5)} The factors are estimated with the PCA method of \citet{xiong2023large}. Since $T > N$ in most cases and interchanging the roles of $\hat\Lambda$ and $\hat F$ in Step \ref{item:step1} of \focus lowers the chance of $(|\cQ_{i,j}|)_{i,j\in [N]}$ being small or zero. Upon observing the scree plots of the PCA for most slices of the data, we choose the dimension $r=7$ to explain at east 80\% of the data for most choices of $T$.

\subsection{Expression of \focus estimator for HeartSteps}

For forecast horizon at $T = 5K,~ K \ge 2$, the estimated factors are $\{\hat F_t, t = 1,\ldots, 5K\}$. We select the estimated factors at slot $s \in \{4,5\}$ as $\hat F^{(s)}$, where 
$$\big(\hat F^{(s)}\big)^\top := \left\{\hat F_{s+5j}: j = 0, \ldots, K-1\right\}.$$ 

Next, we regress the estimated factors at slot $5$ on that of slot $4$ to obtain the estimated coefficient matrix as $$\hat A_{4\to 5} = \big( \hat F^{(5)}\big)^\top \hat F^{(4)} \left[ \big( \hat F^{(4)}\big)^\top \hat F^{(4)} \right]^{-1}.$$ 

This approach is similar to the VAR(1) coefficient matrix estimation in \eqref{eq:A_hat}. The 5-step forecast estimator for a user $i$ for $T = 5K$ is $\hat \theta_{i, T + 5} = \hat \Lambda_i^\top \hat A_{4\to 5} \hat F_{5K}$.

\subsection{Performance metric}

We compare \focus and mSSA using the Mean Squared Relative Prediction Error (MSRPE) for forecasting future step counts under intervention. Specifically, for each individual $i$, we predict the step count at horizon $T+5$ using observations up to time $T$, and evaluate performance only for individuals with non-zero realized steps. The MSRPE is defined as
$$ \mathrm{MSRPE}
= \frac{\sum_{i : Y_{i,T+5} > 0}
\bigl(\hat{\theta}_{i,T:T+5} - Y_{i,T+5}\bigr)^2 / Y_{i,T+5}^2}
{\#\{i : Y_{i,T+5} > 0\}}.$$
The steps exhibit substantial heterogeneity across individuals; scaling the squared error by $Y_{i,T+5}^2$ prevents the metric from being dominated by individuals with unusually large activity levels and yields a more robust comparison across methods.

\newpage
\section{Additional figures for experiments}

\begin{figure*}[!h]
    \centering
    \subfloat[]{\includegraphics[width=0.25\textwidth]{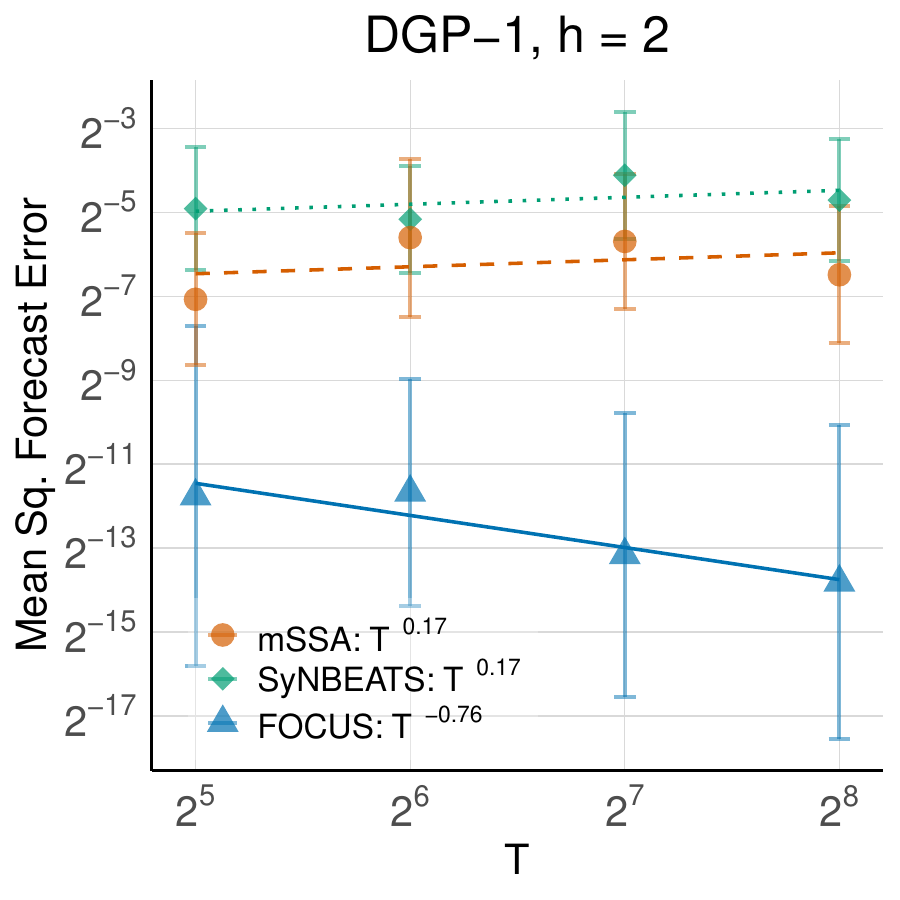}}
    \subfloat[]{\includegraphics[width=0.25\textwidth]{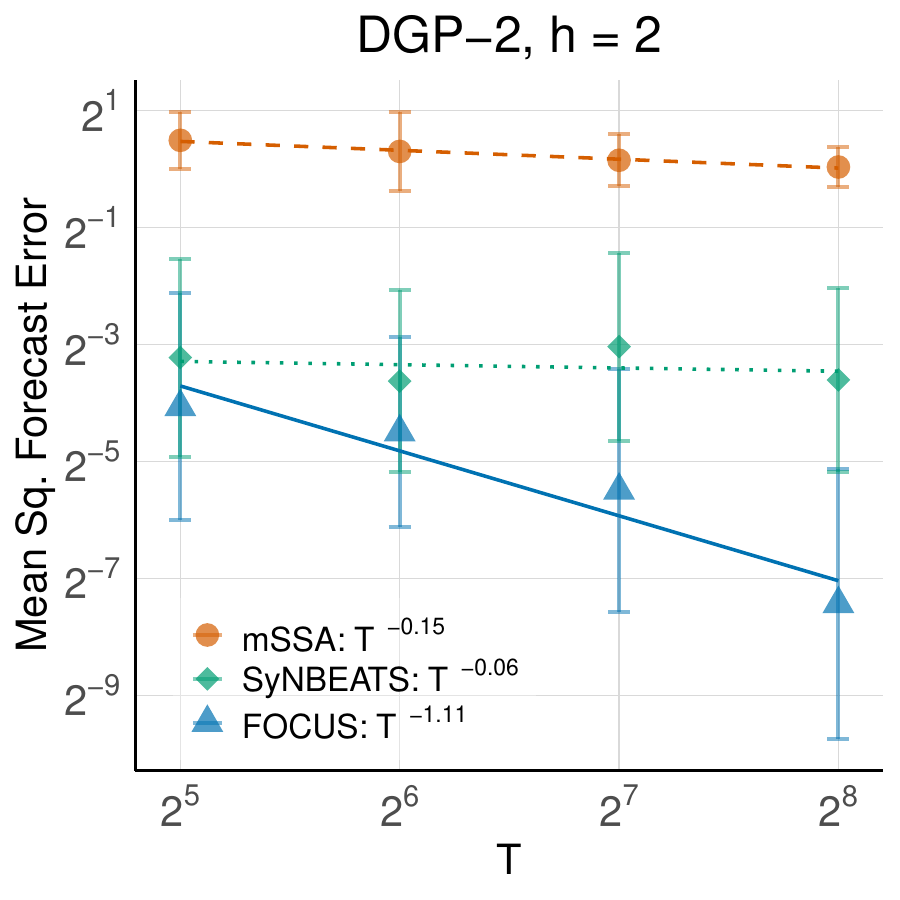}}
    \subfloat[]{\includegraphics[width=0.25\textwidth]{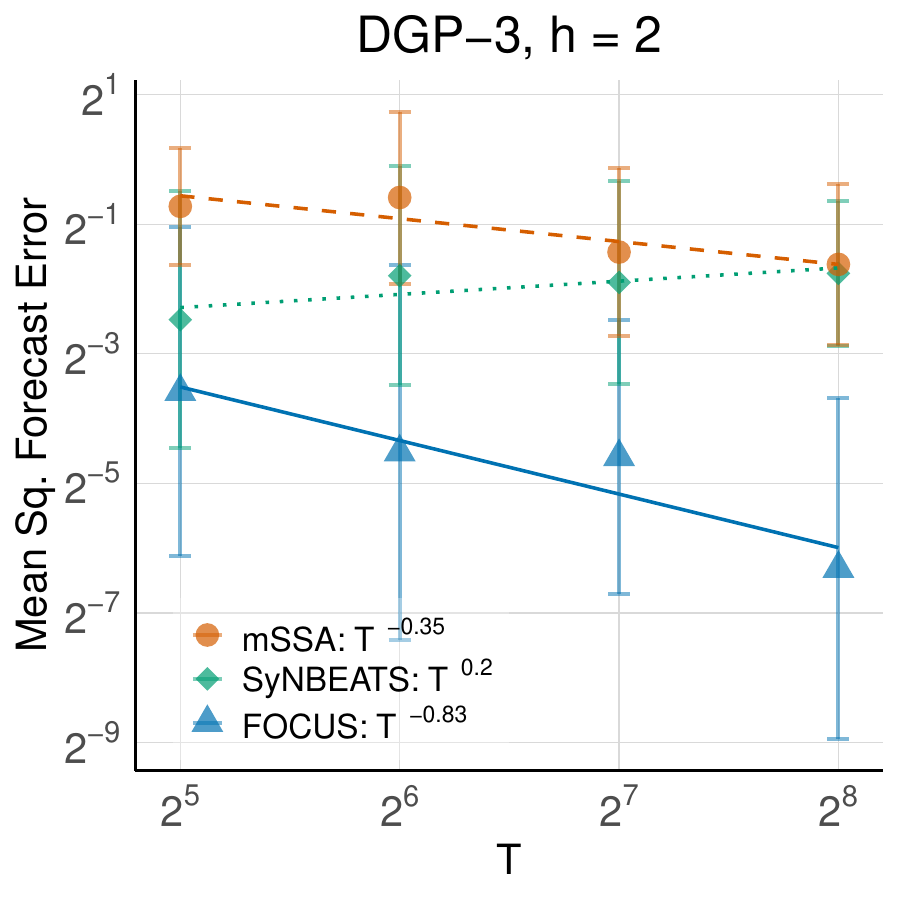}}
    \subfloat[]{\includegraphics[width=0.25\textwidth]{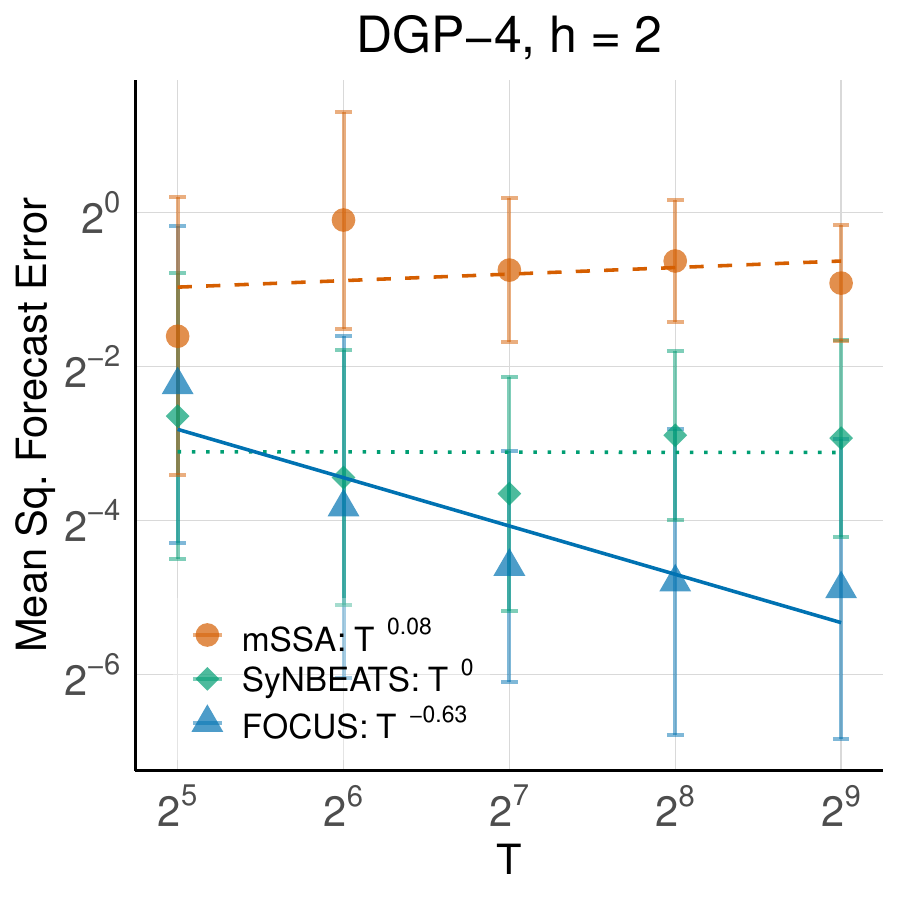}}\\
    \subfloat[]{\includegraphics[width=0.25\textwidth]{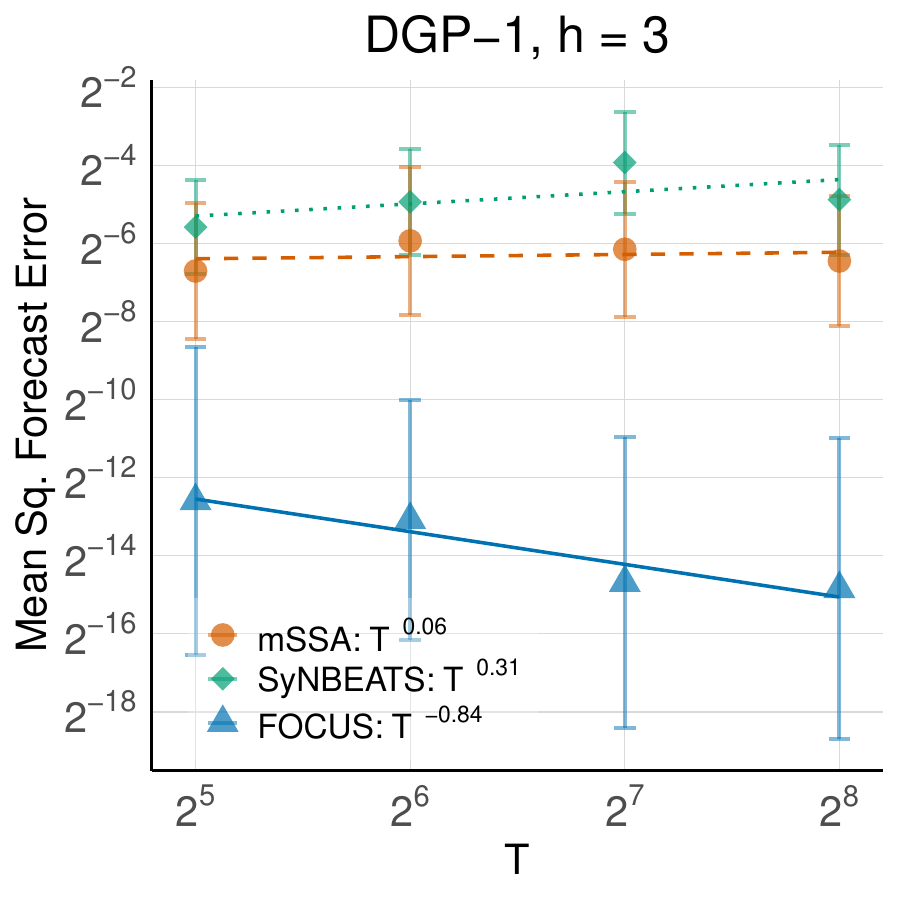}}
    \subfloat[]{\includegraphics[width=0.25\textwidth]{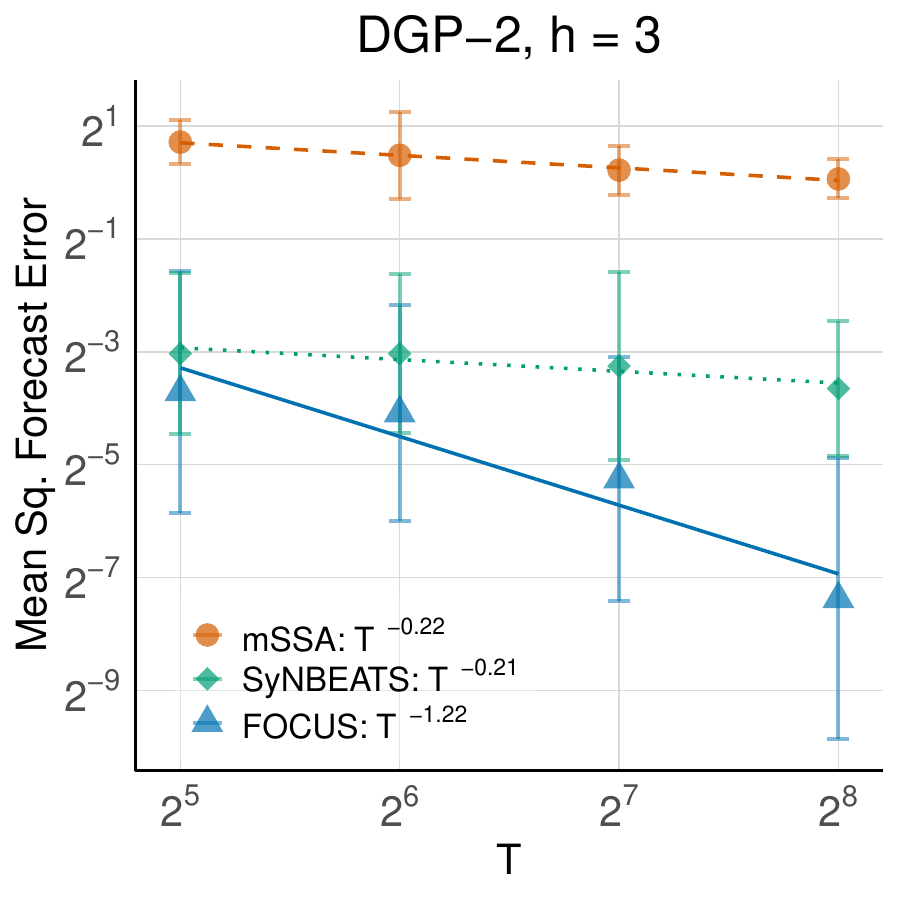}}
    \subfloat[]{\includegraphics[width=0.25\textwidth]{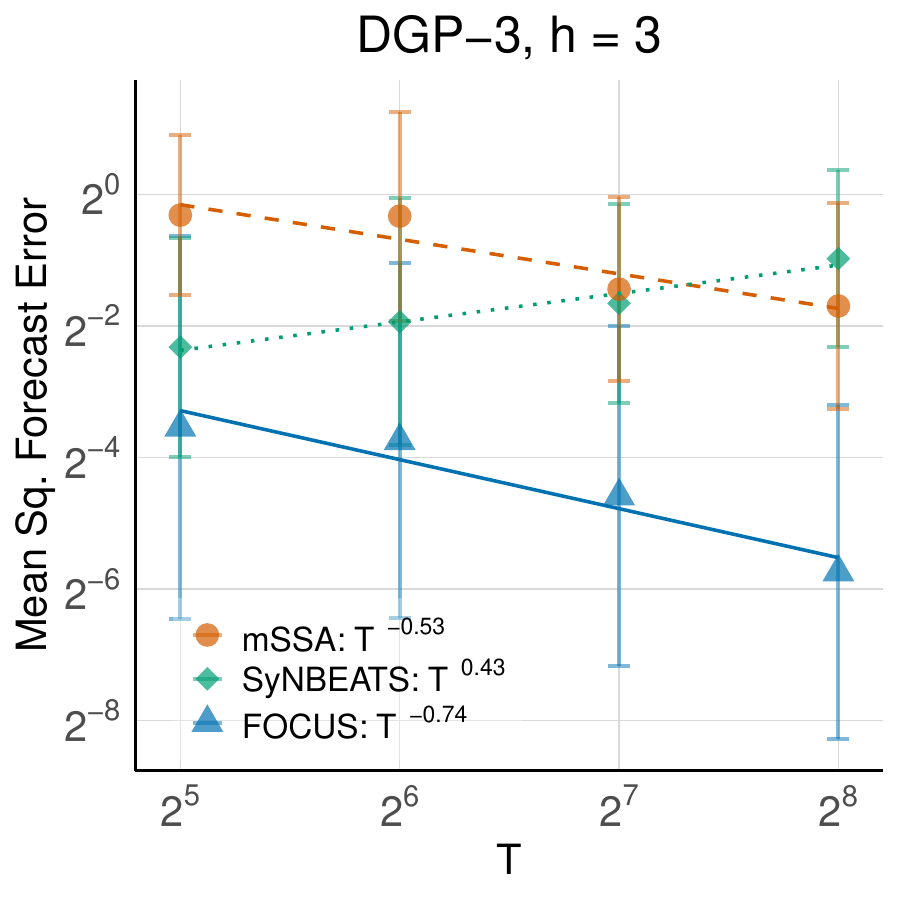}}
    \subfloat[]{\includegraphics[width=0.25\textwidth]{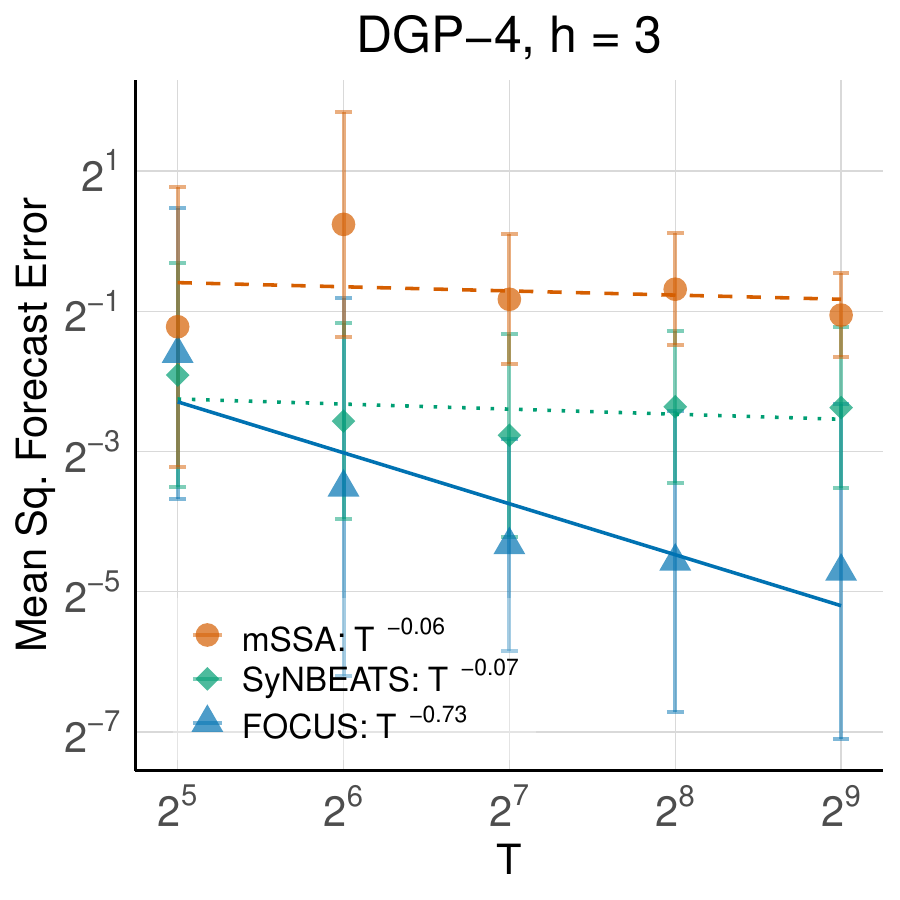}}
    \caption{\textbf{Mean Squared Forecast Error (MSFE, averaged over 30 trials) at $h = 2, 3$ across the benchmarks for $N = 64$ and 4 generative models.} Panels (a)-(h) present the average MSFE of \focus (blue triangle), mSSA (orange circle) and SyNBEATS (green diamond) across $T \in \{2^5,\ldots, 2^8\}$, and the vertical lines mark the one standard deviation error bars. As comapared to SyNBEATS and mSSA, \focus has lower average MSFE that decreases faster with $T$ (empirical rates in the legends).}
    \label{fig:err_vs_T_extra}
\end{figure*}

\begin{figure}[!h]
    \centering
    \includegraphics[width=0.55\linewidth]{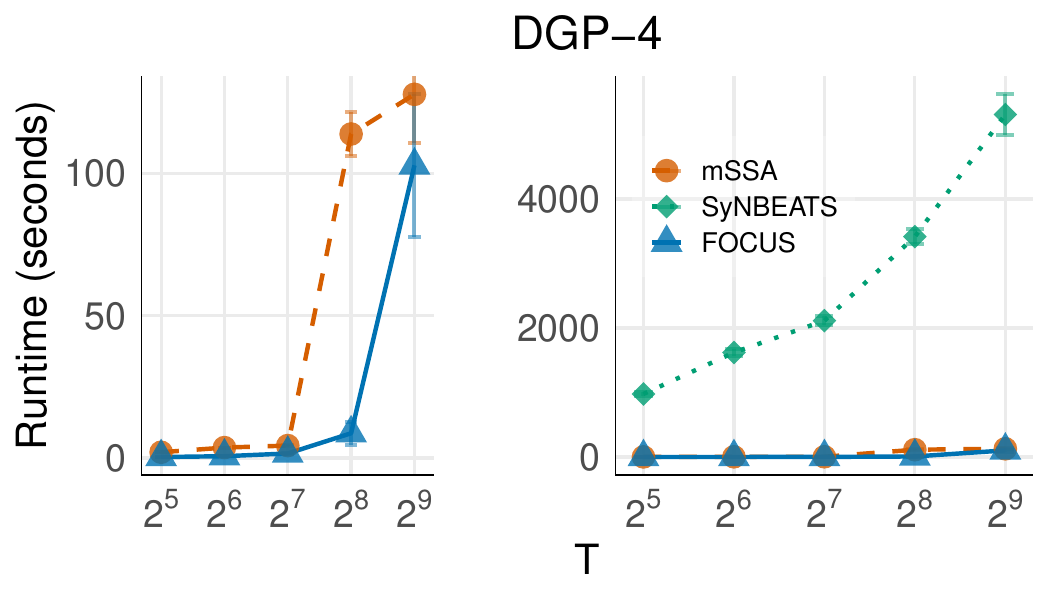}
    \caption{\textbf{Median runtime in DGP-4 across 30 trials for $N=64$ and $T\in\{2^5,\ldots,2^9\}$}. The vertical bars represent the median absolute deviation of the runtime. The right plot shows all methods, and the left plot zooms in on \focus and mSSA. Across all $T$, \focus (blue triangles) is faster than mSSA (orange circles) and substantially faster than SyNBEATS (green diamonds).}
    \label{fig:runtime_dgp4}
\end{figure}

\begin{figure*}[!t]
    \centering
    \subfloat[]{\includegraphics[width=0.25\textwidth]{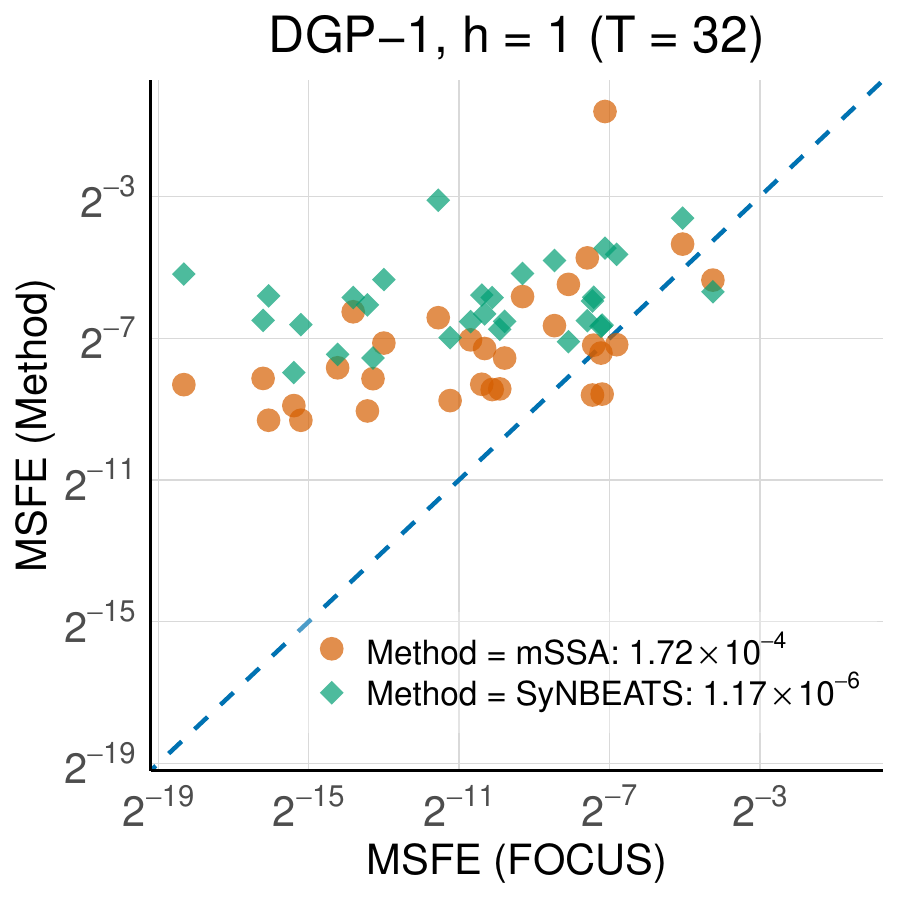}}
    \subfloat[]{\includegraphics[width=0.25\textwidth]{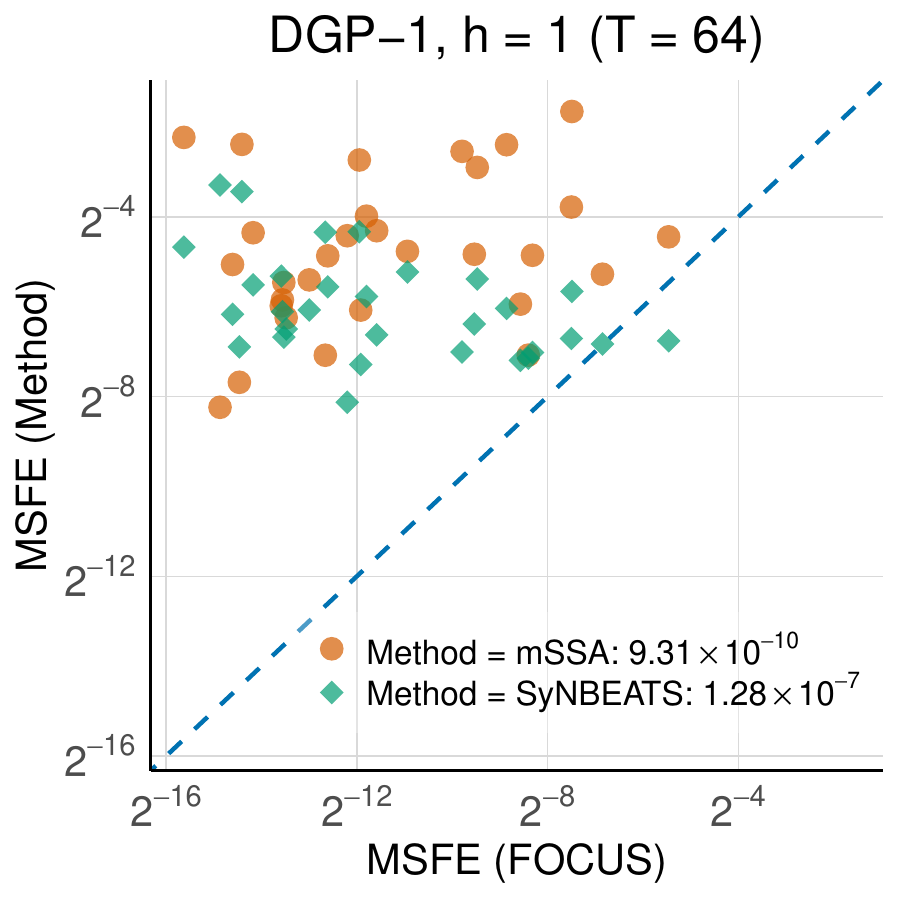}}
    \subfloat[]{\includegraphics[width=0.25\textwidth]{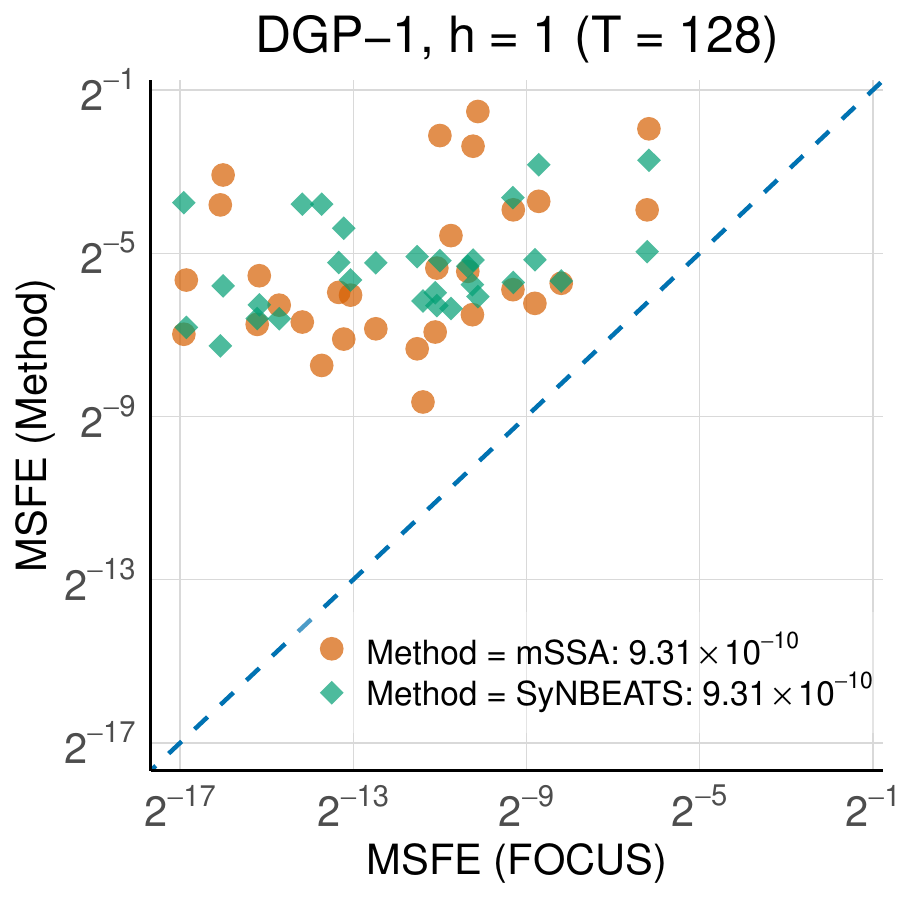}}
    \subfloat[]{\includegraphics[width=0.25\textwidth]{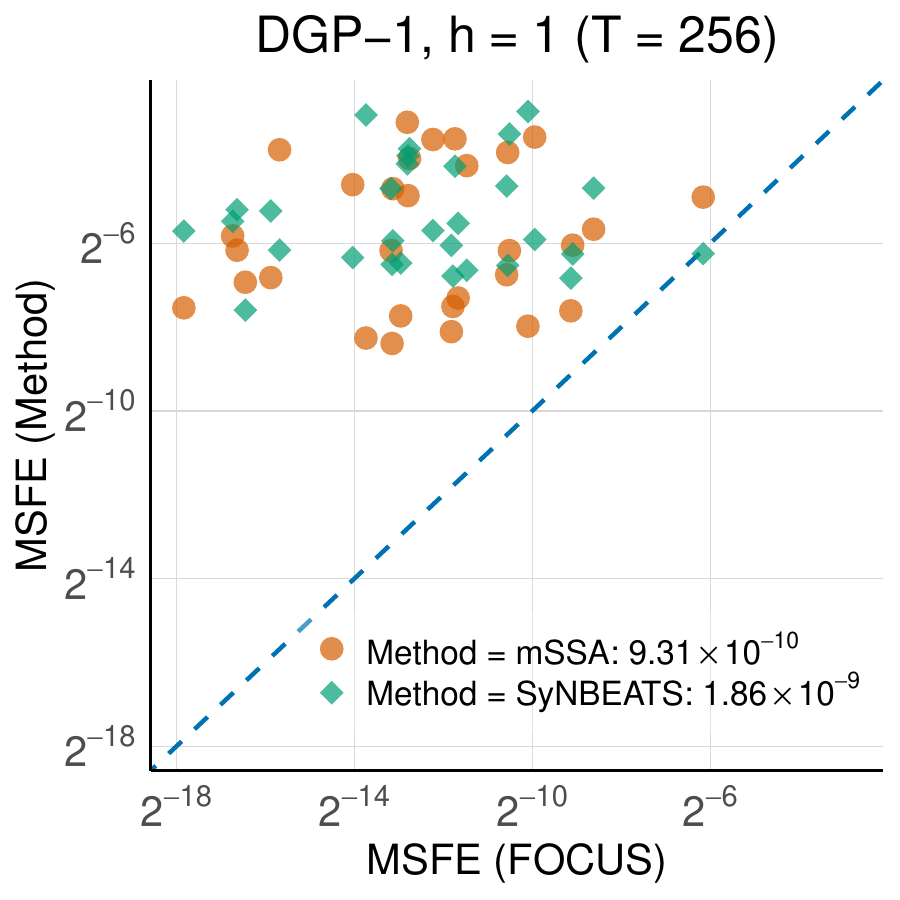}}\\
    \subfloat[]{\includegraphics[width=0.25\textwidth]{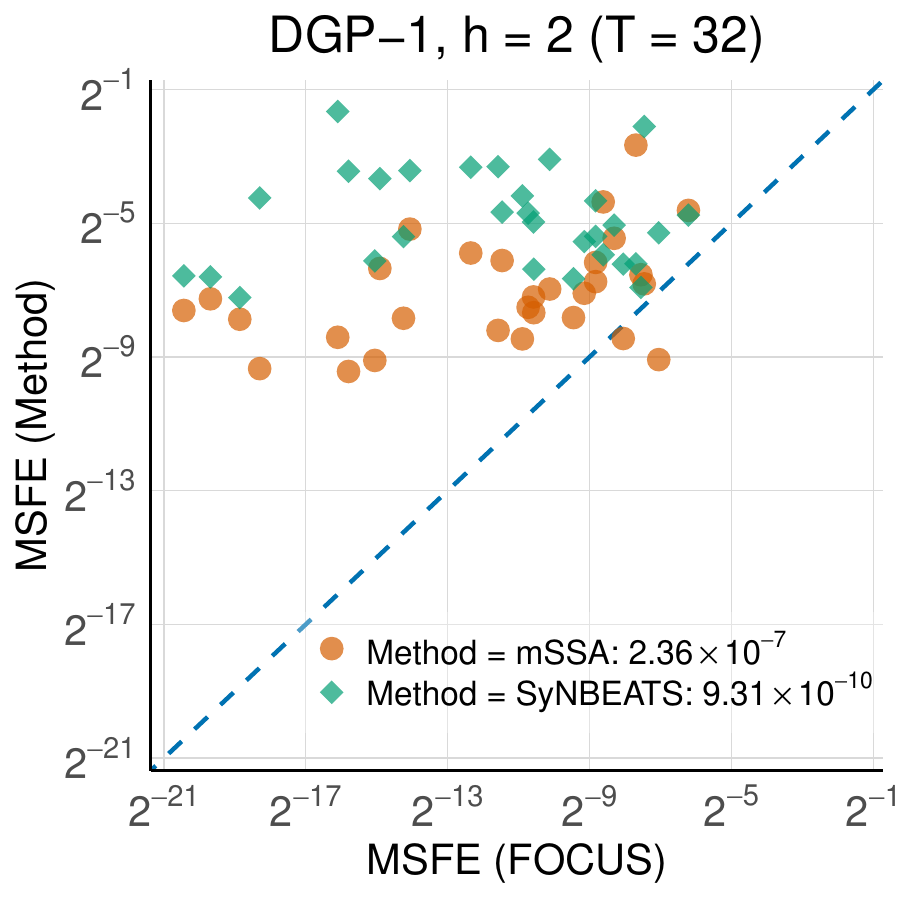}}
    \subfloat[]{\includegraphics[width=0.25\textwidth]{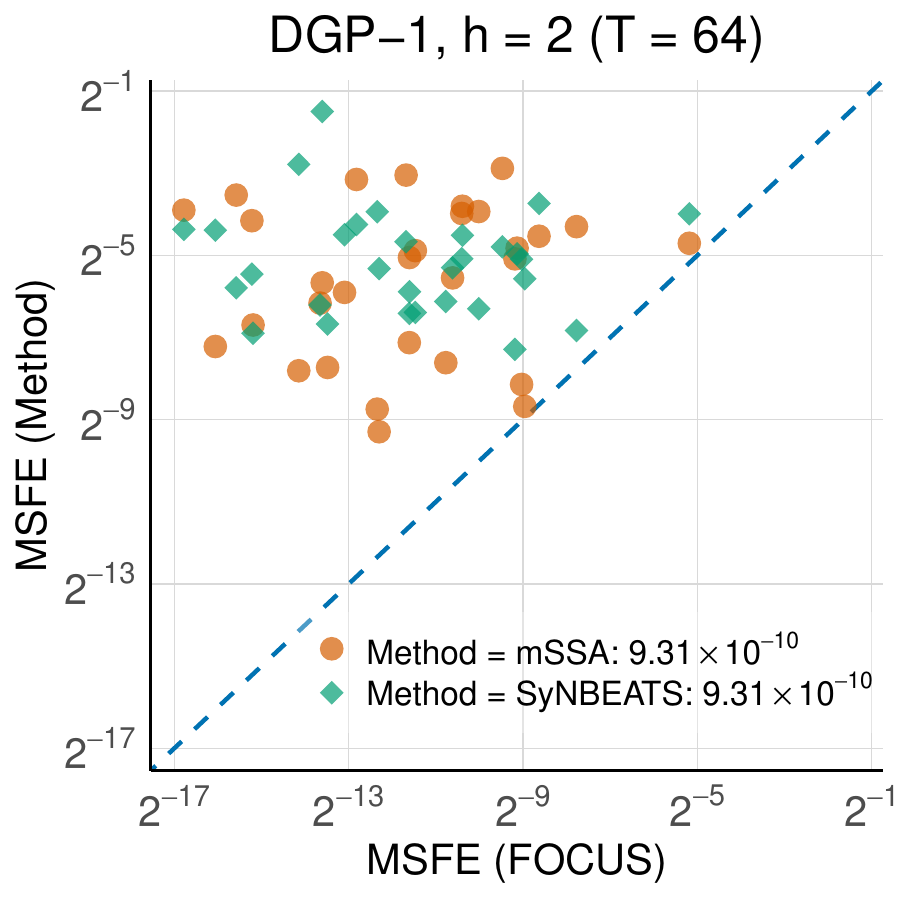}}
    \subfloat[]{\includegraphics[width=0.25\textwidth]{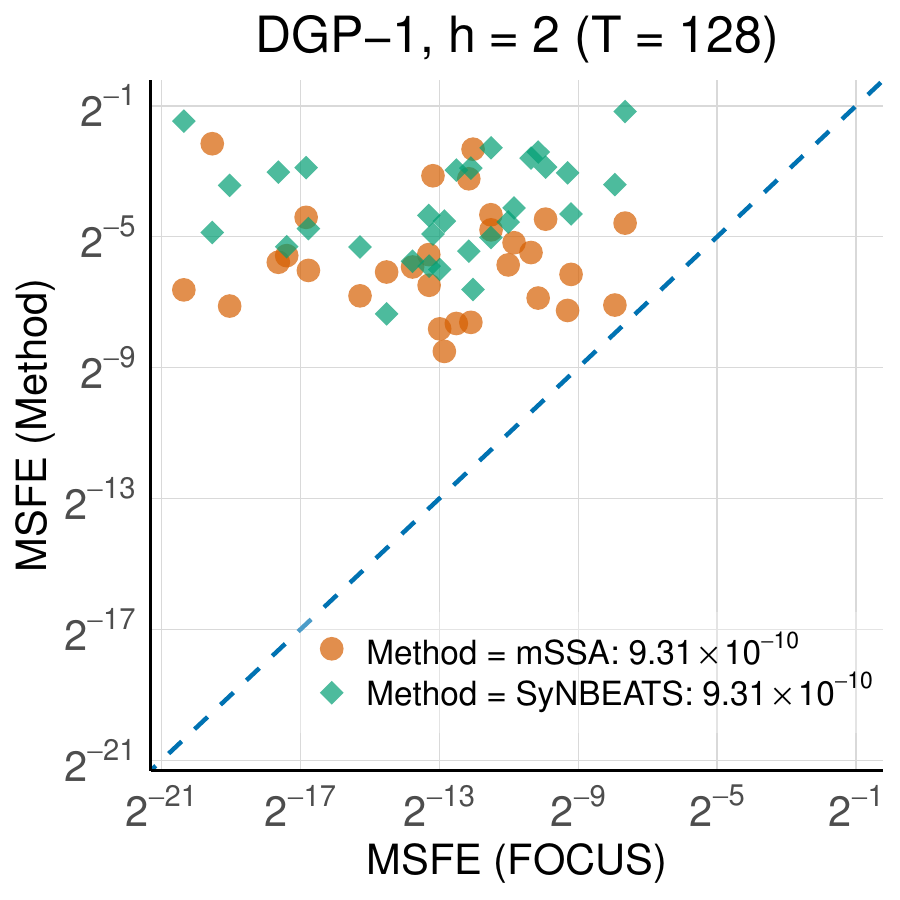}}\subfloat[]{\includegraphics[width=0.25\textwidth]{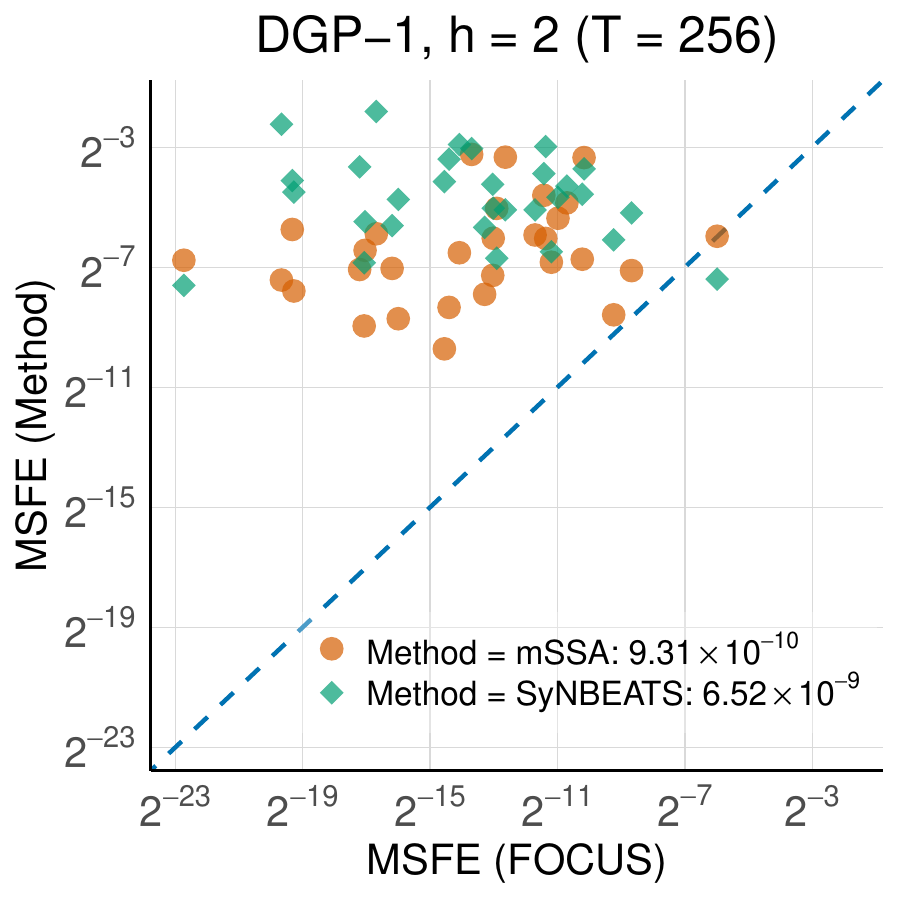}}\\
    \subfloat[]{\includegraphics[width=0.25\textwidth]{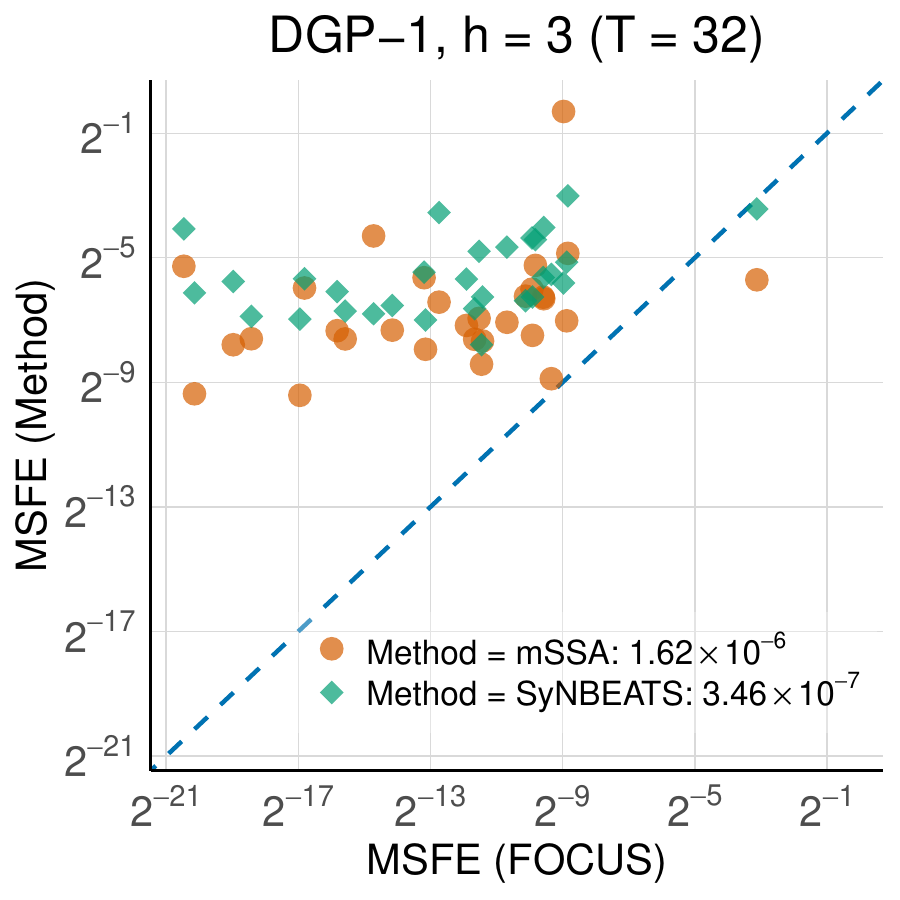}}
    \subfloat[]{\includegraphics[width=0.25\textwidth]{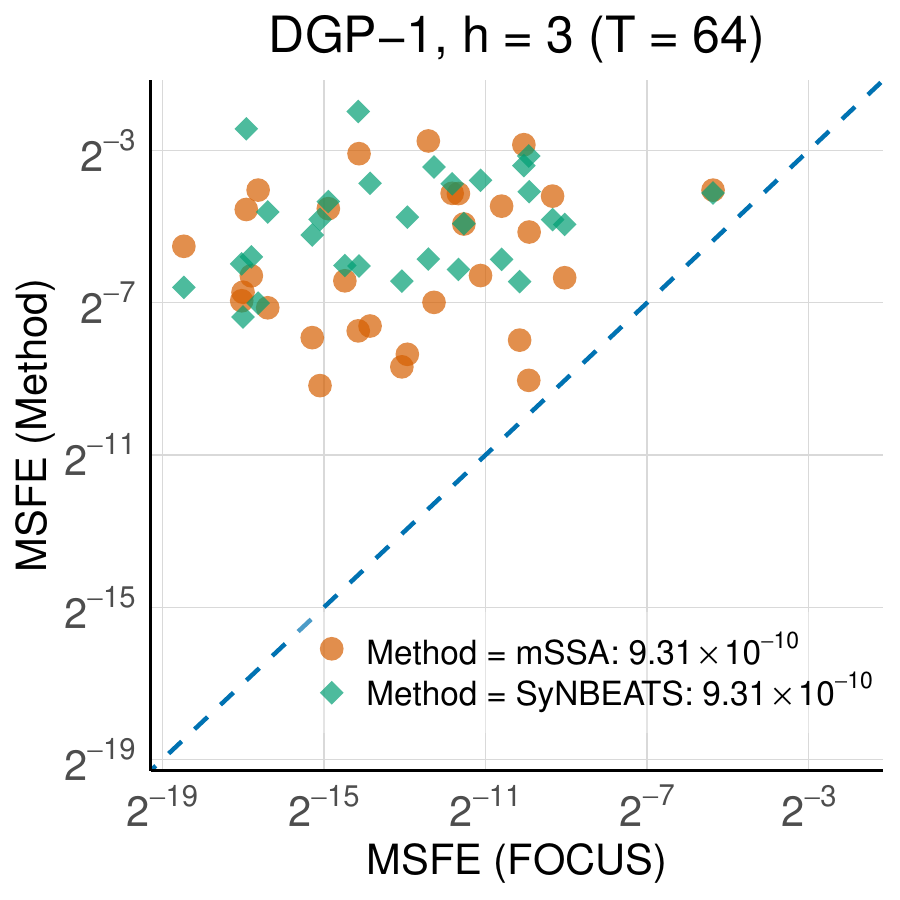}}
    \subfloat[]{\includegraphics[width=0.25\textwidth]{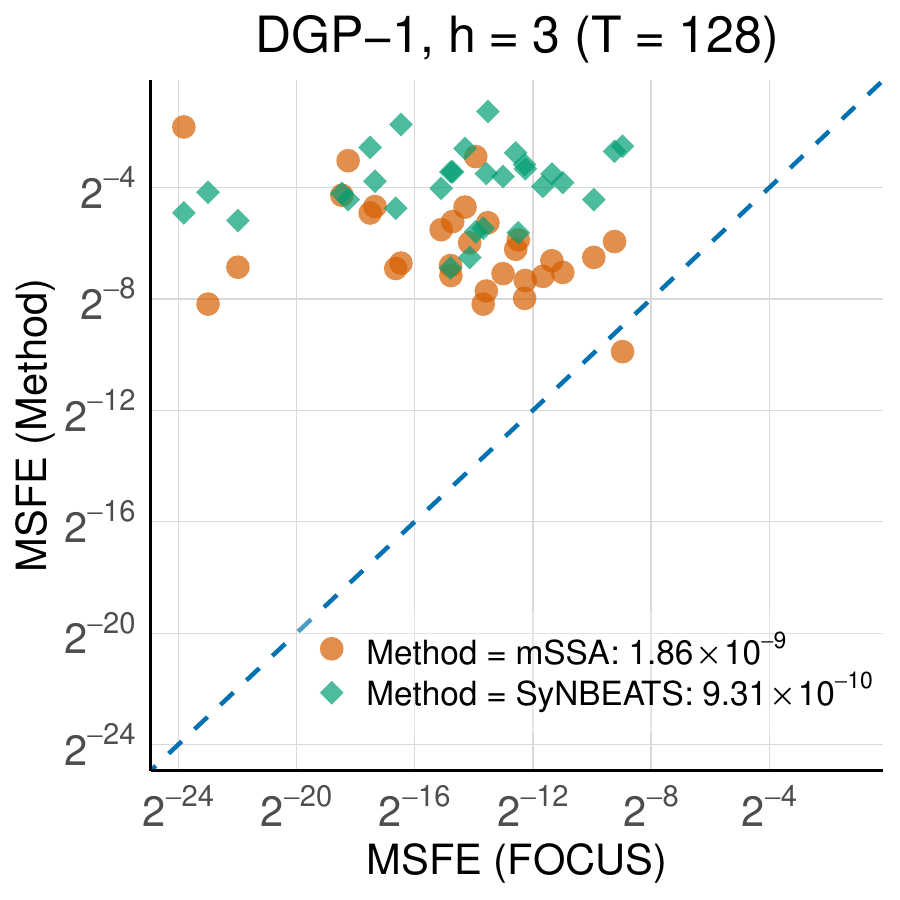}}
    \subfloat[]{\includegraphics[width=0.25\textwidth]{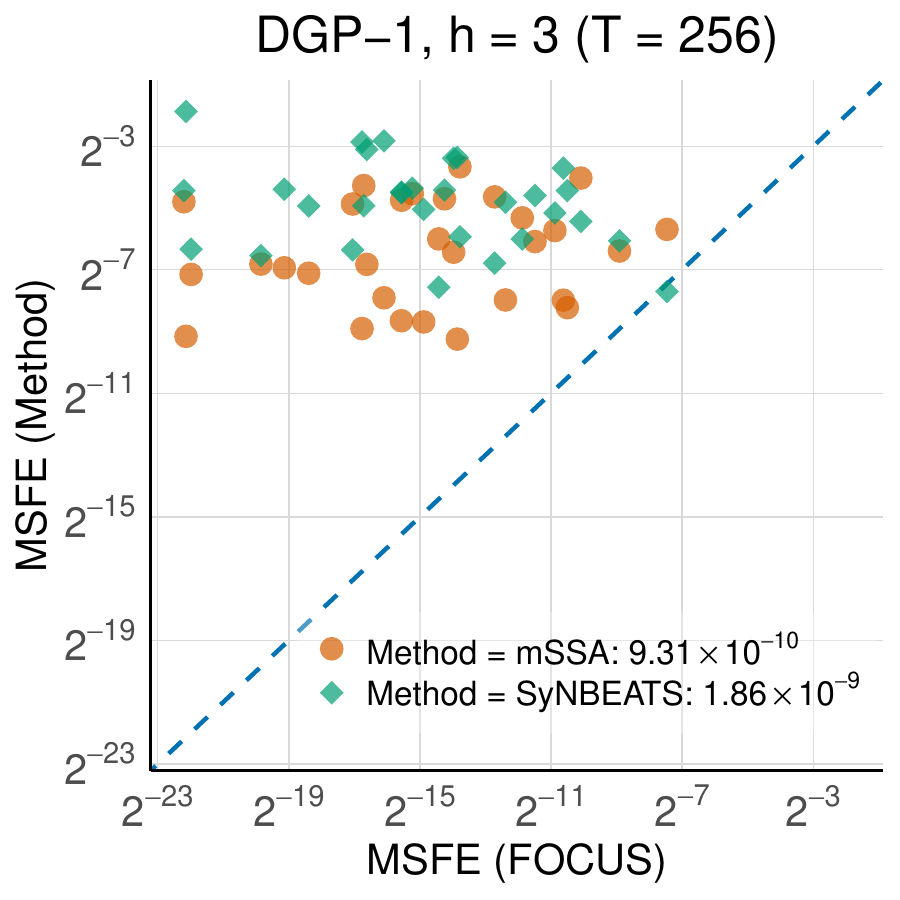}}
    \caption{\textbf{Scatter plots of Mean Squared Forecast Error (MSFE) over 30 trials across benchmarks for $N=64$ under DGP-1.} Across all $T$ and $h$, \focus attains significantly lower errors (one-sided Wilcoxon signed-rank test, $p<0.01$; see legends), with scatter points concentrated in the $y>x$ line.}
\end{figure*}

\begin{figure*}[!t]
    \centering
    \subfloat[]{\includegraphics[width=0.25\textwidth]{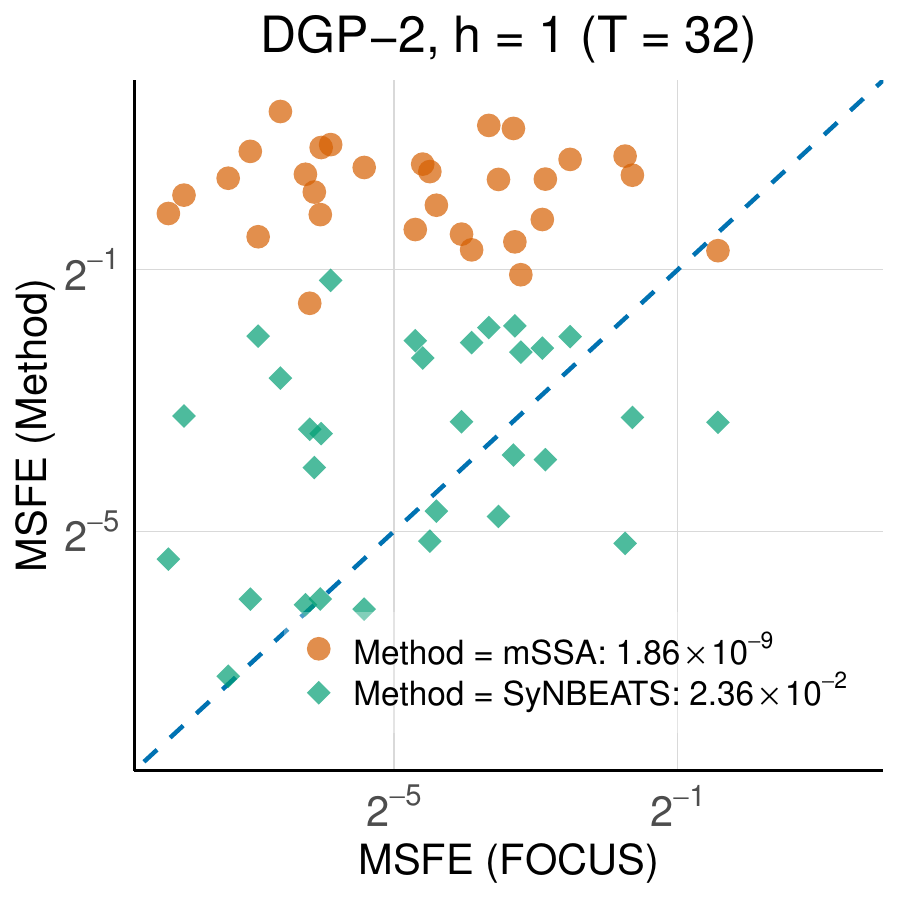}}
    \subfloat[]{\includegraphics[width=0.25\textwidth]{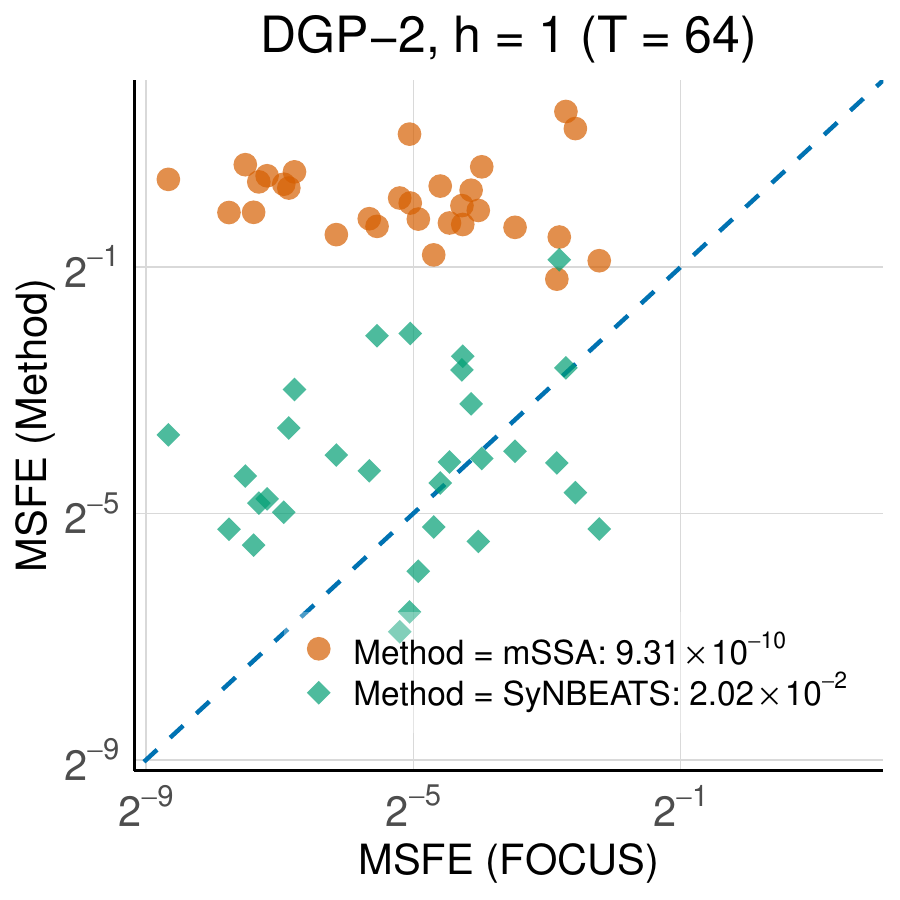}}
    \subfloat[]{\includegraphics[width=0.25\textwidth]{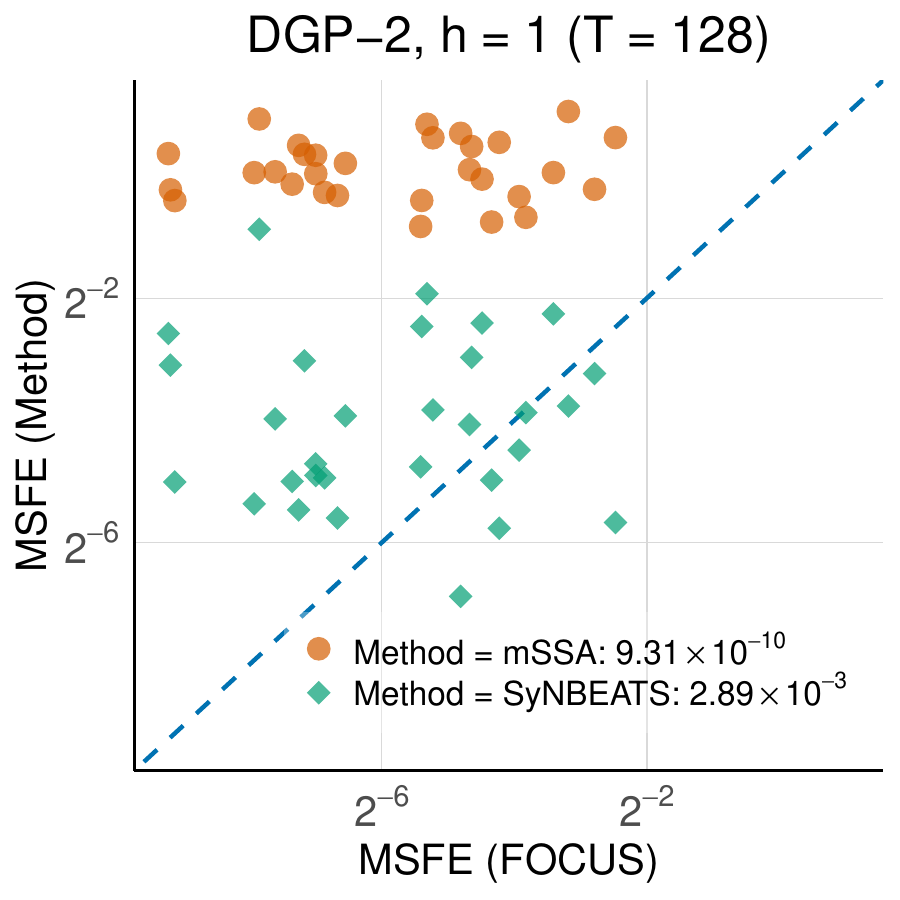}}
    \subfloat[]{\includegraphics[width=0.25\textwidth]{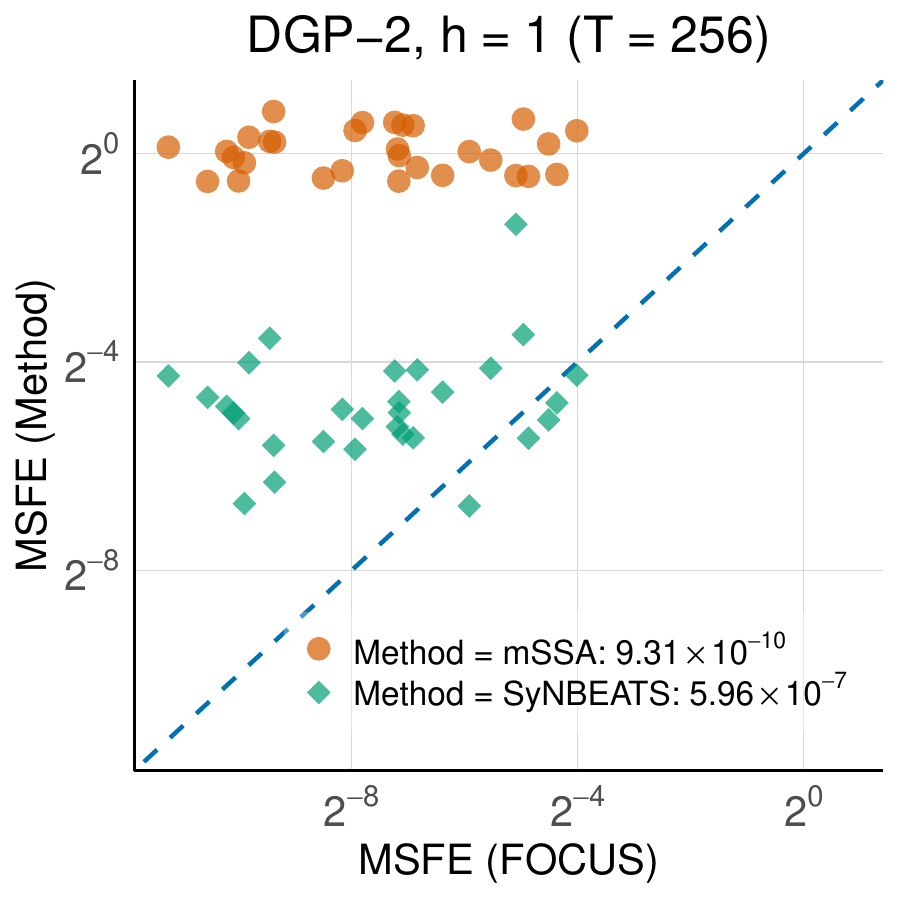}}\\
    \subfloat[]{\includegraphics[width=0.25\textwidth]{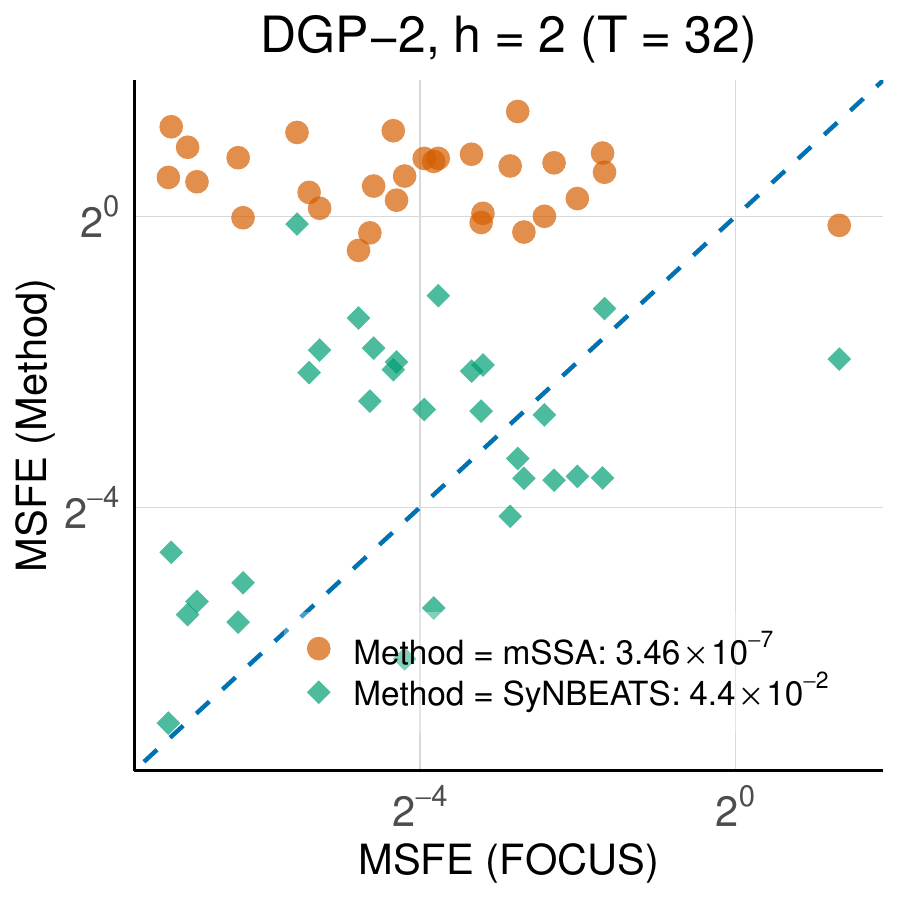}}
    \subfloat[]{\includegraphics[width=0.25\textwidth]{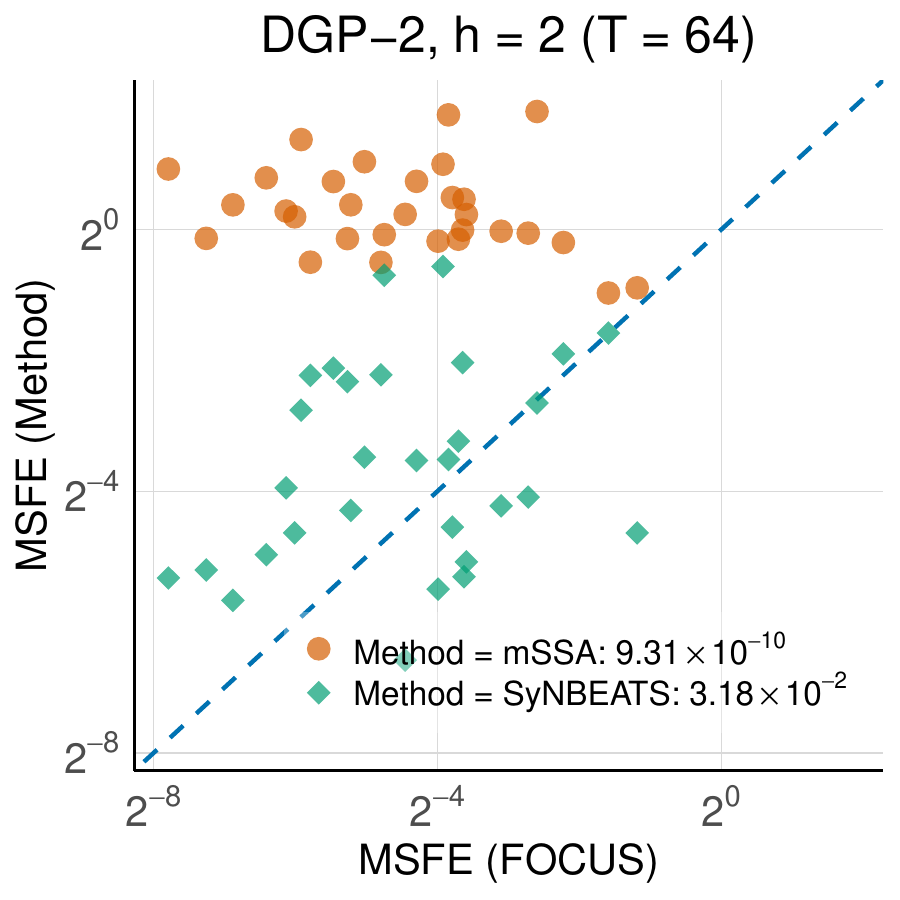}}
    \subfloat[]{\includegraphics[width=0.25\textwidth]{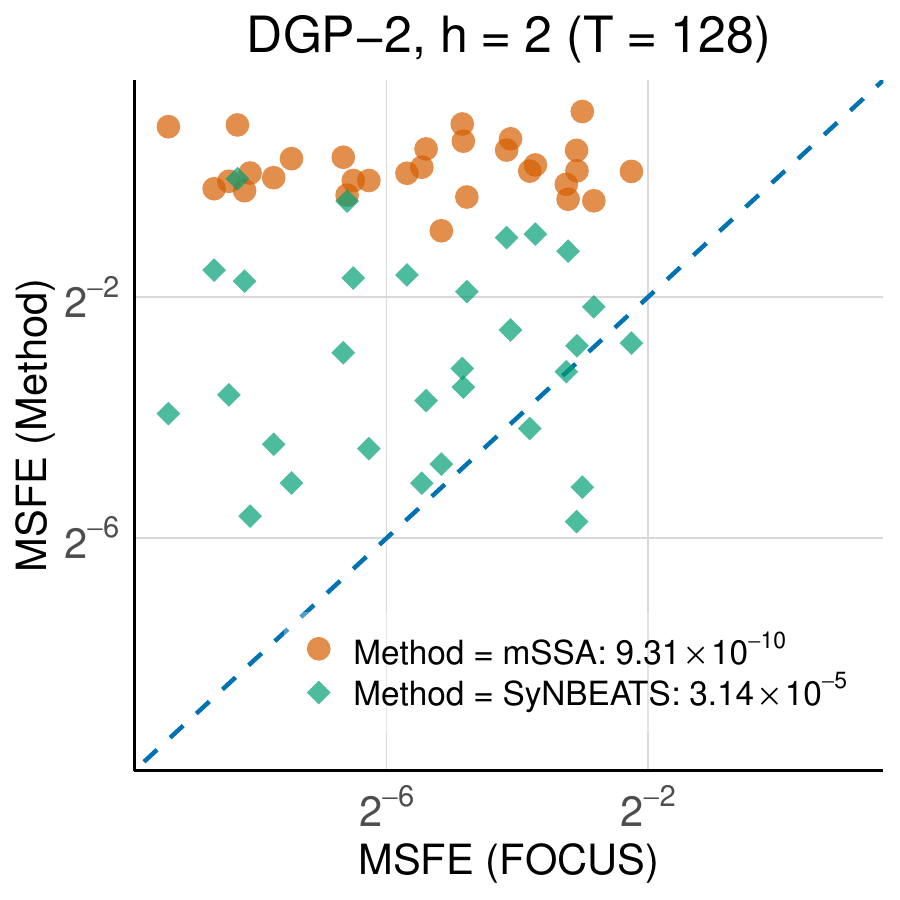}}\subfloat[]{\includegraphics[width=0.25\textwidth]{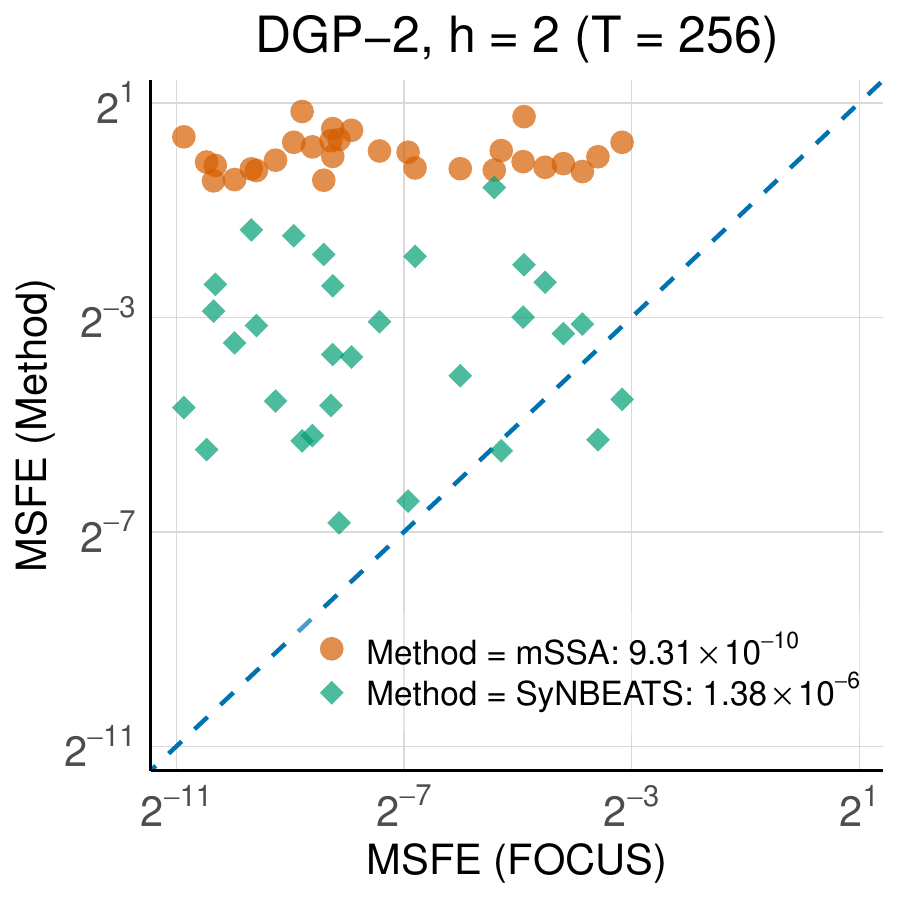}}\\
    \subfloat[]{\includegraphics[width=0.25\textwidth]{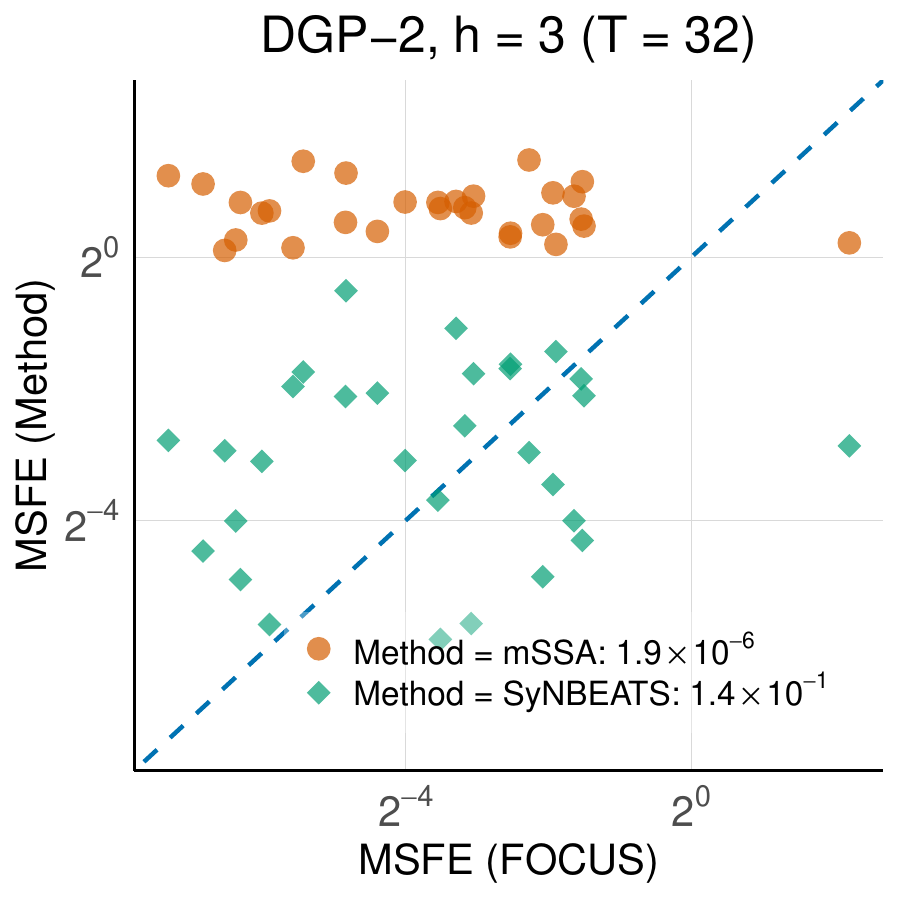}}
    \subfloat[]{\includegraphics[width=0.25\textwidth]{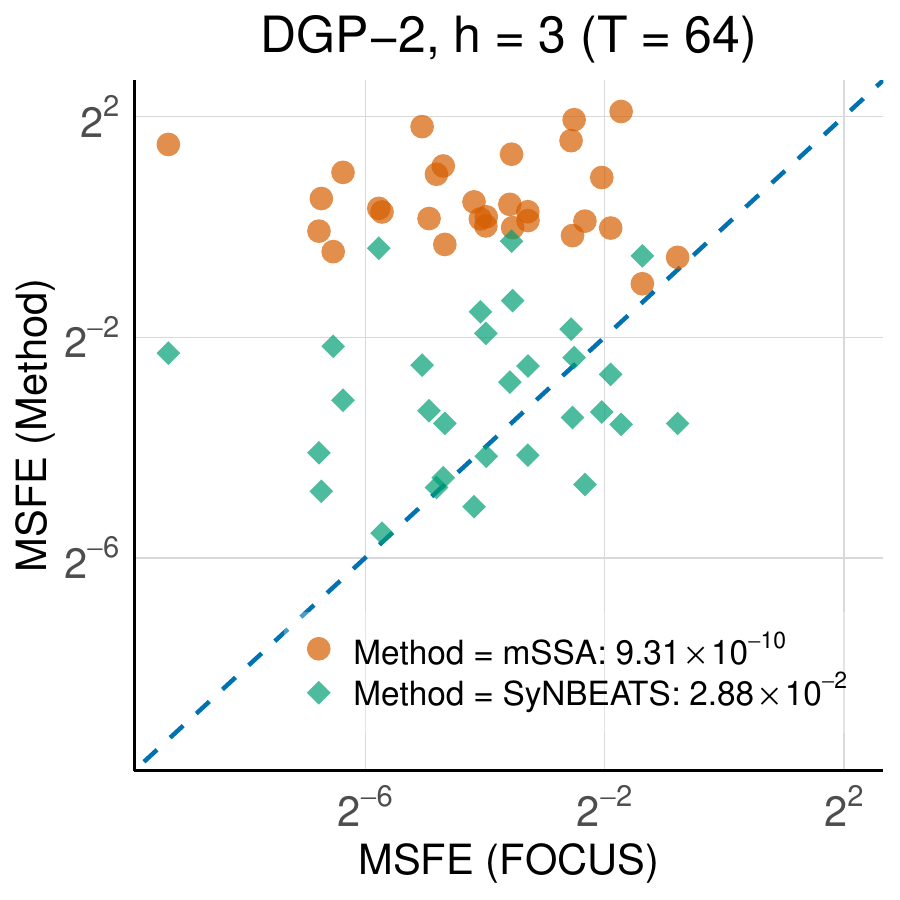}}
    \subfloat[]{\includegraphics[width=0.25\textwidth]{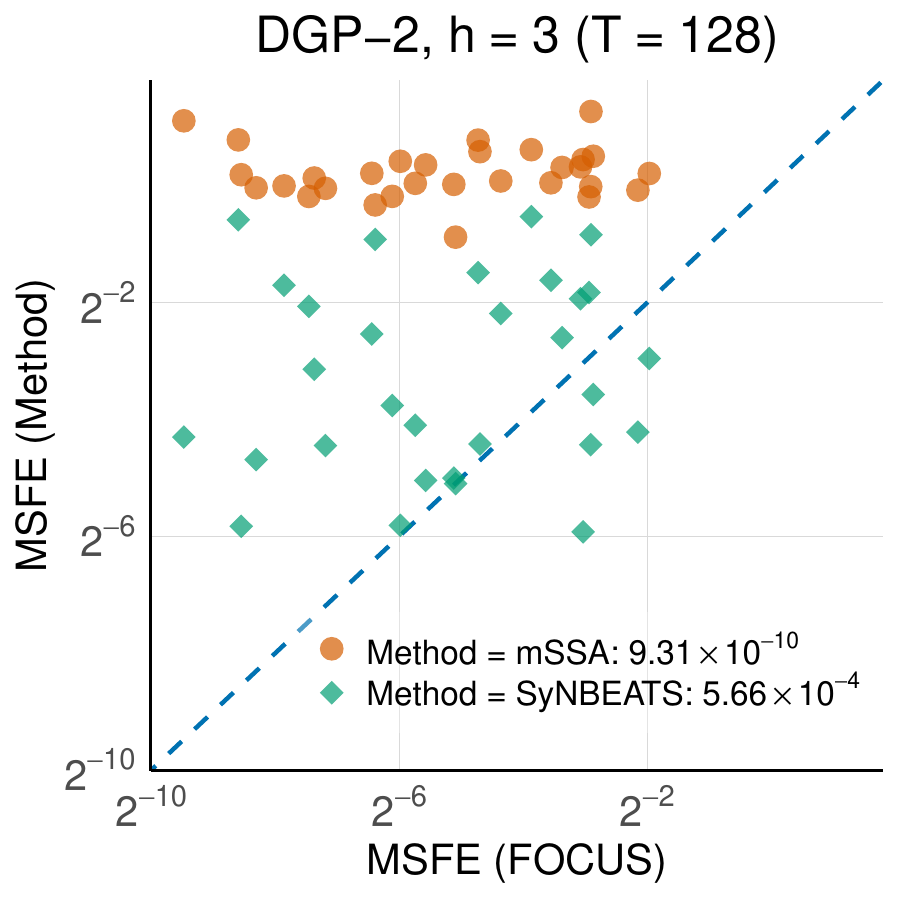}}
    \subfloat[]{\includegraphics[width=0.25\textwidth]{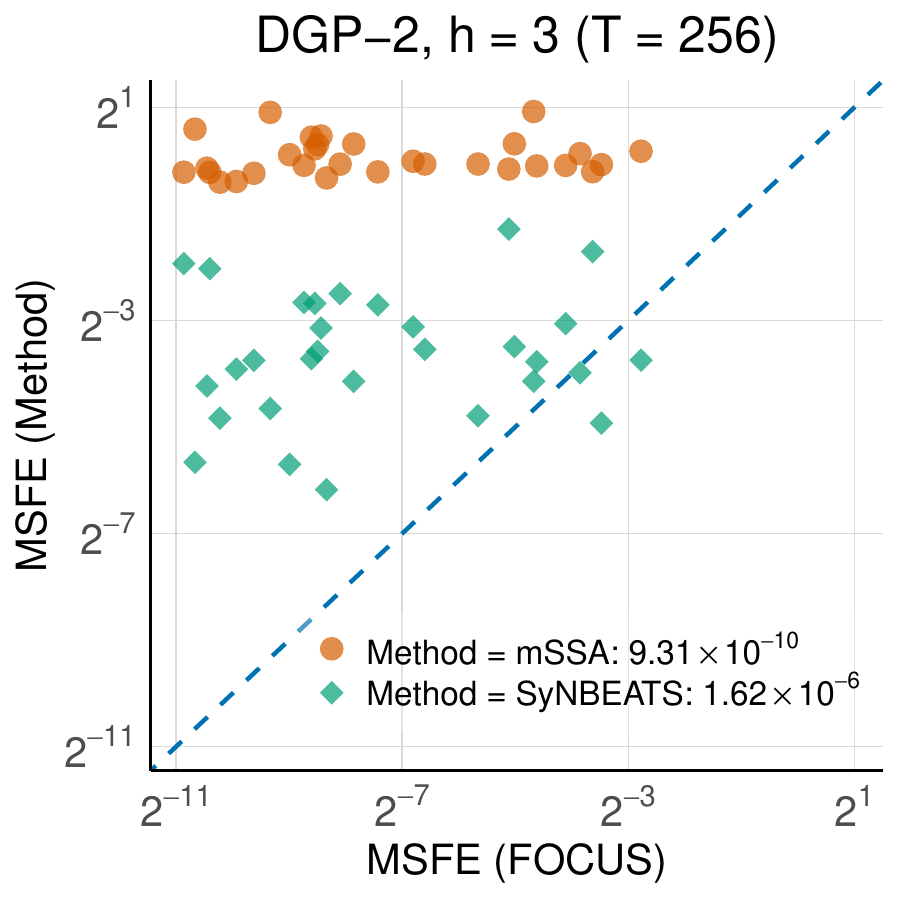}}
    \caption{\textbf{Scatter plots of Mean Squared Forecast Error (MSFE) over 30 trials across benchmarks for $N=64$ under DGP-2.} Across all $T$ and $h$, \focus attains significantly lower errors (one-sided Wilcoxon signed-rank test, $p<0.01$; see legends), with scatter points concentrated in the $y>x$ line.}
\end{figure*}

\begin{figure*}[!t]
    \centering
    \subfloat[]{\includegraphics[width=0.25\textwidth]{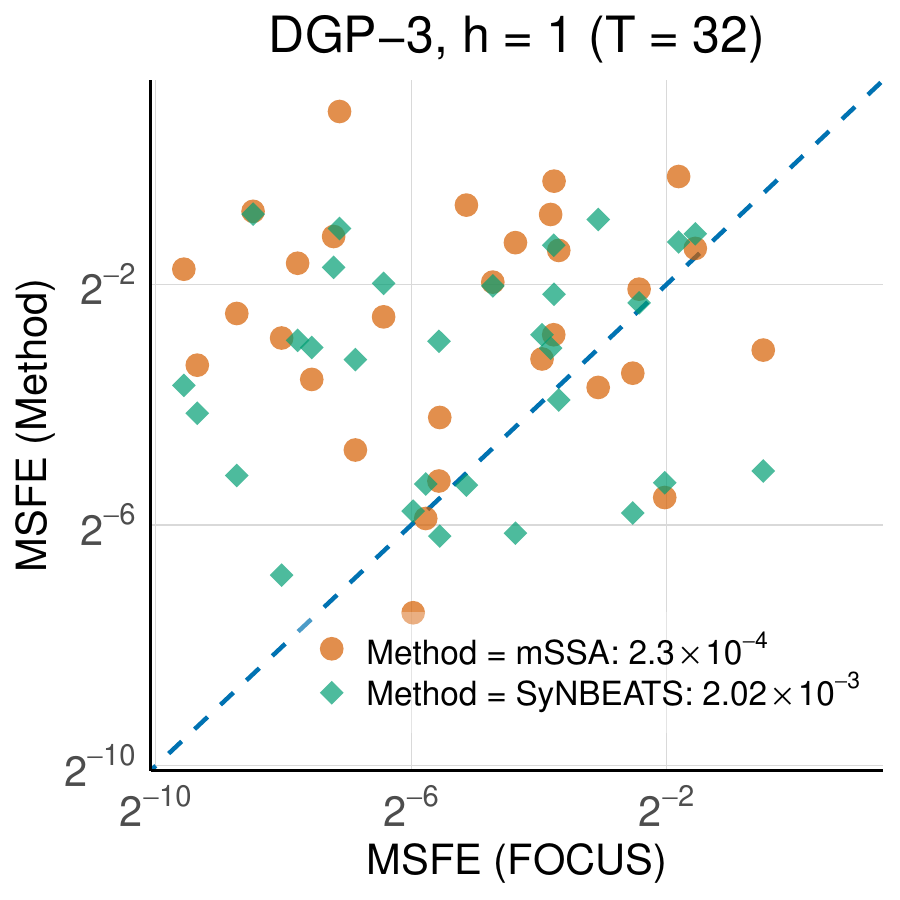}}
    \subfloat[]{\includegraphics[width=0.25\textwidth]{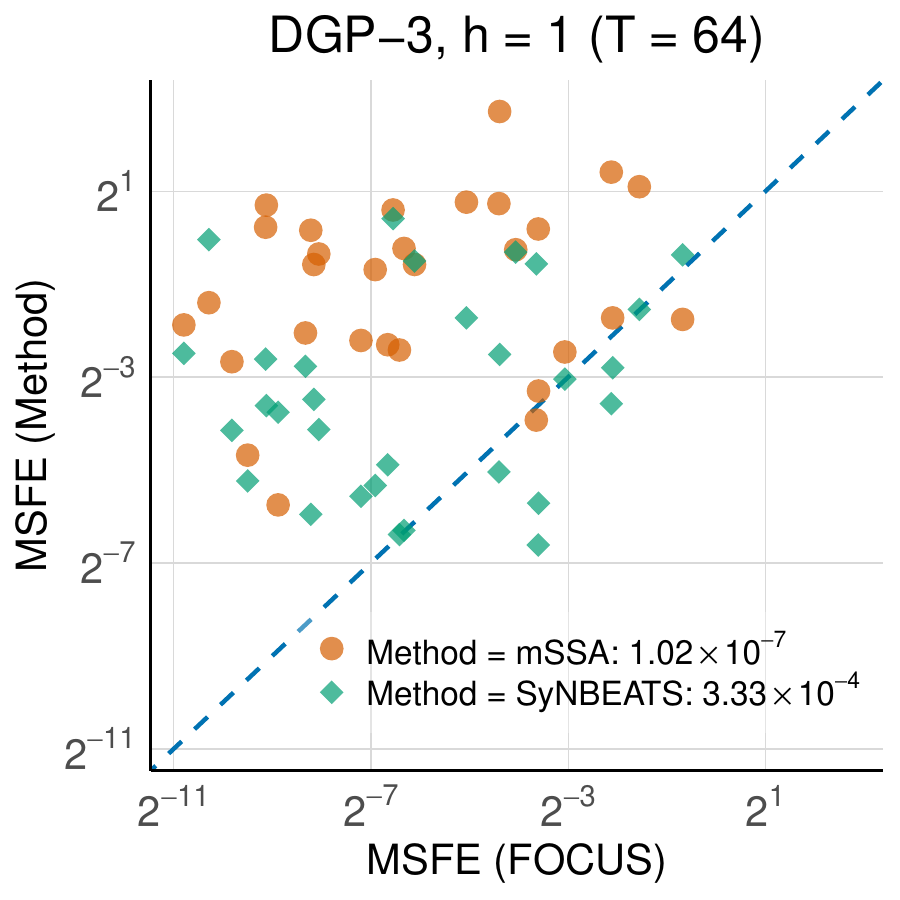}}
    \subfloat[]{\includegraphics[width=0.25\textwidth]{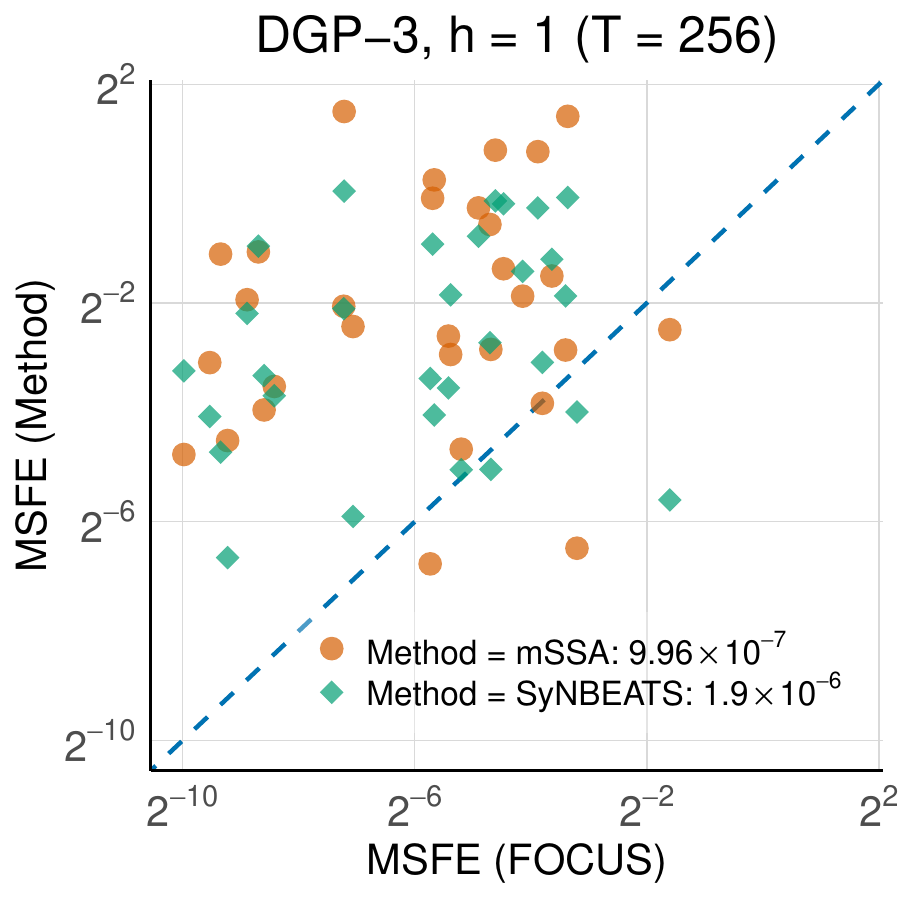}}\\
    \subfloat[]{\includegraphics[width=0.25\textwidth]{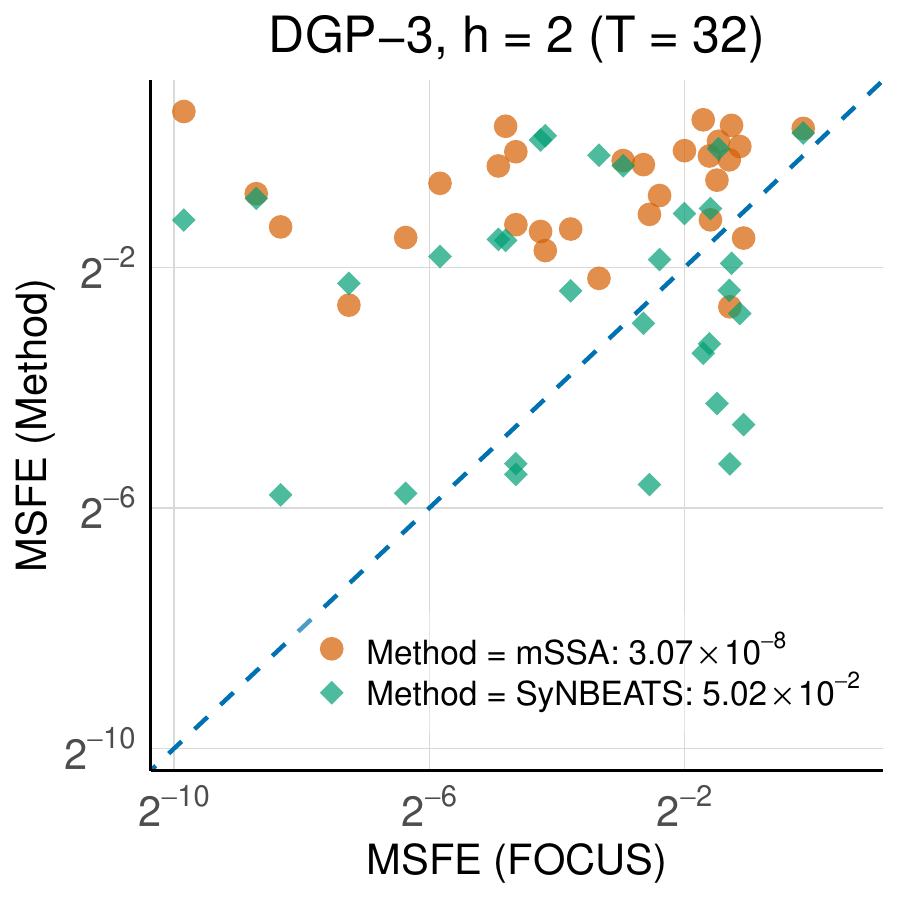}}
    \subfloat[]{\includegraphics[width=0.25\textwidth]{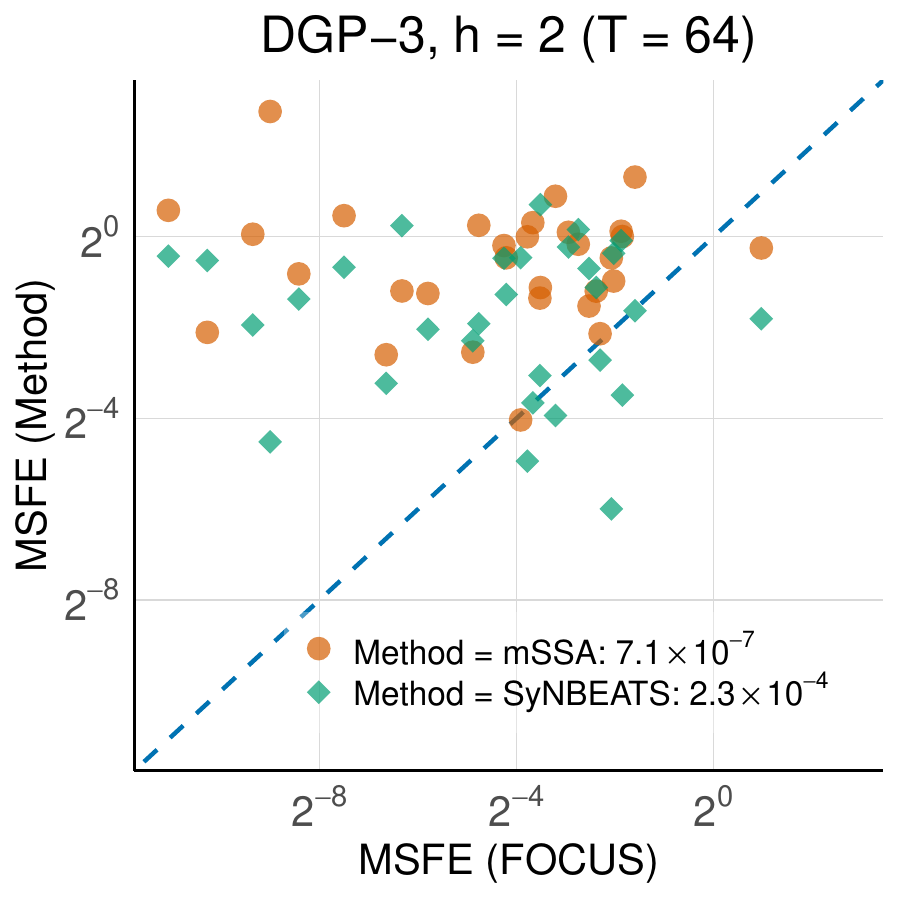}}
    \subfloat[]{\includegraphics[width=0.25\textwidth]{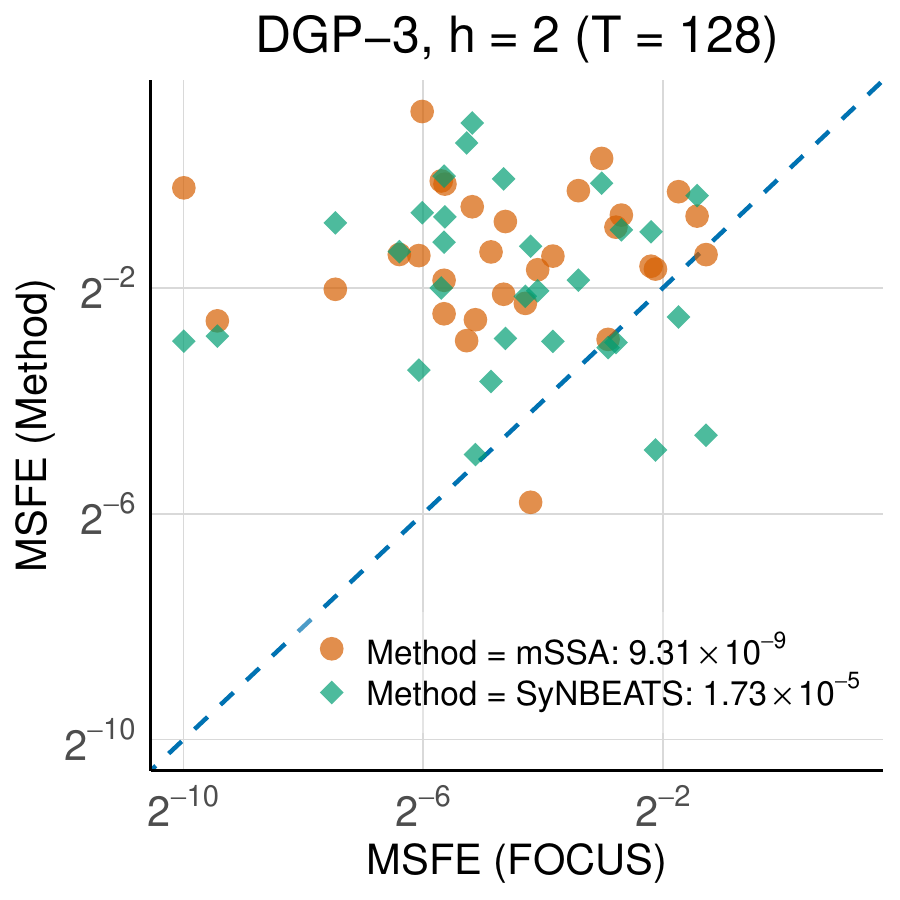}}\subfloat[]{\includegraphics[width=0.25\textwidth]{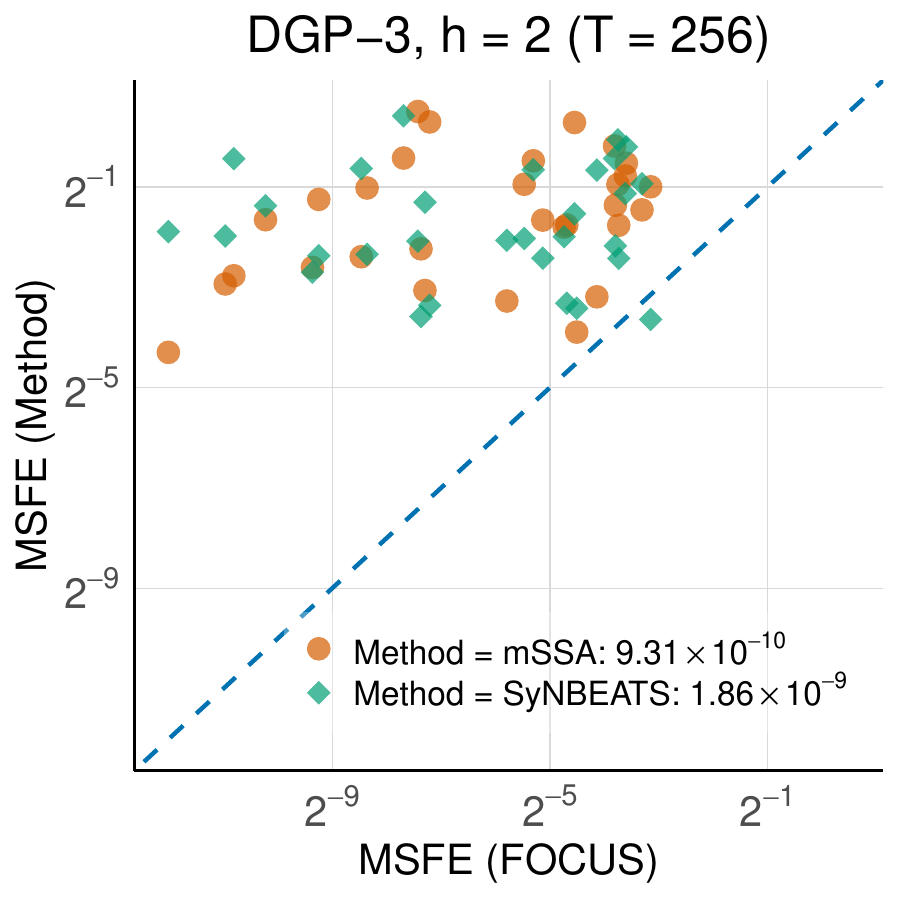}}\\
    \subfloat[]{\includegraphics[width=0.25\textwidth]{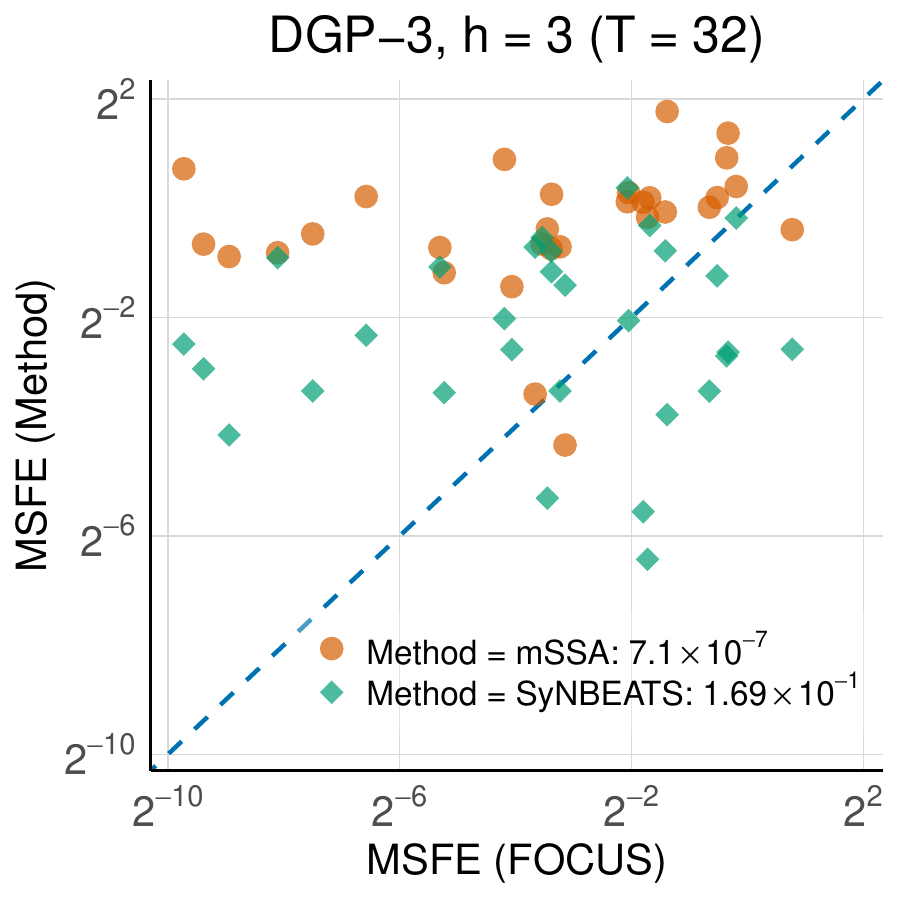}}
    \subfloat[]{\includegraphics[width=0.25\textwidth]{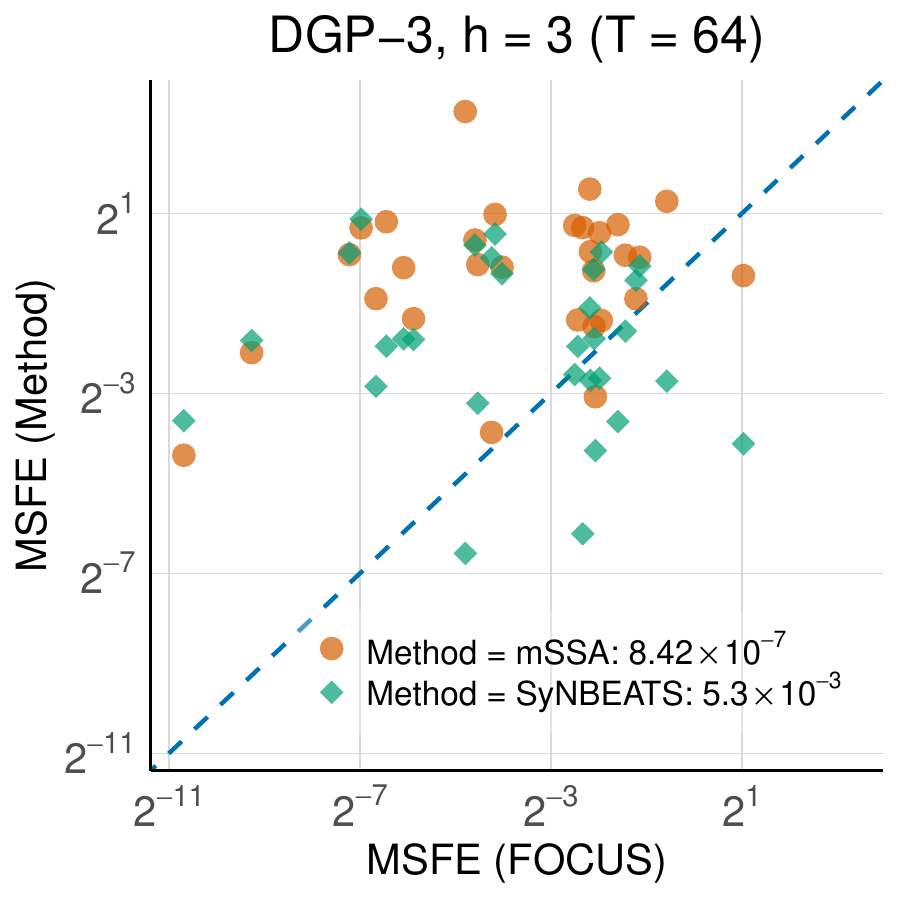}}
    \subfloat[]{\includegraphics[width=0.25\textwidth]{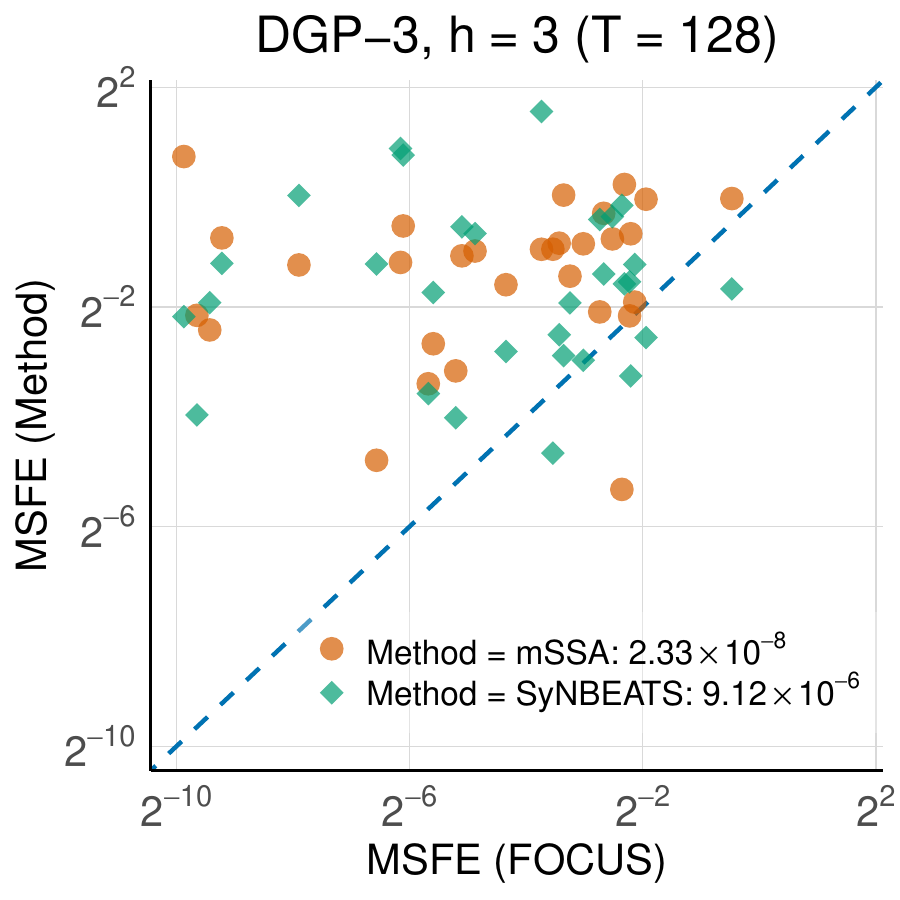}}
    \subfloat[]{\includegraphics[width=0.25\textwidth]{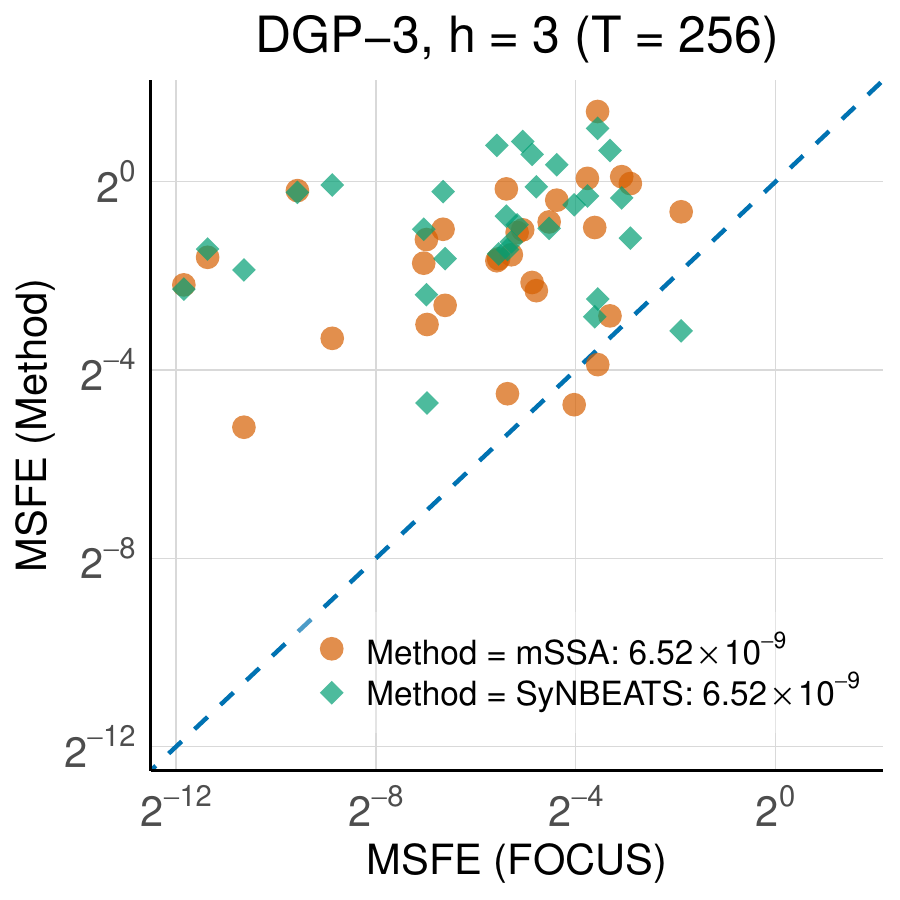}}
    \caption{\textbf{Scatter plots of Mean Squared Forecast error (MSFE) over 30 trials across benchmarks for $N=64$ under DGP-3.} Across all $T$ and $h$, \focus attains significantly lower errors (one-sided Wilcoxon signed-rank test, $p<0.01$; see legends), with scatter points concentrated in the $y>x$ line.}
\end{figure*}

\begin{figure*}[!t]
    \centering
    \subfloat[]{\includegraphics[width=0.25\textwidth]{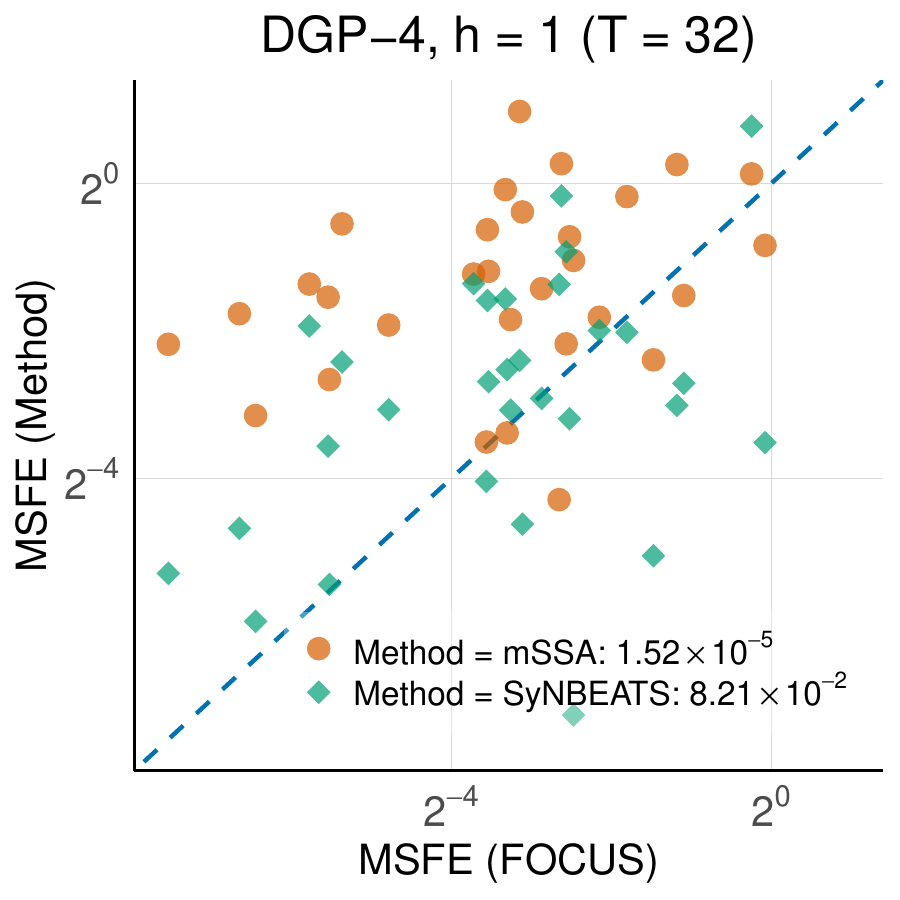}}
    \subfloat[]{\includegraphics[width=0.25\textwidth]{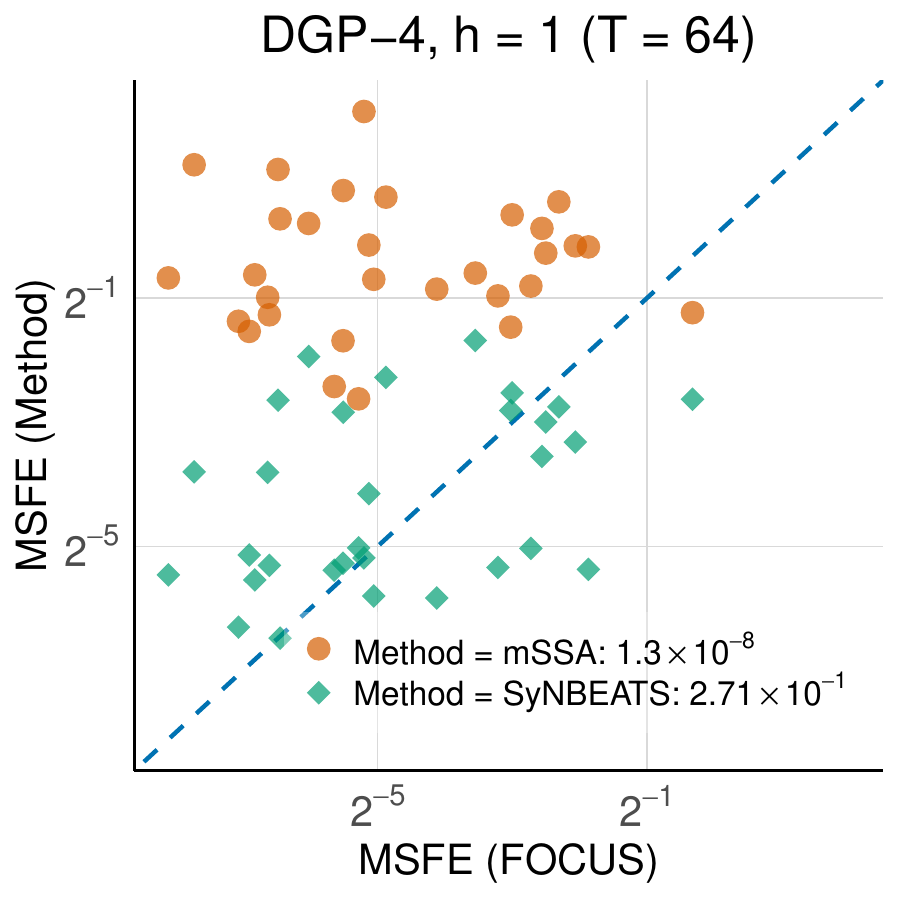}}
    \subfloat[]{\includegraphics[width=0.25\textwidth]{fig/Simulation_results/scatter_plots/scatter_plot_DGP4_T256_h1.pdf}}\\
    \subfloat[]{\includegraphics[width=0.25\textwidth]{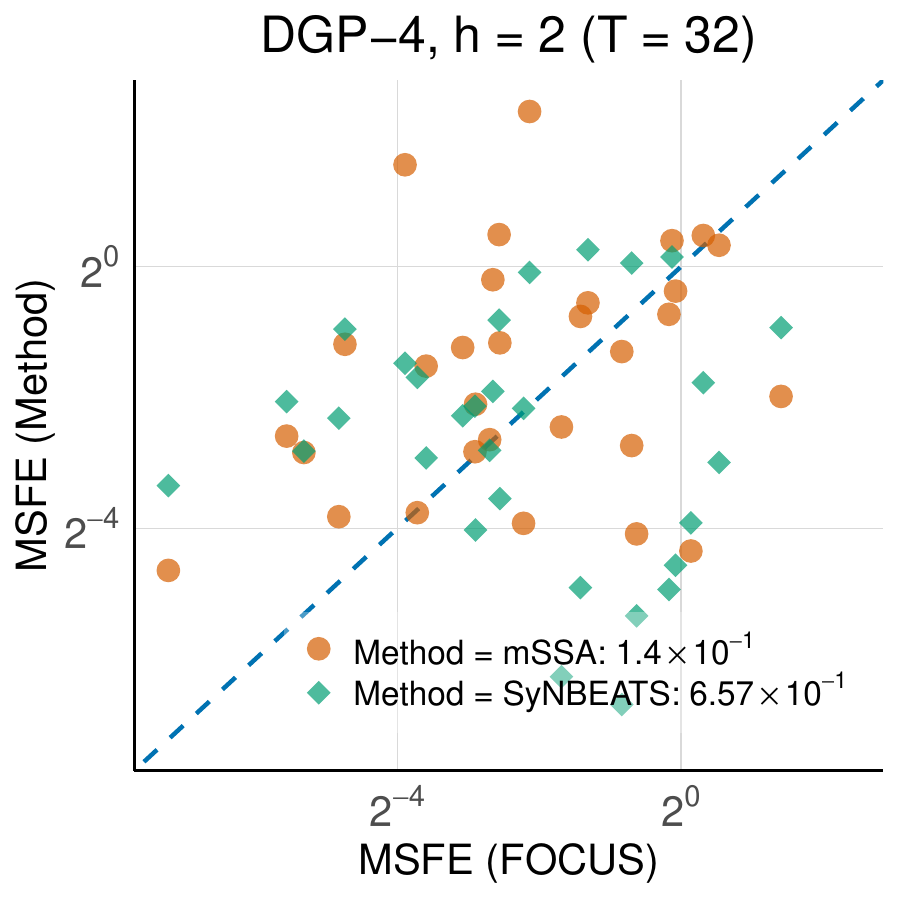}}
    \subfloat[]{\includegraphics[width=0.25\textwidth]{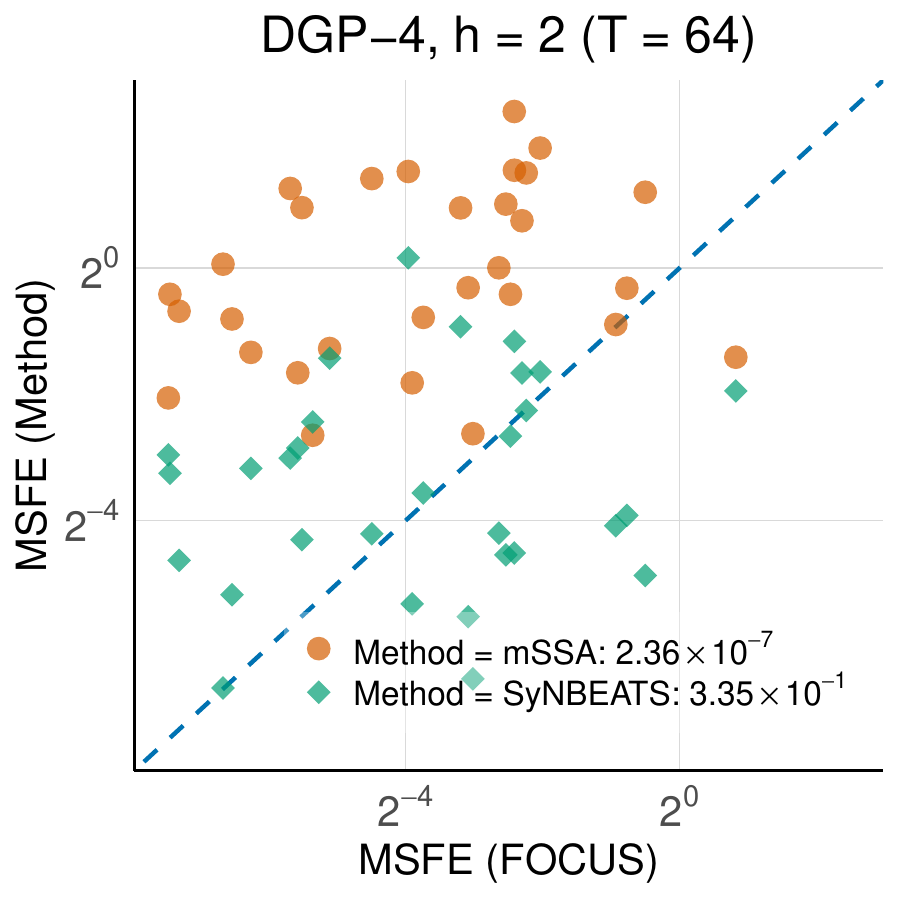}}
    \subfloat[]{\includegraphics[width=0.25\textwidth]{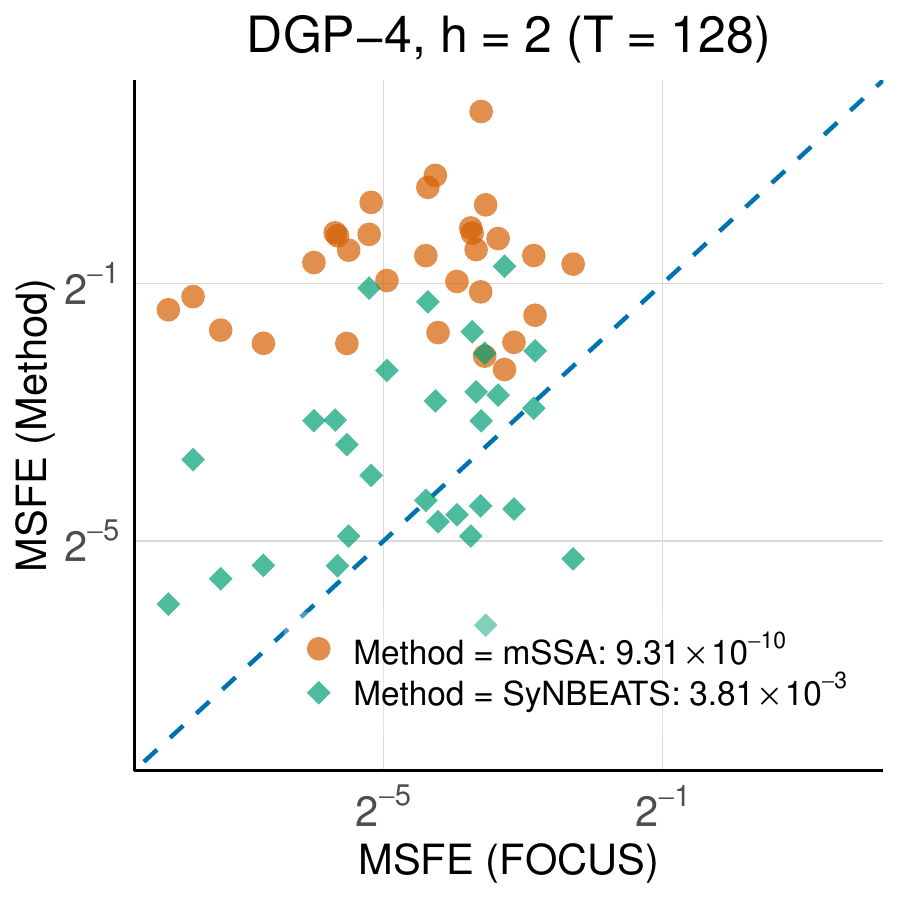}}\subfloat[]{\includegraphics[width=0.25\textwidth]{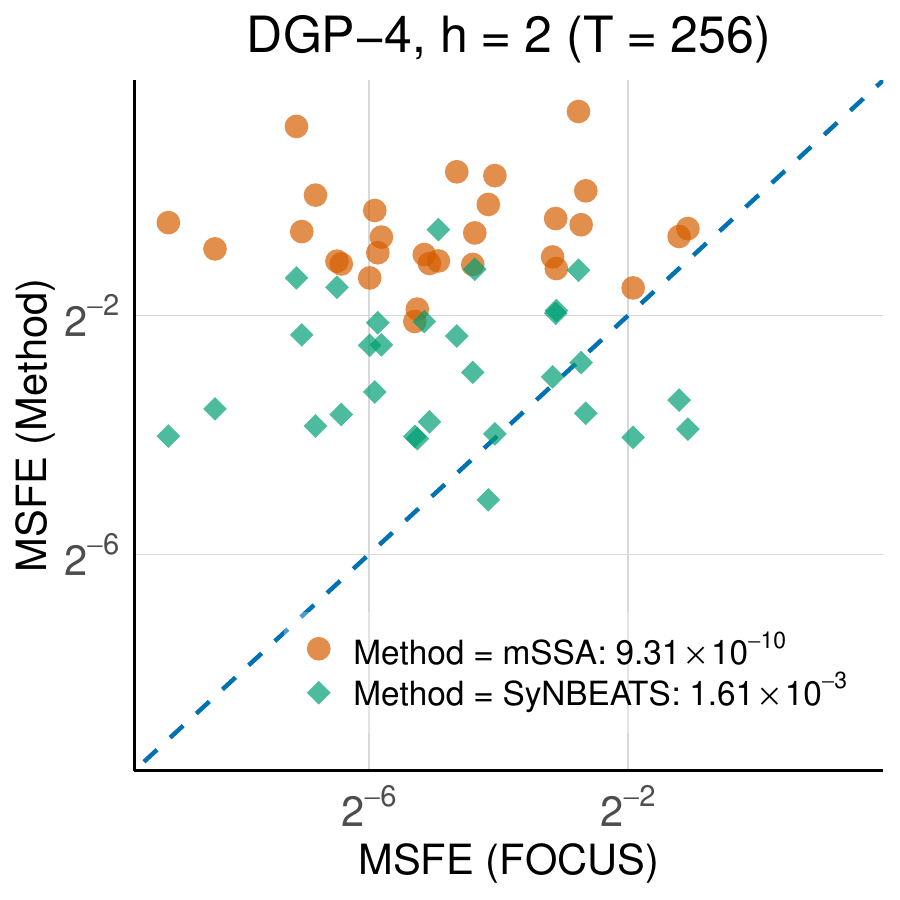}}\\
    \subfloat[]{\includegraphics[width=0.25\textwidth]{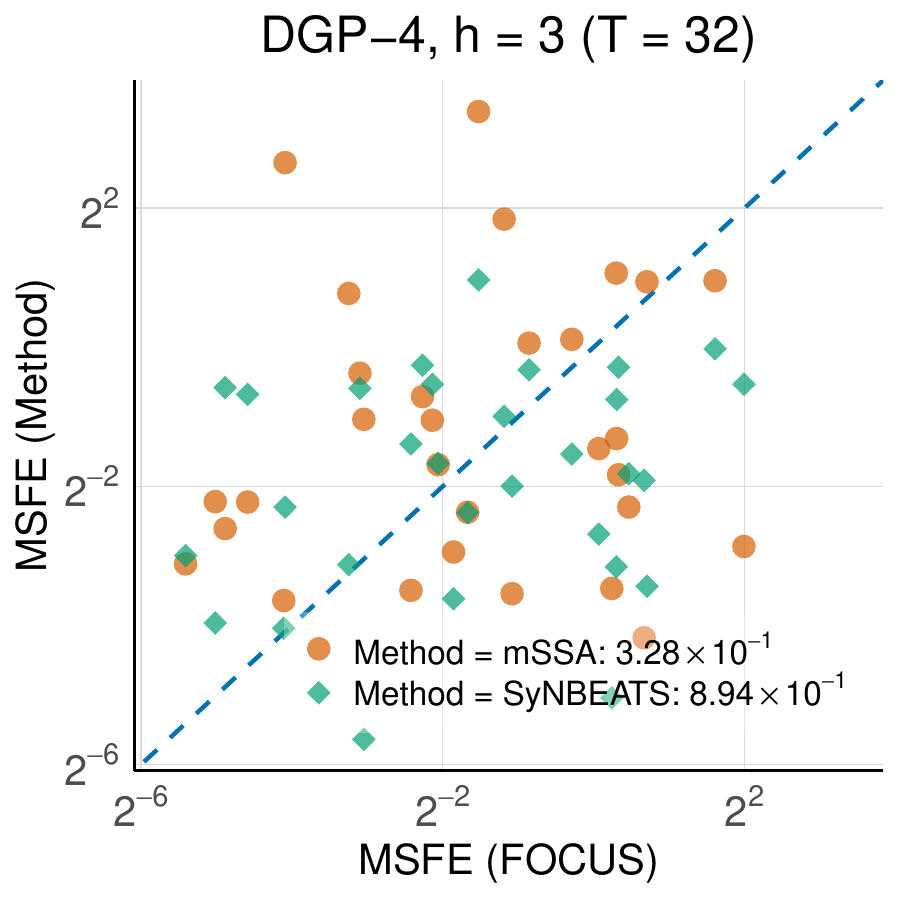}}
    \subfloat[]{\includegraphics[width=0.25\textwidth]{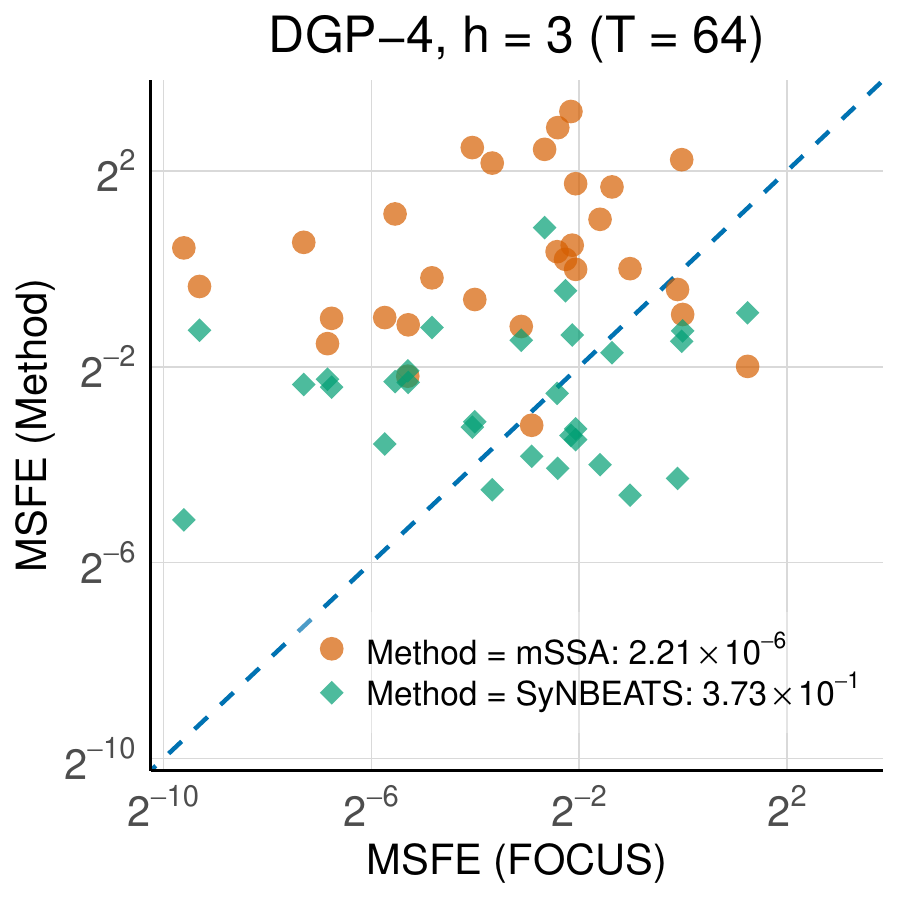}}
    \subfloat[]{\includegraphics[width=0.25\textwidth]{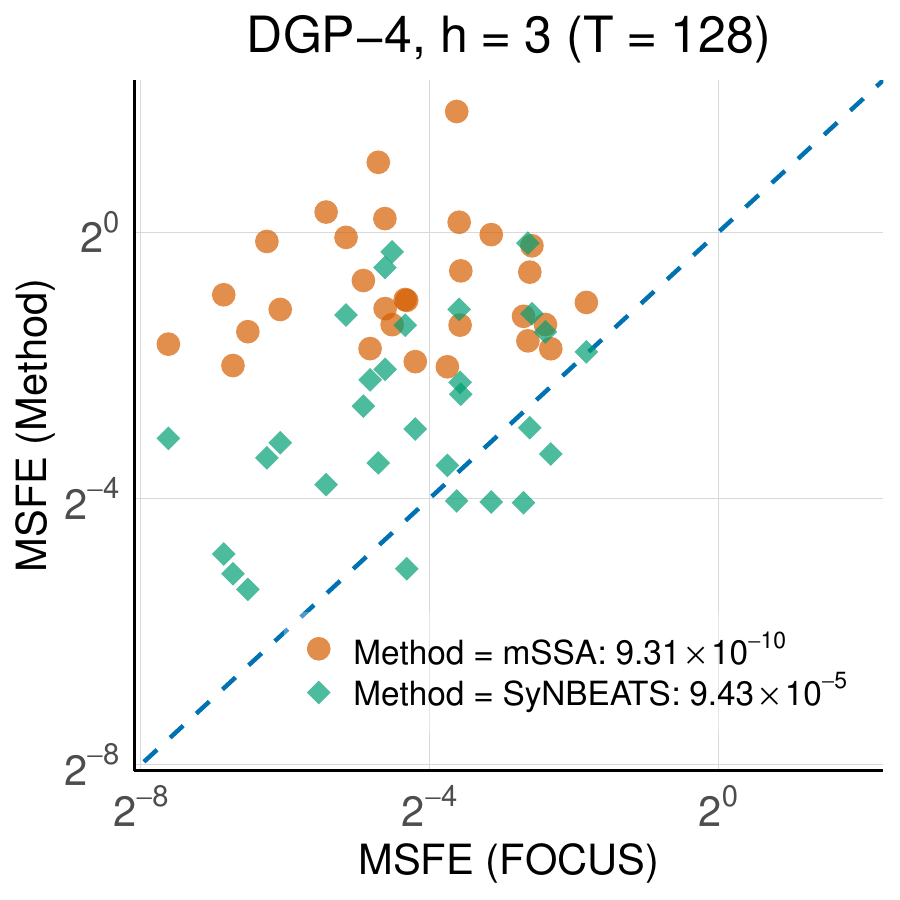}}
    \subfloat[]{\includegraphics[width=0.25\textwidth]{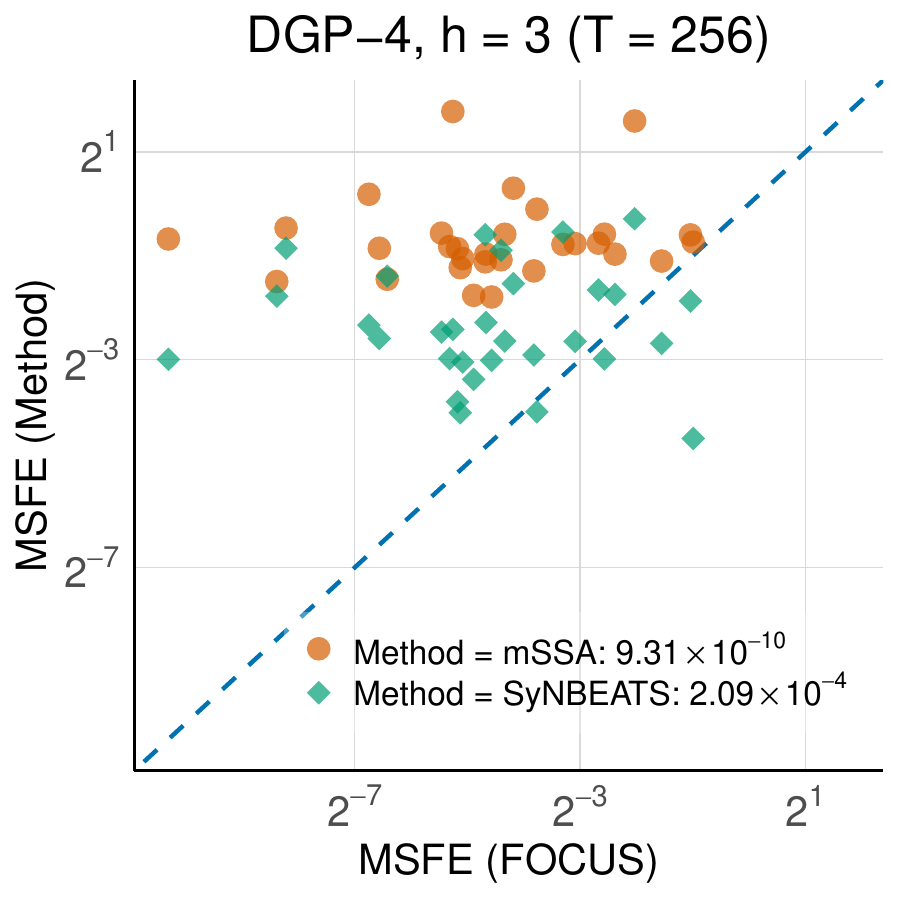}}
    \caption{\textbf{Scatter plots of Mean Squared Forecast Error (MSFE) over 30 trials across benchmarks for $N=64$ under DGP-4.}  Across all $T$ and $h$, \focus attains lower MSFE on average (one-sided Wilcoxon signed-rank test $p$-values in the legends). The comparison with mSSA (orange circles) shows scatter points concentrated in the $y>x$ region, while the comparison with SyNBEATS (green diamonds) exhibits a general shift toward the $y>x$ region with noticeable overlap around the diagonal $y = x$.}
    \label{fig:scatterplot_dgp4}
\end{figure*}

\begin{figure}[!t]
    \centering
    \subfloat{\includegraphics[trim=0 25 0 25, clip, width=0.6\columnwidth]{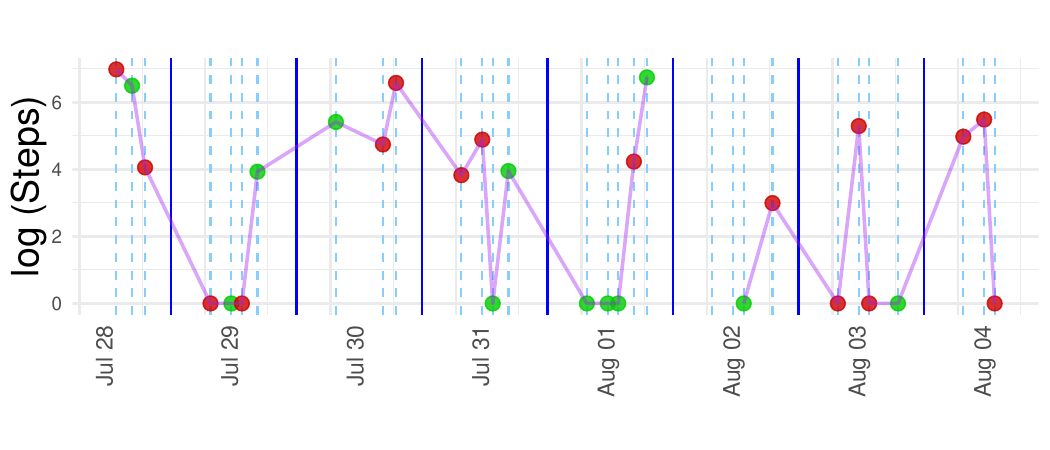}}\\
    \subfloat{\includegraphics[trim=0 25 0 25, clip, width=0.6\columnwidth]{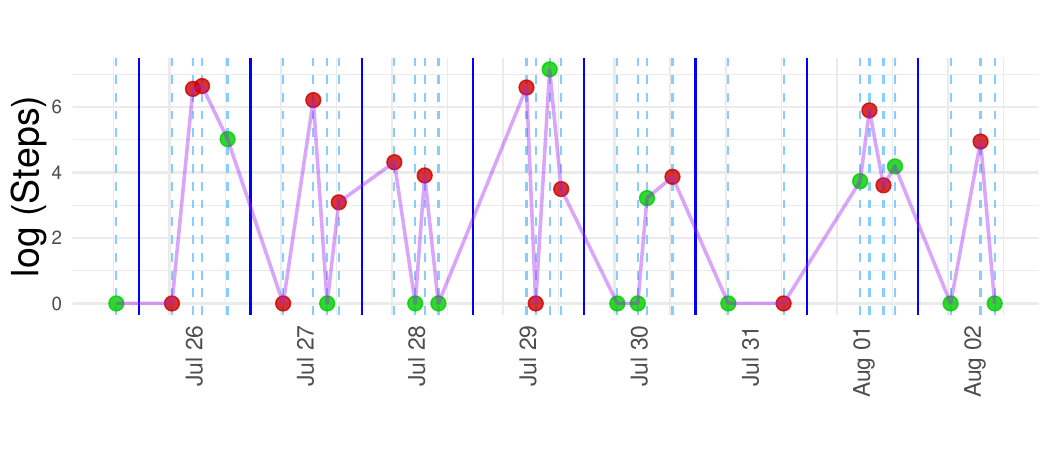}}\\
    \subfloat{\includegraphics[trim=0 25 0 25, clip, width=0.6\columnwidth]{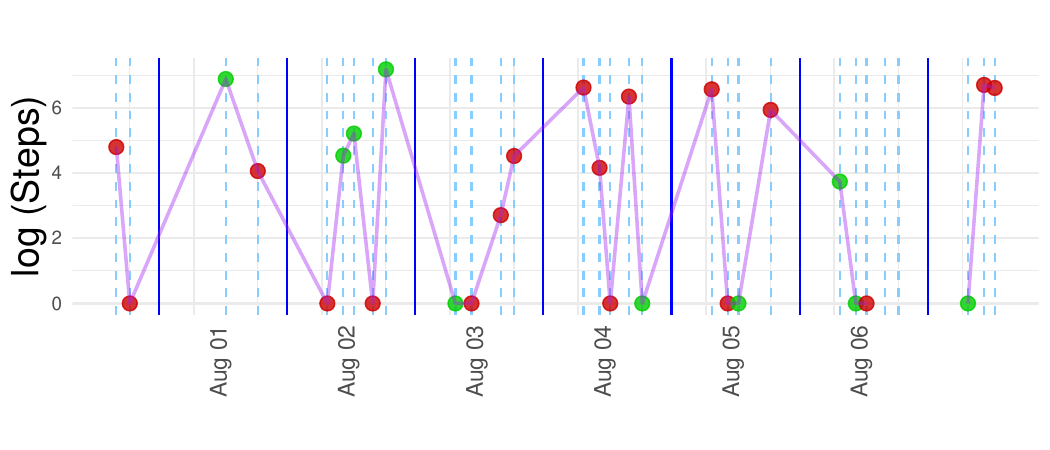}}
    \caption{\textbf{$\log(1 + \texttt{jbsteps30})$ vs time for three users in HeartSteps data.} The 5 decision slots each day are marked by the dashed blue vertical lines. The green (and red) dots represent that the user was nudged (not nudged). Between consecutive slots, the steps exhibit a negative correlation shared across users. }
    \label{fig:jb30_time_plots}
\end{figure}

\begin{figure*}[!h]
    \centering
    \subfloat[]{\includegraphics[width=0.35\textwidth]{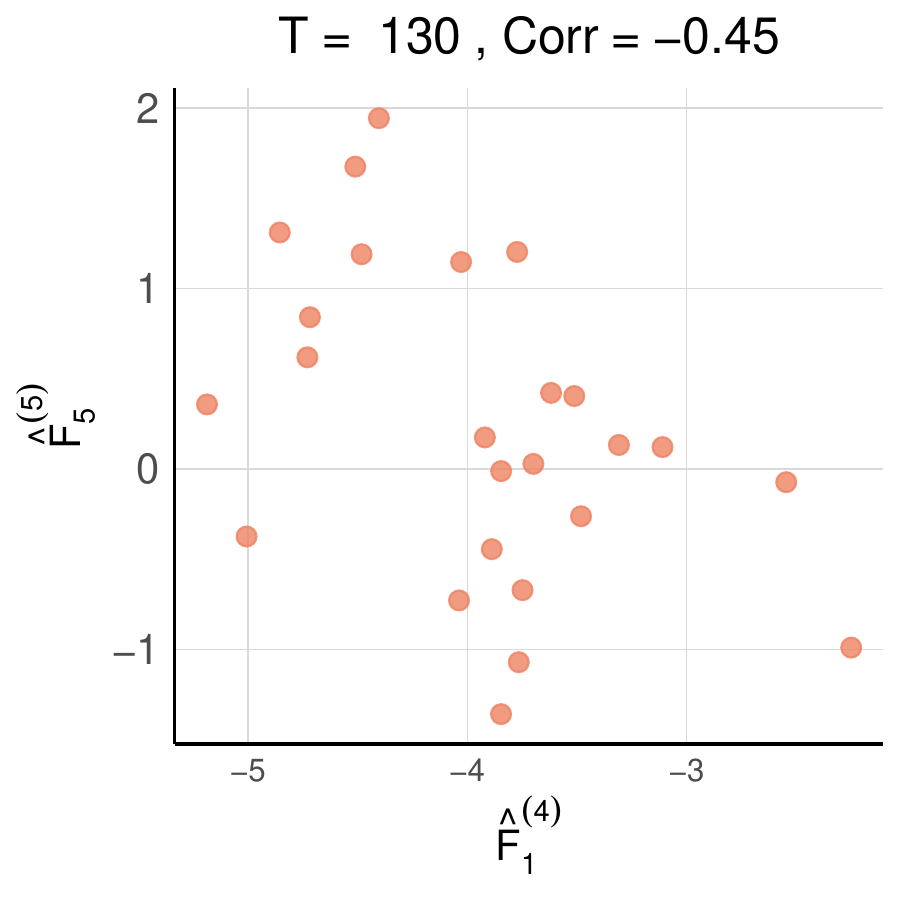}}
    \subfloat[]{\includegraphics[width=0.35\textwidth]{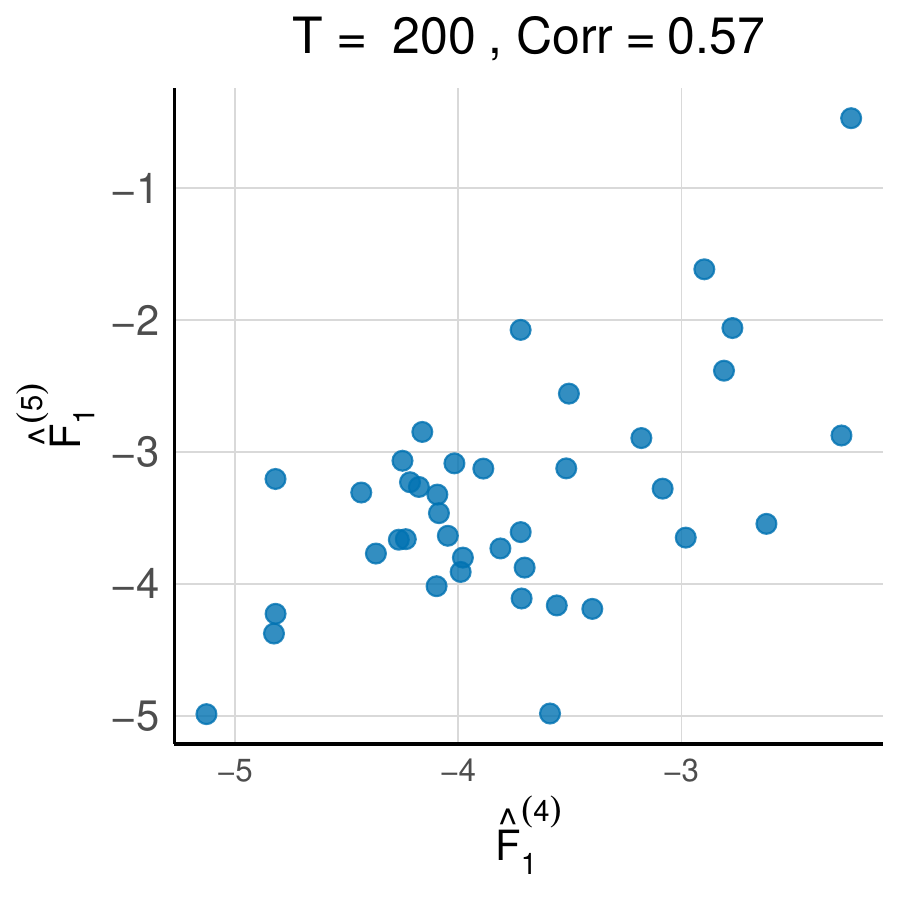}}\\
    \subfloat[]{\includegraphics[width=0.35\textwidth]{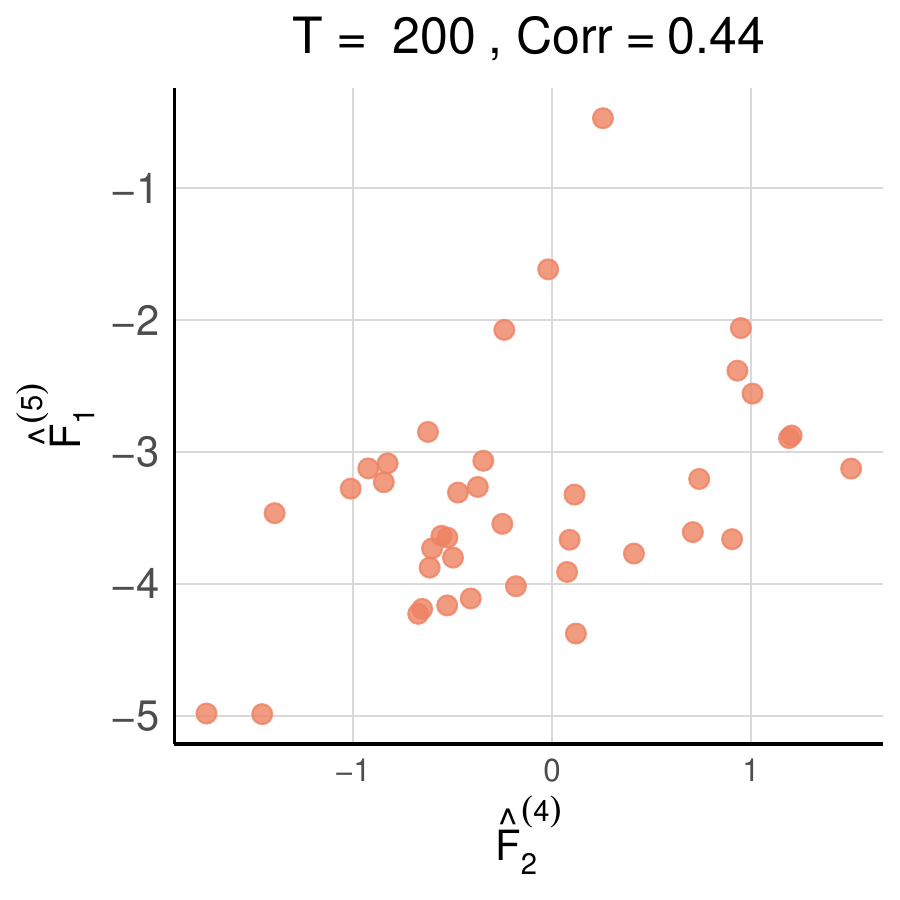}}
    \subfloat[]{\includegraphics[width=0.35\textwidth]{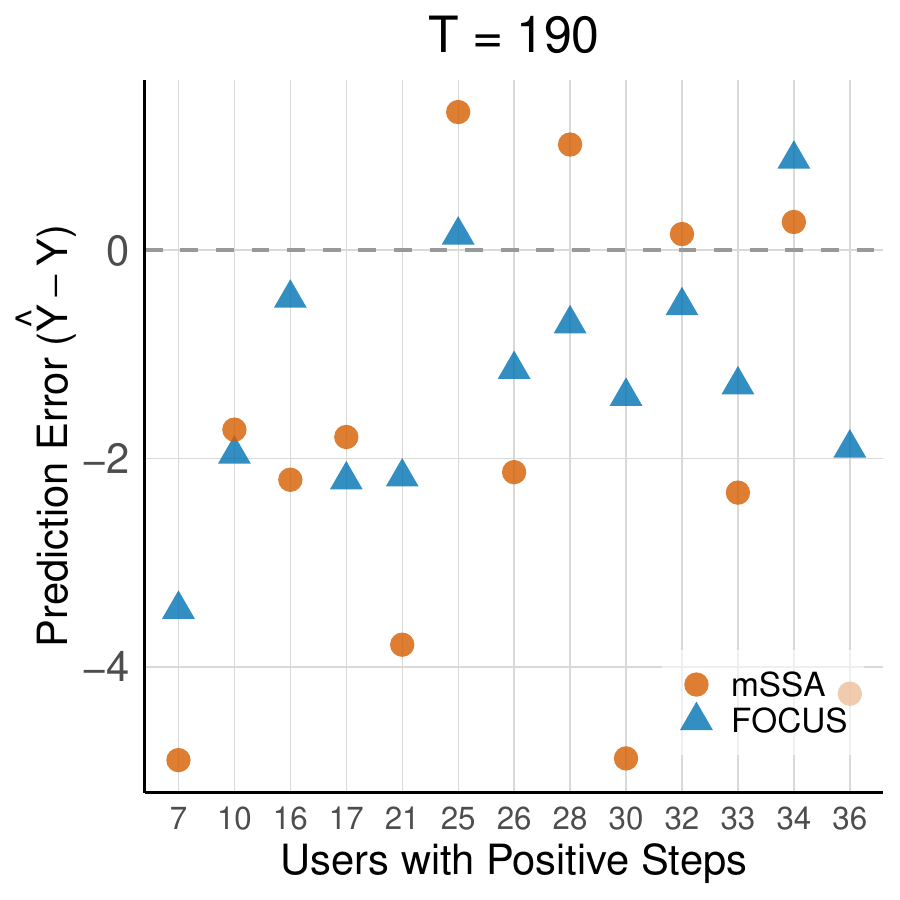}}
    \caption{\textbf{Additional results for \focus and mSSA on HeartSteps data.} Panels (a), (b) and (c) present scatter plots of the estimated factors at slot 5 i.e. $\hat F^{(5)}$ vs the same at slot 4 denoted by $\hat F^{(4)}$. In all three panels, $\hat F_i^{(4)}$ and $\hat F_j^{(5)}$ are strongly correlated for several pairs of $(i,j)$. Panel (d) shows a comparison of the prediction errors between \focus and mSSA at $T = 190$ and $h = 5$. The prediction error is measured with Mean squared prediction error for the users with positive steps at $T+h$. The blue triangles (\focus) are closer to the origin axis than the orange circles (mSSA)-- indicating that \focus exhibits better forecasting performance than mSSA at $T = 190$ and $h = 5$.}
    \label{fig:hs_extra_plot}
\end{figure*}
\end{document}

